\documentclass[12pt]{article}
\usepackage[utf8]{inputenc}
\usepackage[english]{babel}
\usepackage{geometry}
\usepackage[dvipsnames]{xcolor}

\newcommand{\iprod}{\mathbin{\lrcorner}} 
\newcommand{\ri}{\mathrm{i}}
\newcommand{\C}{\mathbb{C}}

\usepackage[sorting=none]{biblatex}
\usepackage{tikz-cd}
\usepackage{physics}
\usepackage{amssymb}
\usepackage{mathtools}
\usepackage{breqn}
\usepackage{amsthm}
\usepackage{amsmath}
\usepackage{graphicx}
\usetikzlibrary{shapes.geometric} 
\usepackage{mathrsfs}
\usepackage{array}	
\usepackage{xcolor}
\graphicspath{ {./Picture/} }
\usepackage{nicematrix}
\usepackage{tikz}
\usepackage{tikz-network}
\usetikzlibrary{fit}
\usepackage{xparse}
\usepackage[export]{adjustbox}
\usepackage[all]{xy}
\newcommand\scalemath[2]{\scalebox{#1}{\mbox{\ensuremath{\displaystyle #2}}}}
\usepackage{caption}
\usepackage[skip=1ex, belowskip=2ex]{subcaption}
\usepackage{float}

\usetikzlibrary{decorations.shapes}
\tikzset{decorate sep/.style 2 args=
{decorate,decoration={shape backgrounds,shape=circle,shape size=#1,shape sep=#2}}}

\usetikzlibrary{matrix,backgrounds,arrows,positioning}
\pgfdeclarelayer{myback}
\pgfsetlayers{myback,background,main}

\NewDocumentCommand{\highlight}{O{blue!40} m m}{%
\draw[mycolor=#1] (#2.north west)rectangle (#3.south east);
}

\NewDocumentCommand{\fhighlight}{O{blue!40} m m}{%
\draw[myfillcolor=#1] (#2.north west)rectangle (#3.south east);
}

\tikzset{mycolor/.style = {line width=1bp,color=#1}}%
\tikzset{myfillcolor/.style = {draw,fill=#1}}%

\usepackage[makeroom]{cancel}

\DeclareMathOperator{\id}{id}

\usepackage{hyperref}

\newtheorem{defn}{Definition}[section]	
\newtheorem{prop}[defn]{Proposition}
\newtheorem{exmp}[defn]{Example}
\newtheorem{thm}[defn]{Theorem}
\newtheorem{cor}[defn]{Corollary}

\newtheorem{lm}[defn]{Lemma}
\newtheorem{rem}[defn]{Remark}

\usepackage[colorinlistoftodos]{todonotes}

\newcommand{\cS}{\mathcal{S}}

\newcommand{\cL}{{\cal L}}

\newcommand{\om}{\omega} 
\newcommand{\rd}{\mathrm{d}}

\newcommand{\Pbkt}[2]{\ensuremath \{#1,#2\}}

\title{Deformed cluster maps of type $A_{2N}$}
\author{Jan E. Grabowski\footnotemark[1] 
\and Andrew N.W. Hone\footnotemark[2] \and Wookyung Kim\footnotemark[1]$\,\,$\footnotemark[3]
}
\date{13th April 2026}

\begin{document}

\maketitle

\begin{abstract}

Cluster algebras are a class of commutative algebras whose generators are constructed by iterating a set of birational transformations, called cluster mutations. Particular compositions of cluster mutations with permutations are  known as cluster maps. The simplest cluster maps 
are obtained from cluster algebras of finite Dynkin type, and exhibit the phenomenon of 
Zamolodchikov periodicity, whereby every orbit is periodic with the same period. 

Here we extend 
recent work of one of us with Kouloukas, by constructing deformations of integrable cluster maps corresponding to the Dynkin types $A_{2N}$, lifting these to higher-dimensional maps possessing the Laurent property, and demonstrating integrability 
of the deformations for all $N\geq 1$. More precisely, in low dimensions, we 
prove directly that the associated maps are integrable in the Liouville sense, while for all deformed $A_{2N}$ maps we verify that 
they satisfy a different integrability criterion, 
namely that 
they have zero algebraic entropy. 
This provides the first infinite class of examples (in arbitrarily high rank) of such deformed maps, and gives further information about the associated discrete integrable systems.  Key to our approach is a ``local expansion'' operation on quivers which allows us to construct and analyse deformed mutations in type $A_{2N}$,  by extending them from those in type $A_{2(N-1)}$.
\end{abstract}

\renewcommand{\thefootnote}{\fnsymbol{footnote}}
\footnotetext[1]{School of Mathematical Sciences, Lancaster University, Lancaster, LA1 4YF, United Kingdom; \\\url{j.grabowski@lancaster.ac.uk}, \url{http://www.maths.lancs.ac.uk/~grabowsj/}}
\footnotetext[2]{School of 
Engineering, Mathematics, 
\& Physics,
University of Kent, Canterbury CT2 7NF, United Kingdom; \url{A.N.W.Hone@kent.ac.uk}}
\footnotetext[3]{ Graduate School of Mathematical Sciences, 
University of Tokyo, 3-8-1 Komaba, Tokyo 153-8914, Japan; 
 \url{kim.wookyung.38u@st.kyoto-u.ac.jp}}
\renewcommand{\thefootnote}{\arabic{footnote}}
\setcounter{footnote}{0}

\tableofcontents

\section{Introduction}
\setcounter{equation}{0}

Cluster algebras are a class of commutative algebras, constructed as subalgebras of rational function fields, which were introduced by Fomin and Zelevinsky in \cite{fomin2001cluster}. 
They are built differently from typical commutative algebras, as cluster algebras are not presented in terms of given generators and relations from the beginning. Instead, one starts with initial data, given by two objects, namely
\begin{itemize}
    \item initial cluster variables: $m$ distinguished generators $\vb{x} = (x_1,\dots, x_m)$; and
    \item an exchange quiver $Q$: a finite directed graph with $m$ nodes which does not contain loops or oriented 2-cycles. 
\end{itemize}
The pair of objects $(\vb{x},Q)$ is called an \textit{initial seed}.
A special iterative process called \textit{mutation} is then applied, 
to produce new cluster variables and exchange quivers. The repeated application of the mutation process results in constructing 
the associated \textit{cluster algebra}, 
as the subalgebra of $\mathbb{Q}(x_{1},\dotsc ,x_{m})$ generated by all the cluster variables.

Fomin and Zelevinsky further extended the notion of mutation to 
encompass the broader setting of a \textit{Y-seed pattern}. This includes additional \textit{coefficient variables}. Similar to cluster mutation, the coefficient variables have their own 
iterative dynamics, described by \textit{coefficient mutation}, but in this article our main focus will be on cluster variables and 
quivers only. 

Along with the Laurent phenomenon, which entails that every cluster variable can be expressed as a Laurent polynomial in the initial (or indeed any) seed, 
the other main result  initially obtained by Fomin and Zelevinsky was the classification of cluster algebras of finite type, i.e. 
having only 
finitely many cluster variables
\cite{Fomin_2003}.  Such cluster algebras are classified by quivers that are orientations of a Dynkin diagram of a semisimple Lie algebra. 

These results raised the importance of 
studying the recurrence relations 
obtained from 
mutations in Y-seeds. 
Each set of  relations is equivalent to a coupled 
difference equation known as a \textit{Y-system}. Y-systems were 
originally discovered by Zamolodchikov in the 
context of exactly solvable models of quantum field theory \cite{Zamolodch2002OnTT}, where 
the solutions of the difference equation 
correspond to the solutions of Bethe ansatz equations for conformal field theories associated with ADE scattering diagrams. 
Furthermore, it was observed that all the solutions appeared to be periodic with a particular period; this was known as \textit{Zamolodchikov's periodicity conjecture}. In \cite{fomin2001ysystems}, this was investigated in the cluster algebra setting, where Fomin and Zelevinsky showed that a specific
composition of mutations exhibited Zamolodchikov periodicity.

A much broader type of periodicity in cluster 
algebras, introduced by Fordy and Marsh \cite{Fordy_2010} 
and generalized
by Nakanishi \cite{nakanishi2011periodicities}, is the 
notion of \textit{cluster mutation periodicity}, 
whereby the composition of a particular sequence of mutations 
is equivalent to a permutation acting on the nodes of the quiver $Q$. 
This leads to the further notion of a \textit{cluster map}, introduced in \cite{2013}, which consists of the composition of the sequence of mutations together with the inverse permutation. By 
taking the natural log-canonical  presymplectic form associated with the quiver $Q$, every cluster 
map is either a symplectic map, or can be reduced to one. 
This means that cluster maps can be considered as 
candidates for discrete integrable systems, 
by applying the analogue of the 
Liouville--Arnold criterion for integrability of symplectic maps, as in \cite{10.3792/pjaa.63.198}. 
Another type of integrability, corresponding to the existence of linear recurrence relations in cluster algebras constructed from certain products of finite and affine Dynkin diagrams, 
was considered in \cite{pavlo}. 

In the context of Zamolodchikov periodicity, the cluster maps that arise are dynamical systems for which all orbits are completely periodic with the same period. The simplest example is the $A_2$ case, for which the initial cluster is $\vb{x} = (x_1,x_2)$ and  the elementary cluster map is given by 
\begin{equation}\label{lyness5}
\varphi: \qquad \left(\begin{array}{c} x_1 \\ x_2 \end{array}\right) 
\mapsto 
\left(\begin{array}{c} x_2 \\ \frac{x_2+1}{x_1} \end{array} \right). 
\end{equation}
It is a symplectic map with respect to the log-canonical 2-form 
$$ 
\om = \frac{1}{x_1x_2}\rd x_1 \wedge \rd x_2,
$$
i.e. $\varphi^* \om =\om$. 
In the dynamical systems community, this map is known as the Lyness 
5-cycle: 
all of the orbits are periodic with period 5 
(which is 2 more than the 
Coxeter number of $A_2$), so $\varphi^5=\mathrm{id}$.  
It is straightforward to construct a conserved quantity (that is, a first integral, or invariant function for the map), by averaging its action on any function over a period: for instance, the quantity 
\begin{equation}\label{hlyness} 
H:= \sum_{j=0}^4 (\varphi^*)^jx_1 =
x_1+x_2 + \frac{x_2+1}{x_1}+\frac{x_1+x_2+1}{x_1x_2}+ 
\frac{x_1+1}{x_2} 
\end{equation}
is a first integral for $\varphi$, that is $\varphi^* H = H$. 
In two dimensions, the existence of one first integral is enough to 
conclude that the symplectic map $\varphi$ is integrable in the Liouville--Arnold sense.

Quite recently, in \cite{hone2021deformations}, Kouloukas and the second author presented a general method of deformation of a
coefficient-free cluster mutation which preserves the covariance of the pre\-symplectic form under mutation. 
In the case of finite Dynkin type algebras exhibiting Zamolodchikov periodicity, they showed that a suitable composition of deformed mutations (depending on 
parameters) produces a 
deformation of the cluster map that is compatible with the pre\-symplectic structure. While this deformation is such that the Laurent property is destroyed, and the periodicity of the cluster map is lost, it turns out that Liouville--Arnold integrability is retained, provided that the parameters are 
constrained appropriately. 

They presented several examples, including integrable 
deformations of cluster maps associated with Dynkin types $A_2$, $A_3$ and $A_4$. Even more recently, 
the geometry of the deformed 
type $A_3$ map was described in more detail, and  integrable 
deformations of types $C_2$, $B_3$ and $D_4$ were also obtained \cite{hkm}. All of these integrable deformations of periodic cluster maps obtained from Dynkin diagrams have the common feature 
that they admit \textit{Laurentification}: they can be lifted to a larger phase space where the Laurent property is restored, and in fact they correspond to cluster maps in the new 
enlarged set of variables, in which this property holds.  

The simplest case of $A_2$ can be used to briefly illustrate what happens in all these examples. The map \eqref{lyness5} admits the 
2-parameter deformation given by 
\begin{equation}\label{genlyness}
\varphi_{(a,b)}: \qquad \left(\begin{array}{c} x_1 \\ x_2 \end{array}\right) 
\longmapsto 
\left(\begin{array}{c} x_2 \\ \frac{ax_2+b}{x_1} \end{array} \right), 
\end{equation} 
where $a,b$ are parameters, which is also referred to as the (general) Lyness map. This map preserves the same symplectic form 
as \eqref{lyness5}, so $\varphi_{(a,b)}^*\,\omega=\omega$, and in this case no constraints on the parameters are needed for integrability of the deformed map, 
since for all $a,b$ 
the map \eqref{genlyness} has the first integral 
\begin{equation}\label{ghlyness} 
H_{Lyness} =
x_1+x_2 + a\left(\frac{x_2+a}{x_1}\right) +
b\left(\frac{x_1+x_2+a}{x_1x_2}\right)+ 
a\left(\frac{x_1+a}{x_2}\right).  
\end{equation}
This is clearly a 
2-parameter generalization 
of the original conserved quantity 
$H$ in \eqref{hlyness}, which is recovered when $a=b=1$. Hence Liouville--Arnold integrability 
holds for the general Lyness map 
\eqref{genlyness}, for all  $a,b$. 

\subsection{Outline of the paper} 

In this paper, we consider the deformation of an integrable cluster map corresponding to the 
even-dimensional case of Dynkin type $A_{2N}$, 
for all values of $N\geq 1$. Section 2 consists of introductory material, illustrated with simple examples. 
In subsection~\ref{ss:cluster-algs}, we give some brief background on cluster algebras, mutation periodicity and cluster maps.  We follow this in subsection~\ref{ss:discrete-int-sys} by recalling the 
appropriate 
notions  required to
define 
discrete integrable systems in this context. Next, in subsection~\ref{ss:sing-analysis}, we give a short introduction to 
a heuristic method for identifying integrable difference equations  known as \emph{singularity confinement analysis} \cite{Grammaticos1991DoIM},
with a subsequent exposition 
in subsection~\ref{s:algentrop} 
of the concept of 
\emph{algebraic entropy}  for rational maps 
\cite{Bellon_1999}, and its connection with degree growth of cluster variables.  Subsection~\ref{ss:deformations} introduces the details of the deformation approach of the second author and Kouloukas, and in subsection~\ref{ss:Laurentification} we present two examples of how singularity 
confinement analysis may be used to perform \emph{Laurentification} of a deformed map, thereby 
lifting it to an undeformed cluster map on a higher-dimensional space.

In Section~\ref{s:type-A2N}, we begin our analysis of 
Zamolodchikov periodic cluster maps in Dynkin type $A_{2N}$, which are the starting point 
for their subsequent  deformations.  After presenting the associated Poisson structure, obtained from the symplectic form, in the undeformed case we construct a set of first integrals and use them to show that the periodic cluster map of type $A_{2N}$ is integrable in the Liouville--Arnold sense (Theorem~\ref{liouvilleA2N}).

In Section~\ref{A2Ndef} 
we present the first new example of an integrable deformation in higher rank, namely an integrable  deformation of type $A_{6}$, that is, a 6-dimensional Liouville integrable map which depends on three arbitrary parameters (Theorem~\ref{A6thm}).  For a 2-parameter subset 
of the latter deformation, we are able to obtain a Laurentification which lifts it to a cluster map in 15 dimensions, with the parameters corresponding to two additional frozen variables (Theorem~\ref{thmforA6}). 
This Laurentified map, and the associated quiver, forms the 
basis for our subsequent inductive approach to obtaining deformed cluster maps of type $A$ for all even rank cases. 
More precisely, by making  
a comparison with the integrable deformed map of type $A_{4}$ 
obtained in \cite{hone2021deformations}, we see that the Laurentification of the $A_6$ deformation is obtained from 
Laurentification of the deformed $A_4$ map, which is a cluster map in 11 dimensions, via a procedure that we call 
``local expansion'', which involves the insertion of a particular 4-node quiver. The quivers $Q_{A_4}$ and $Q_{A_6}$ 
associated with the Laurentifications of the 
deformed $A_4$ and $A_6$ cluster maps, respectively, are shown in Figure~\ref{fig:localexp-intro} below, while the relevant subquivers involved in the local expansion procedure are 
highlighted in Figure~\ref{fig:Q4toQ6-intro} (repeated later as Figures~\ref{fig:Localexpansion} and \ref{fig:Q4toQ6}).
\begin{figure}[!ht]
    \centering
\resizebox{1 \textwidth}{!}{%
 \begin{tikzpicture}[every circle node/.style={draw,scale=0.6,thick},node distance=15mm]

  \node [draw,circle,fill=blue!50,"$12$"] (12) at (0,0) {};
  
     \node [draw,circle,fill=red!50,"$6$"] (6) [right= of 12] {};
      \node [draw,circle,fill=red!50,"$7$"] (7) [right=of 6] {};
      \node [draw,circle,fill=red!50,"$8$"] (8) [right=of 7] {};
      \node [draw,circle,fill=red!50,"$9$"] (9) [right=of 8] {};
      \node [draw,circle,fill=red!50,"$10$"right] (10) [below right=of 9] {};
      \node [draw,circle,fill=red!50,"$5$" left] (5) [below left=of 6] {};

      \node [draw,circle,fill=blue!50,"$13$"] (13) [right=of 9] {};

       \node [draw,circle,fill=red!50,"$4$"below] (4) [below right=of 5] {};
       \node [draw,circle,fill=red!50,"$11$"below] (11) [right=of 4] {};
       \node [draw,circle,fill=red!50,"$1$"below] (1) [right=of 11] {};
       \node [draw,circle,fill=red!50,"$2$"below] (2) [right=of 1] {};
       
        \node [draw,circle,fill=red!50,"$3$" below] (3) at (4.4,-3.5) {};
       
       \node (a) at (4.4,-5.5) {\Large (a) $Q_{A_{4}}$};

  \begin{scope}[>=Latex]
  
  \draw[-> , thick]  (1) edge (2); 
  \draw[-> , thick]  (1) edge (7); 
  \draw[-> , thick]  (9) edge (1);
  \draw[-> , thick]  (11) edge (1);
   \draw[-> , thick]  (1) edge (10);

 \draw[-> , thick]  (2) edge (3);
 \draw[-> , thick]  (2) edge (8);
 \draw[-> , thick]  (2) edge (13);

 \draw[-> , thick]  (4) edge (11);
  \draw[-> , thick]  (12) edge (4);
   \draw[-> , thick]  (7) edge (4);
    \draw[-> , thick]  (3) edge (4);
    
 \draw[-> , thick]  (5) edge (12);
  \draw[-> , thick]  (5) edge (11);
   \draw[-> , thick]  (3) edge (5);
    \draw[-> , thick]  (7) edge (5);
   
   \draw[-> , thick]  (6) edge (12);
    \draw[-> , thick]  (6) edge (7);
     \draw[-> , thick]  (11) edge (6);
      \draw[-> , thick]  (6) edge[bend left= 15] (3);
      
    \draw[-> , thick]  (7) edge (8);
    
     \draw[-> , thick]  (8) edge (11);
      \draw[-> , thick]  (8) edge (9);
       \draw[-> , thick]  (10) edge (8);
       
        \draw[-> , thick]  (13) edge (9);
         \draw[-> , thick]  (3) edge[bend left=15] (9);
         
          \draw[-> , thick]  (13) edge (10);
           \draw[-> , thick]  (10) edge (3);
           
            \draw[-> , thick]  (3) edge (13);
            
             \draw[-> , thick]  (12) edge (3);

    \end{scope}

\draw [-{Latex[length=3mm]}] (10,-1.2) -- (11,-1.2) node[midway,sloped,above] {Expansion};

  \node [draw,circle,fill=blue!50,"$12$"] (12) at (12,0) {};
  
     \node [draw,circle,fill=red!50,"$7$"] (6) [right= of 12] {};
      \node [draw,circle,fill=red!50,"$8$"] (7) [right=of 6] {};
       \node [draw,circle,fill=green!50,"$9$"] (a1) [right=of 7] {};
        \node [draw,circle,fill=green!50,"$10$"] (a2) [right=of a1] {};
      
      \node [draw,circle,fill=red!50,"$11$"] (8) [right=of a2] {};
      \node [draw,circle,fill=red!50,"$12$"] (9) [right=of 8] {};
      \node [draw,circle,fill=red!50,"$13$"right] (10) [below right=of 9] {};
      \node [draw,circle,fill=red!50,"$6$" left] (5) [below left=of 6] {};

      \node [draw,circle,fill=blue!50,"$17$"] (13) [right=of 9] {};

       \node [draw,circle,fill=red!50,"$5$"below] (4) [below right=of 5] {};
       \node [draw,circle,fill=red!50,"$14$"below] (11) [right=of 4] {};
        \node [draw,circle,fill=green!50,"$15$"below] (b1) [right=of 11] {};
         \node [draw,circle,fill=green!50,"$1$"below] (b2) [right=of b1] {};
       
       \node [draw,circle,fill=red!50,"$2$"below] (1) [right=of b2] {};
       \node [draw,circle,fill=red!50,"$3$"below] (2) [right=of 1] {};
       
        \node [draw,circle,fill=red!50,"$4$" below] (3) at (18,-3.5) {};
       
       \node (a) at (18,-5.5) {\Large (b) $Q_{A_{6}}$};

  \begin{scope}[>=Latex]
  
  \draw[-> , thick]  (1) edge (2); 
 
  \draw[-> , thick]  (9) edge (1);

   \draw[-> , thick]  (1) edge (10);

 \draw[-> , thick]  (2) edge (3);
 \draw[-> , thick]  (2) edge (8);
 \draw[-> , thick]  (2) edge (13);

 \draw[-> , thick]  (4) edge (11);
  \draw[-> , thick]  (12) edge (4);
   \draw[-> , thick]  (7) edge (4);
    \draw[-> , thick]  (3) edge (4);
    
 \draw[-> , thick]  (5) edge (12);
  \draw[-> , thick]  (5) edge (11);
   \draw[-> , thick]  (3) edge[bend left=35] (5);
    \draw[-> , thick]  (7) edge (5);
   
   \draw[-> , thick]  (6) edge (12);
    \draw[-> , thick]  (6) edge (7);
     \draw[-> , thick]  (11) edge (6);
      \draw[-> , thick]  (6) edge[bend left= 25] (3);

      \draw[-> , thick]  (8) edge (9);
       \draw[-> , thick]  (10) edge (8);
       
        \draw[-> , thick]  (13) edge (9);
         \draw[-> , thick]  (3) edge[bend left=25] (9);
         
          \draw[-> , thick]  (13) edge (10);
           \draw[-> , thick]  (10) edge[bend left= 35] (3);
           
            \draw[-> , thick]  (3) edge (13);
            
             \draw[-> , thick]  (12) edge (3);
        

 \draw[-> , thick]  (7) edge (a1);
  \draw[-> , thick]  (a1) edge (a2);
   \draw[-> , thick]  (a2) edge (8);
    \draw[-> , thick]  (11) edge (b1);
     \draw[-> , thick]  (b1) edge (b2);
      \draw[-> , thick]  (b2) edge (1);
      
      \draw[-> , thick]  (b1) edge (7);
     \draw[-> , thick]  (b2) edge (a1);
      \draw[-> , thick]  (1) edge (a2);

      \draw[-> , thick]  (a1) edge (11);
	\draw[-> , thick]  (a2) edge (b1);  
	\draw[-> , thick]  (8) edge (b2);

    \end{scope}
    \end{tikzpicture}
    }
    \caption{Extension from $Q_{A_{4}}$ to $Q_{A_{6}}$}
    \label{fig:localexp-intro}
\end{figure}

\begin{figure}[!ht]
    \centering

\resizebox{0.7 \textwidth}{!}{%
 \begin{tikzpicture}[every circle node/.style={draw,scale=0.6,thick},node distance=15mm]
    
    \node [draw,circle,fill=red!50,"$7$"] (5) at (6,-7) {};
   \node [draw,circle,fill=red!50,"$8$"] (6)[right=of 5]{};
  
  \node [draw,circle,fill=red!50,"$11$"below] (4) at (6,-9.2) {};
   \node [draw,circle,fill=red!50,"$1$"below] (7) [right=of 4] {};
    
  \begin{scope}[>=Latex]
            
       \draw[-> , thick]  (5) edge (6);
        \draw[-> , thick]  (6) edge (4);
         \draw[-> , thick]  (4) edge (7);
          \draw[-> , thick]  (7) edge (5);

    \end{scope}
\draw [-{Latex[length=3mm]}] (10,-8.2) -- (11,-8.2) node[midway,sloped,above] {Expansion };
    \node [draw,circle,fill=red!50,"$8$"] (6) at (13,-7) {};
  \node [draw,circle,fill=red!50,"$14$" below] (4) at (13,-9.2) {};
   \node [draw,circle,fill=green!50,"$9$"] (7) [right=of 6] {};
  \node [draw,circle,fill=green!50,"$15$"below] (11) [right=of 4 ]  {};
   \node [draw,circle,fill=green!50,"$10$"] (8) [right=of 7] {};
  \node [draw,circle,fill=green!50,"$1$" below] (1) [right=of 11 ]  {};
    \node [draw,circle,fill=red!50,"$11$"] (9) [right=of 8] {};
  \node [draw,circle,fill=red!50,"$2$" below] (2) [right=of 1 ]  {};

  \begin{scope}[>=Latex]
   
    \draw[-> , thick]  (7) edge (8);
     \draw[-> , thick]  (8) edge (11);
      \draw[-> , thick]  (11) edge (1);
       \draw[-> , thick]  (1) edge (7);
       
        \draw[-> , thick]  (8) edge (9);
         \draw[-> , thick]  (9) edge (1);
          \draw[-> , thick]  (1) edge (2);
           \draw[-> , thick]  (2) edge (8);
       
        \draw[-> , thick]  (6) edge (7);
         \draw[-> , thick]  (7) edge (4);
          \draw[-> , thick]  (4) edge (11);
           \draw[-> , thick]  (11) edge (6);

    \end{scope}

\end{tikzpicture}
}
    \caption{Local expansion of the subquiver in $Q_{A_{4}}$}
    \label{fig:Q4toQ6-intro}
\end{figure}

By rewriting the action of local expansion in terms of exchange matrices, we construct an associated family of quivers by repeated iteration of this procedure. This produces a sequence of quivers with $4N+5$ nodes, two of which are frozen. 
The particular structure of the local expansion allows us to show that this does indeed result in the Laurentification of a  type $A_{2N}$ deformed cluster map depending on two arbitrary parameters (subsection~\ref{ss:type-A2N-Laurent}). This is the main result of Section 3. 

Section \ref{s:degreeA2N} addresses the question of integrability of the deformed $A_{2N}$ cluster maps constructed in the previous section. For want of a general procedure to obtain a sufficient number of Poisson-commuting first integrals when $N>3$, we address this question by considering a different integrability criterion, namely algebraic entropy. By using tropicalization and explicit recurrence formulae for denominator vectors, we are able to prove that the Laurentification of each deformed $A_{2N}$ cluster map has quadratic growth, and thereby show that the algebraic entropy is zero.  

\subsubsection*{Acknowledgements} 

ANWH was supported by grant IEC$\backslash$R3$\backslash$193024 from the Royal Society. 
WK acknowledges studentship funding from the EPSRC. All three authors are grateful to Lancaster University for financial support during the project. 

WK is supported by Grant-in-Aid for Scientific Research of Japan Society for the Promotion of Science, JSPS KAKENHI Grant Number 24KF0208

\section{Preliminaries}\label{s:prelims}
\setcounter{equation}{0}

\subsection{Cluster algebras}\label{ss:cluster-algs}

In this section, we recall the definition of two types of mutation, quiver mutation and cluster mutation, and introduce an example to see the construction of cluster algebras.

Let $Q=(V,E)$ be a quiver with $m$ nodes $V=\qty{1,2,\dots,m}$ and directed edges $E$. We assume that $Q$ does not possess any loops or oriented 2-cycles.  However, multiple edges between nodes are allowed, and when  two vertices have 
multiple edges between them we write $\xrightarrow{p}$ as a shorthand for $p$ parallel arrows.

\begin{defn}
    Let $Q$ be a quiver. \textit{Quiver mutation} at node $k$, to obtain the new quiver $\mu_{k}(Q)$, is performed by following the steps below:
    \begin{enumerate}
        \item For each full subquiver $i \xrightarrow{p} k \xrightarrow{q} j $, insert $pq$ edges $i \xrightarrow{pq} j $;
        \item Reverse all arrows which are connected to $k$; 
        \item Remove any 2-cycles which have been formed by inserting arrows. 

\end{enumerate}
\end{defn}

\begin{exmp}[Quiver mutation at node 2] {\ }

\begin{center}
\adjustbox{scale=0.8}{%
\begin{tikzcd}
	2 &&& 2 && {} & 2 &&& 2 \\
	&& {} \\
	3 && 1 & 3 && 1 & 3 && 1 & 3 && 1
	\arrow[from=1-1, to=3-1]
	\arrow[shift left=1, from=3-3, to=1-1]
	\arrow[shift right=1, from=3-3, to=1-1]
	\arrow[from=3-1, to=3-3]
	\arrow[shift left=2, from=3-4, to=3-6]
	\arrow[shift left=2, color={red}, from=3-6, to=3-4]
	\arrow[color={red}, from=3-6, to=3-4]
	\arrow[shift left=1, from=3-6, to=1-4]
	\arrow[shift right=1, from=3-6, to=1-4]
	\arrow[from=1-4, to=3-4]
	\arrow[from=3-7, to=1-7]
	\arrow[shift left=1, from=1-7, to=3-9]
	\arrow[shift right=1, from=1-7, to=3-9]
	\arrow[shift left=2, from=3-7, to=3-9]
	\arrow[color={red}, from=3-9, to=3-7]
	\arrow[shift left=2, color={red}, from=3-9, to=3-7]
	\arrow[from=3-10, to=1-10]
	\arrow[shift right=1, from=1-10, to=3-12]
	\arrow[shift left=1, from=1-10, to=3-12]
	\arrow[color={red}, from=3-12, to=3-10]
\end{tikzcd}\
}
\end{center}

\end{exmp}
As we assumed that the quiver has no 2-cycles and no loops, then we can associate it with an $m\times m$ skew-symmetric matrix. 
\begin{defn}
    Let $Q$ be a quiver with $m$ vertices and no 2-cycles 
    or loops. 
    Then  this quiver can be  encoded by the $n\times n$ skew-symmetric integer matrix $B=B(Q)=(b_{ij})$ by setting the matrix entry $b_{ij}$ to be the number of arrows from $i$ to $j$ minus the  number of arrows from $j$ to $i$. The matrix $B$ is called the  \textit{exchange matrix}. 
\end{defn}

As the quiver $Q$ can be represented by the exchange matrix $B=B(Q)$, one can formulate the quiver mutation in terms of entries of $B$, giving a new exchange matrix $\mu_{k}(B)$. This is referred as \textit{matrix mutation}, defined by the following formulae.

\begin{defn}
    Let $B$ be an exchange matrix and let $B^{'}=\mu_{k}(B)$ be the new exchange matrix obtained by applying mutation to the exchange matrix $B$ in direction $k$. The entries of $B^{'}=(b^{'}_{ij})$ are given by  
    \begin{equation}\label{matrixmu}
        b^{'}_{ij} =
        \begin{cases}
           - b_{ij} & \text{if} \ i= k \ \text{or} \ j=k  \\ 
            b_{ij} + \frac{1}{2} \qty(\abs{b_{ik}}b_{kj} + b_{ik}\abs{b_{kj}})& \text{otherwise}  \\ 
        \end{cases}
    \end{equation}
\end{defn}

Alongside quiver  mutation, and the 
corresponding matrix mutation (as above), cluster variables transform under \textit{cluster mutation}. 

\begin{defn}
    Let $\mathcal{F}=\mathbb{C}(x_{1},\dots,x_{m})=\mathbb{C}(\vb{x})$ denote the field of rational functions in $m$ independent variables $x_{1},\dots,x_{m}$ over $\mathbb{C}$, and denote the initial cluster by $\vb{x} = (x_{1},\dots,x_{m}) \in \mathcal{F}^{m}$. The \textit{cluster mutation} of $\vb{x}$ in direction $k$ is $\mu_{k}(\vb{x}) = \qty(x_{1},\dots,x_{k-1},x'_{k},\dots x_{m})\in \mathcal{F}^{m}$, where $x'_{k}$ is the element defined by the expression 
    \begin{equation}\label{mu1}
    \mu_{k}(x_{k}) =x'_{k} =  \dfrac{1}{x_{k}}\qty( \prod^{m}\limits_{\substack{j=1 \\ b_{jk} > 0}}x_j^{b_{jk}}  + \prod^{m}\limits_{\substack{j=1 \\ b_{jk} < 0}}x_j^{-b_{jk}} )  
\end{equation}
This expression is known as a (coefficient-free) \textit{exchange relation}.
\end{defn}

Note that we could work over base fields other than $\mathbb{C}$ but we will restrict to this choice for the  
geometric setting to be considered later.

As mentioned in the introduction, a pair  $(\vb{x},Q)$ is called an initial seed, but it is convenient to denote such a seed  alternatively 
by $(\vb{x},B)$, where $B=B(Q)$ is the exchange matrix corresponding to $Q$.  
Given an initial seed $(\vb{x},B)$ of size $m$, one can apply  mutation in $m$ possible directions, which produces $m$ new seeds. Subsequently, mutation can again be applied to each such seed in $m$ possible directions, and so on. It is important to note that applying two consecutive mutations in the same direction does not yield a new seed. This is due to the fact that mutation is involutive: $\mu_{i}^{2} = \mu_{i}\circ\mu_{i}=\id$.  The structure of the seeds can be encoded by labelling the vertices of a rooted $n$-valent tree by clusters. Note that, in general, mutations on different vertices do not commute (in the sense that $\mu_{i}\circ \mu_{j} \neq \mu_{j}\circ \mu_{i}$) unless the vertices are ``far apart'' (i.e.\ are at least distance 2 apart in the quiver). Then one can produce the collection of cluster variables which are induced by iterated cluster mutations in all directions. The set of all cluster variables obtained in this way generates the \textit{cluster algebra} corresponding to the seed $(\vb{x},B)$. 

\begin{defn}[Cluster algebra]
    The cluster algebra $\mathcal{A}(\vb{x},B)$ of \emph{rank} $m$ is the $\mathbb{C}$-subalgebra of the field $\mathcal{F}$ whose generating set is the set of all cluster variables produced by all possible sequences of mutations applied to the initial seed $(\vb{x},B)$ of size $m$. 
\end{defn} 
There is a more general way to define a cluster algebra, including coefficients. One way to do this is by introducing frozen variables. 
This requires a slightly wider 
class of exchange matrices, which are called 
skew-symmetrizable, defined as follows.

\begin{defn}
    An $m\times m$ integer matrix $B$ is called skew-symmetrizable 
    if there exists an integer diagonal matrix $D$ (called  a  \textit{skew-symmetrizer}) such that $(DB)^{T} = -DB$.
\end{defn}

Note that a skew-symmetric matrix $B$ satisfies the above definition with $D$ being the identity matrix. The notion of a cluster algebra can be extended by including \textit{frozen variables} in clusters, which are extra variables that do not mutate.  
An \textit{extended cluster}  $\tilde{\vb{x}}=(x_{1},x_{2},\dots,x_{m},x_{m+1},\dots,x_{m+s})$ consists 
of $m$ mutable cluster variables $x_{1},\dots,x_{m}$ together with $s$ frozen variables $x_{m+1},\dots,x_{m+s}$.  If the $(m+s)\times m$ matrix $\tilde{B}$ has an upper $m\times m$ submatrix that is skew-symmetrizable, then we call $\tilde{B}$ an \textit{extended exchange matrix}, and we write 
$\tilde{B}=(b_{ij})$, with the same notation for matrix elements as before. 
Thus we obtain a cluster algebra with initial seed $(\tilde{\vb{x}},\tilde{B})$ generated by cluster variables obtained from successive   mutations $\mu_{k}(\tilde{\vb{x}},\tilde{B}) = (\tilde{\vb{x}}',\tilde{B}')$ in non-frozen directions $1\leq k\leq m$, where $\tilde{\vb{x}}' = (x_{1},\dotsc, x_{k}',\dotsc,x_{m},x_{m+1},\dotsc,x_{m+s})$ and $\tilde{B}'$ are given by the 
exchange relation 
\begin{equation}\label{coeffmut}
    x_{k}'x_{k} = \alpha_{k}\prod^{m}\limits_{\substack{j=1 \\ b_{jk} > 0}}x_{j}^{b_{jk}}  + \beta_{k}\prod^{m}\limits_{\substack{j=1 \\ b_{jk} < 0}}x_{j}^{-b_{jk}} 
\end{equation}
where the coefficients are 
\begin{equation}\label{coeffvar}
    \alpha_{k} = \prod^{m+s}\limits_{\substack{j=m+1 \\ b_{jk} > 0}}x_{j}^{b_{j,k}}, \quad \beta_{k} = \prod^{m+s}\limits_{\substack{j=m+1 \\ b_{jk} < 0}}x_{j}^{-b_{j,k}} 
\end{equation}
and the matrix mutation is defined by \eqref{matrixmu}, 
so that the formula for 
$\tilde{B}'$ is the same as previously given for $B'$, except that the range of row indices is now from $1$ to $m+s$. These more general cluster algebras are referred to 
as \textit{cluster algebras of geometric type}.

Cluster algebras possess several interesting structural features. One of their most significant features 
is that cluster variables, obtained from successive mutations, are expressed as Laurent polynomials in the initial mutable cluster variables, defined over the ring of polynomials in the frozen variables with integer coefficients. 
This is known as the \textit{Laurent phenomenon}, and is stated as follows.

\begin{thm}[Laurent phenomenon]\label{LP}
Every cluster variable generated by cluster mutations belongs to the ring of Laurent polynomials over $\mathbb{Z}$ in the initial cluster variables, more precisely the ring $\mathbb{Z}[x_1^{\pm 1},\ldots,x_m^{\pm 1}, x_{m+1},\ldots,x_{m+s}]$. 
\end{thm}

This property of cluster mutations is essential for our later results. We will consider another special feature, 
which 
(unlike the Laurent phenomenon) 
may or may not be present in a cluster algebra, namely \textit{periodicity} under sequences of 
quiver/matrix mutations, which was introduced by Fordy and Marsh \cite{Fordy_2010}.

\begin{defn}[Mutation periodicity]\label{mutperiod}
Let $Q$ be a quiver with $m$ vertices. Then $Q$ 
is said to be mutation periodic with period $r$ if there exists a sequence of quiver mutations 
$\mu_{i_1},\ldots, \mu_{i_r}$ whose action is equivalent to a cyclic permutation of the labels of the quiver $Q$, i.e. 
\begin{align*}
            \mu_{i_{r}}\mu_{i_{r-1}}\cdots \mu_{i_{2}}\mu_{i_{1}}(Q) = \rho^{r}(Q)
        \end{align*}
where $\rho$  is the cyclic permutation 
$\rho=(1\,2\cdots m-1\,m)$. We say a skew-symmetric exchange 
matrix $B$ is mutation periodic (with period $r$) if its associated quiver is, so that 
\begin{align*}
            \mu_{i_{r}}\mu_{i_{r-1}}\cdots \mu_{i_{2}}\mu_{i_{1}}(B) = \rho^{r}(B)
        \end{align*}
also holds. 
\end{defn}

\begin{exmp}[Quiver of type $A_2$] \label{QA2}
The quiver associated with type $A_2$ is drawn as 
\begin{center}
\begin{tikzpicture}[every circle node/.style={draw,scale=0.6,thick},node distance=25mm]
  \node [draw,circle,"$1$"] (a1) at (0,0) {};
  \node [draw,circle,"$2$"] (a2) [right=of a1] {};
  
  \begin{scope}[>=Latex]
  \draw[-> , thick] (a1) edge (a2);
    \end{scope}

\end{tikzpicture}
    
\end{center}
The action of the mutation $\mu_{1}$ on the quiver is given by just reversing the direction of the single arrow. Permuting the labels via the transposition $\rho=(1\,2)$ will therefore return the quiver to its original state, so we have $\mu_1(Q)=\rho(Q)$ and 
\begin{equation} \label{A2mu1}
\mu_1(B)=\rho(B), \qquad where\quad 
B=
\left(\begin{array}{cc}  0 & 1 \\ 
-1 & 0
\end{array}\right)
\end{equation}
is the exchange matrix of type $A_2$. So in this case, the quiver and its associated exchange matrix are mutation periodic with period 1. 
\end{exmp}

In a slightly more general setting, periodicity of exchange matrices/quivers was defined by Nakanishi \cite{nakanishi2011periodicities} as follows. 
\begin{defn}[$\hat{\rho}$-periodicity] Let $Q$ be a quiver with $n$ vertices and let $\hat{\rho}$ be a permutation of the nodes 
$V=\qty{1,2,\dots,m}$ of $Q$. Then the quiver $Q$ is  $\hat{\rho}$-periodic if there exists sequence of mutations such that 
\begin{align*}
    \mu_{i_{r}}\mu_{i_{r-1}}\cdots \mu_{i_{2}}\mu_{i_{1}}(Q) = \hat{\rho}(Q), 
\end{align*} 
or equivalently the associated exchange matrix $B=B(Q)$ 
satisfies 
\begin{align*}
    \mu_{i_{r}}\mu_{i_{r-1}}\cdots \mu_{i_{2}}\mu_{i_{1}}(B) = \hat{\rho}(B). 
\end{align*}
\end{defn}

\begin{exmp} The composition  of mutations $\mu_5\mu_1$ 
transforms the following quiver as 
\begin{center}
 \begin{tikzpicture}[every circle node/.style={draw,scale=0.6,thick},node distance=15mm]
  \node [draw,circle,"$2$"] (a1) at (0,0) {};
  \node [draw,circle,"$3$"] (a2) [right=of a1] {};
  \node [draw,circle,"$4$" right] (a3) [below right=of a2] {};
  \node [draw,circle,"$5$" below] (a4) [below left=of a3] {};
  \node [draw,circle,"$6$" below] (a5) [left=of a4] {};
  \node [draw,circle,"$1$" left]  (a6) [above left=of a5] {};
  \begin{scope}[>=Latex]
  \draw[->> , thick] (a1) edge (a2);
  \draw[-> , thick] (a2) edge (a3);
  \draw[-> , thick] (a4) edge (a3);
  \draw[-> , thick] (a4) edge (a6);
  \draw[-> , thick] (a6) edge (a1);
  \draw[-> , thick] (a6) edge (a5);
  \draw[-> , thick] (a5) edge (a1);
  \draw[-> , thick] (a1) edge (a4);
  \draw[-> , thick] (a4) edge (a6);
  \draw[-> , thick] (a5) edge (a2);
  \draw[-> , thick] (a2) edge (a4);
  \draw[-> , thick] (a3) edge (a1);
  \draw[-> , thick] (a3) edge (a5);
  \draw[-> , thick] (a2) edge (a6);
    \end{scope}
\draw [-{Latex[length=3mm]}] (3.5,-1.2) -- (5,-1.2) node[midway,sloped,above] {$\mu_{5}\mu_{1}$};
\node [draw,circle,"$2$"] (a1) at (7,0) {};
  \node [draw,circle,"$3$"] (a2) [right=of a1] {};
  \node [draw,circle,"$4$" right] (a3) [below right=of a2] {};
  \node [draw,circle,"$5$" below] (a4) [below left=of a3] {};
  \node [draw,circle,"$6$" below] (a5) [left=of a4] {};
  \node [draw,circle,"$1$" left]  (a6) [above left=of a5] {};
  \begin{scope}[>=Latex]
  \draw[-> , thick] (a1) edge (a2);
  \draw[-> , thick] (a1) edge (a6);
  \draw[-> , thick] (a5) edge (a1);
  \draw[-> , thick] (a3) edge (a1);
  \draw[->> , thick] (a2) edge (a3);
  \draw[-> , thick] (a6) edge (a2);
  \draw[-> , thick] (a2) edge (a5);
  \draw[-> , thick] (a4) edge (a2);
  \draw[-> , thick] (a6) edge (a3);
  \draw[-> , thick] (a3) edge (a5);
  \draw[-> , thick] (a3) edge (a4);
  \draw[-> , thick] (a4) edge (a6);
  \draw[-> , thick] (a5) edge (a4);
  
    \end{scope}

\end{tikzpicture}
\end{center}
If we permute the labels of the  quiver on the 
right-hand side by $\hat{\rho} = (1\,2\,3\,4\,5\,6)$, then we retrieve the original quiver $Q$. So we have 
$$ 
\mu_5\mu_1 (Q) = \hat{\rho}(Q), 
$$
and in this case $Q$ is $\hat{\rho}$-periodic. 
\end{exmp}

In the examples above, we have seen that there are particular  quivers with the property that a certain composition of $r$ mutations satisfies $\mu_{i_{r}}\cdots\mu_{i_{1}}(Q) = \hat{\rho}(Q)$ for 
some permutation $\hat\rho$. Suppose we define a map $\varphi = \hat{\rho}^{-1}\mu_{i_{r}}\cdots\mu_{i_{1}}$; then the action of $\varphi$ preserves the structure of the quiver, and the same is true for $B=B(Q)$, so equivalently we may write 
$$\varphi(B)=B.$$ 
The action of $\varphi$ on the cluster variables 
$\vb{x}=(x_1,\ldots,x_m)$ induces a birational map on $\mathbb{C}^m$ that we refer to as a \textit{cluster map}.
\begin{defn}[Cluster map]
\label{cmap} 
Let $(\vb{x},B)$ be an initial seed with an initial cluster $\vb{x}$ and mutation periodic 
($\hat{\rho}$-periodic) quiver $Q$ such that $B=B(Q)$.  Then a birational map $\varphi \colon \mathbb{C}^{m} \to \mathbb{C}^{m}$ defined by 
\begin{align*}
	\varphi = \hat{\rho}^{-1}\mu_{i_{r}}\mu_{i_{r-1}}\cdots \mu_{i_{2}}\mu_{i_{1}},
\end{align*}
which acts on the seed according to 
$\varphi(\vb{x},B) = (\varphi(\vb{x}),B)$, is called a 
cluster map. 
\end{defn}

\begin{exmp}
For Example \ref{QA2}, the cluster map $\varphi=\rho^{-1}\mu_1$ 
is given by \eqref{lyness5}. Every orbit of this map is periodic with period 5, i.e. $\varphi^5=\mathrm{id}$. This is the simplest example of Zamolodchikov periodicity. 
\end{exmp}

The classification of cluster algebras of finite type (having only finitely many seeds) in terms of Dynkin diagrams of simple Lie algebras was carried out by Fomin and 
Zelevinsky in \cite{fomin2006cluster}. In the context of an associated cluster map $\varphi$, this entails that all orbits of 
$\varphi$ are periodic with the same period: in due course, we will describe the details for cluster maps obtained 
in the even rank case of $A_{2N}$.   

In general, 
however, the behaviour of the dynamical system defined by a 
cluster map $\varphi$ can be extremely complicated. The reason 
we focus specifically on cluster maps, rather than arbitrary compositions of mutations, is that they naturally preserve 
an associated symplectic and/or Poisson structure. In 
this context, this allows the application of 
a suitable definition of discrete integrability, as described 
in the next subsection. 

\subsection{Poisson brackets and discrete integrable systems}
\label{ss:discrete-int-sys}

A cluster map $\varphi\colon \,\mathbb{C}^{m} \to \mathbb{C}^{m}$ is a birational map, whose  iterates can be  considered as a discrete dynamical system. 
In order to define a notion of Liouville integrability for such a map, it is usually necessary to have an invariant symplectic structure, as in \cite{10.3792/pjaa.63.198}, or more generally a Poisson structure (see e.g. \cite{2013, hone2019cluster}), 
yet for a general dynamical system there is no available 
algorithm to find such a structure. 
Fortunately, in the context of cluster maps, there is a natural 
presymplectic form that is compatible with the cluster algebra, which can be used to obtain both symplectic and Poisson structures for such maps.

We begin by recalling some basic facts about Poisson brackets. Given a manifold/variety $M$ of dimension $m$, let $\cal F$ denote a suitable algebra of functions on $M$ defined over some  field. In the traditional context of classical mechanics, the phase space $M$ is a smooth real manifold, the field is $\mathbb{R}$, and ${\cal F} =  C^{\infty}(M)$ denotes smooth functions on $M$. However, in the algebraic setting that is 
required here, we will take $M=\mathbb{C}^m$, and consider the rational functions ${\cal F}=\mathbb{C}(x_1,\ldots,x_m)$, which is a $\mathbb{C}$-algebra. 
A \textit{Poisson bracket} is a 
skew-symmetric bilinear map $\qty{\cdot,\cdot} \colon \,
{\cal F} \times {\cal F} \to {\cal F}$ 
that satisfies the following two properties, which 
are required to hold for all $f,g,h\in{\cal F}$:  
\begin{enumerate}
    \item Leibniz rule: $\qty{fg,h} = f\qty{g,h} + \qty{f,h}g$
    \item Jacobi identity: $\qty{f,\qty{g,h}} + \qty{g,\qty{h,f}} + \qty{h,\qty{f,g}} = 0 $. 
\end{enumerate}
Any $M$ equipped with such a bracket is called a Poisson manifold (or Poisson variety). 
In the local coordinates $\vb{x}= \qty(x_{1},\dots,x_{m})$, the explicit form of the Poisson bracket between two functions $f$ and $g$ is written as 
\begin{equation}\label{expoissbrackt}
    \qty{f,g} = \sum_{i,j} \vb{P}_{ij}(\vb{x}) \pdv{f}{x_{i}}\pdv{g}{x_{j}}
\end{equation}
where the coefficients
\begin{equation}
   \vb{P}_{ij}(\vb{x}) = \qty{x_{i},x_{j}}
\end{equation}
are the entries of an $m\times m$ skew-symmetric matrix. 
Equivalently, the Poisson structure defined by the bracket 
$\{ \, , \} $
is encoded into the bivector field $\vb{P}$ given locally by 
$$\vb{P}=\sum_{i<j}  \vb{P}_{ij}(\vb{x}) \, \partial_i\wedge\partial_j.
$$
Note that the rank of the Poisson structure at the 
point $\vb{x}$ is the rank of the skew-symmetric matrix $(\vb{P}_{ij})$ 
at $\vb{x}$, 
and if it is full rank
then this must be equal to the dimension of the manifold, in which case $M$ must be even-dimensional. 


If the Poisson structure is of full rank everywhere on $M$ (or at least on some open set), then it has an inverse  such that 
$\vb{P}\iprod \om=1$,\footnote{ The symbol $\iprod$ denotes the interior product, or hook product, giving the contraction between vector fields and differential forms.} which defines a closed, nondegenerate 
2-form $\om$ called the \textit{symplectic form} on $M$. Conversely, if $M$ is a symplectic manifold, meaning that it has a globally defined nondegenerate 2-form $\om$ with $\dd \om =0$, then its inverse $\vb{P}$ is a Poisson structure of full rank.

The main focus of study in Hamiltonian mechanics is on the solutions of 
Hamilton's equations, which are the ordinary differential equations 
\begin{equation}\label{hameq}
\dot{\vb{x}} = X_H (\vb{x})
\end{equation}
(with the dot denoting time derivative), 
defined by a fixed function $H\in{\cal F}$ (the Hamiltonian) and the associated vector field $X_H$ given by 
\begin{equation}\label{Hvector}
    X_{H}(\cdot) = \qty{\cdot,H} = \vb{P} \iprod \dd H.
\end{equation}
Such a vector field is called a \textit{Hamiltonian vector field}. 
In the setting of canonical Hamiltonian mechanics, when $M$ is a  symplectic manifold, the Hamiltonian vector field admits an equivalent definition in terms of the symplectic form $\om$, being specified by the equation 
$$ 
\dd H = \om (\cdot, X_H).
$$
Subject to Hamilton's equations, the time derivative of any function $f\in{\cal F}$ is given by 
$$ 
\dot{f}=X_H(f) = \{ f,H\}, 
$$
and this vanishes for any non-constant function $f$ that is in involution with $H$ (i.e.,
Poisson commutes with it). 
Any such $f$ with 
$\{ f,H\}=0$ is called a constant of motion, or \textit{first integral}. 

Clearly $H$ itself is a first integral for Hamilton's equations, since $\{H,H\}=0$ by skew-symmetry. In the setting of canonical Hamiltonian mechanics on a real symplectic manifold $M$ of dimension $m=2N$, when the Poisson structure is of full rank, Hamilton's equations (\ref{hameq}) are said to be \textit{completely integrable} (or 
\textit{Liouville integrable}) whenever there are $N$ independent first integrals $H_1=H$, $H_2,\ldots,H_N$ that are in involution with one another, i.e. $\{ H_i,H_j\}=0$ for all $i,j$. 
Liouville's theorem says that, for a completely integrable system, Hamilton's equations can be integrated by quadratures, while Arnold's addendum to the theorem states that the compact, 
connected level sets of the first integrals 
$H_1,\ldots,H_N$ are diffeomorphic to $N$-dimensional tori $T^N$ on which the flow of 
Hamilton's equations is conditionally periodic \cite{arnold}. 

We can now discuss the analogous notions in the case of discrete dynamical systems (maps). On a 
symplectic manifold (or variety) $M$, a regular map $\varphi$ is said to be symplectic whenever 
$\varphi^*\om=\om$. In the more general setting where $M$ is a Poisson manifold, we have the 
following 
\begin{defn}[Poisson map]\label{pm}
Let $M$ be 
equipped with the Poisson bracket $\qty{\cdot,\cdot}$ and let $\vb{x} = (x_{1},x_{2},\dots, x_{m})$ be local coordinates on $M$. Then a map $\varphi\colon M \to M$ is a Poisson map if it preserves the Poisson bracket, that is, 
\begin{equation}\label{poismapcond}
    \varphi^{*}\qty{x_{i},x_{j}} = \qty{\varphi^{*}x_{i},\varphi^{*}x_{j}}
    \qquad \mathrm{for} \,\, \mathrm{all} \quad i,j. 
\end{equation}
In particular, when the Poisson structure has full rank, the above condition implies that 
$\varphi$ is a symplectic map on $M$. 
\end{defn}

In the discrete context, a function $H$ is said to be a first integral for a map $\varphi$ whenever it is invariant under the action of 
$\varphi$, that is 
$\varphi^* H=H$. 
Following Maeda \cite{10.3792/pjaa.63.198}, there is a precise analogue of Liouville integrability 
for symplectic maps: 
such a map is completely integrable whenever 
it has first integrals $H_j$, $j=1,\ldots, N$ satisfying the following conditions:  
\begin{itemize}
\item The $N$ first integrals of 
$\varphi$ are independent, i.e. 
$\rd H_1\wedge \rd H_2 \wedge \cdots \wedge \rd H_N \neq 0; $
\item $\{ H_i,H_j\}=0$ for $1\leq i,j\leq N$. 
\end{itemize}
Moreover, in the setting of a real symplectic 
manifold $M$, Maeda showed that an appropriately modified version of 
Liouville's theorem holds, so that an analytic solution for the iterates of 
an integrable symplectic map $\varphi$ can be obtained via quadratures. 

In the context of Poisson manifolds, a more general definition of Liouville integrability is required. For the sake of completeness, here we provide a definition that is applicable to  Poisson maps, which is adapted from the setting of Hamiltonian systems on Poisson varieties 
considered in \cite{vanhaecke}. We restrict 
ourselves to the 
specific case of $M=\mathbb{C}^m$.  

\begin{defn}[Liouville integrable map: Poisson case]  
\label{dis}
Let $M=\mathbb{C}^m$ be equipped with 
a Poisson bracket of rank $2N \leq m$ 
on the $\mathbb{C}$-algebra of rational 
functions $\cal F$ on $M$,
and suppose that 
$\varphi: \, \mathbb{C}^{m} \to \mathbb{C}^{m}$ is a Poisson map.  
Then $\varphi$ is said to be \textit{Liouville integrable} if it admits  $m-N$ first integrals, of which 
\begin{itemize} 
\item  $m-2N$ are invariant Casimir functions 
$\mathcal{C}_{i}$, i.e. 
$\varphi^{*}\mathcal{C}_{i} =\mathcal{C}_{i}$ 
and they satisfy $\{\mathcal{C}_{k},F\} =0 $ for 
all functions $F\in{\cal F}$;  
\item  the remaining $N$ are additional invariant functions 
$H_{j}$ in involution, so that 
$\varphi^{*}H_{j} = H_{j}$ and 
$\{H_{i},H_{j}\} = 0$ for all $i,j$.
\end{itemize} 
Moreover, 
$\rd {\cal C}_1\wedge \rd {\cal C}_2 \wedge \cdots \wedge \rd {\cal C}_{m-2N} 
\wedge\rd H_1\wedge \rd H_2 \wedge \cdots \wedge \rd H_N \neq 0$, i.e.\  these invariants are required to be   
functionally independent.  
\end{defn}

In the case of  a cluster algebra 
$\mathcal{A}(\vb{x},B)$ of rank $m$, 
with  the  $m$-tuple of cluster variables 
$\vb{x}=(x_j)$ being affine 
coordinates on $M=\mathbb{C}^m$, 
we are interested in Poisson brackets taking a particularly simple form on clusters, namely  
\begin{equation}\label{poissbrk}
    \qty{x_{i},x_{j}} = P_{ij} x_{i} x_{j} ,
\end{equation}
where the  $m \times m$ skew-symmetric matrix 
$P=(P_{ij})$ is referred to as the associated \textit{Poisson matrix}. 
The terminology coined by Gekhtman, Shapiro and Vainshtein in \cite{gekhtman2003cluster}
for this is a \textit{log-canonical Poisson bracket}. Such a bracket is compatible with the cluster algebra if it satisfies the following requirement.
\begin{defn}
For a cluster algebra $\mathcal{A}(\vb{x},B)$, a Poisson bracket $\qty{\cdot,\cdot}$ 
between functions in $\mathcal{F}=\mathbb{C}(\vb{x})$ is said to be \emph{mutation compatible} if the bracket restricted to any cluster is log-canonical.
\end{defn}
Given the $m\times m$ exchange matrix 
$B$ for the cluster algebra $\mathcal{A}(\vb{x},B)$, with diagonal skew-symmetrizer 
$D$ such that $(DB)^{T} = - (DB) $,  by 
imposing the condition of mutation compatibility, one can obtain the following result (of 
which further details can be found in \cite{gekhtman2003cluster}, \cite{inoue2011difference}).  

\begin{thm}
    Assume that $B$ is skew-symmetrizable with skew-symmetrizer $D$ and that $B$ is of full rank.  Then, for any scalar 
    $\lambda\neq 0$, the Poisson matrix     %
    \begin{equation}\label{poissm1}
        P = \lambda D B^{-1}
    \end{equation}
    on the initial cluster $\vb{x}$ extends to a mutation compatible Poisson bracket. In addition to this, the product $PB$ is mutation invariant. 
\end{thm}

In the above result, $B$ and $D$ are both integer matrices of full rank, and it is convenient to choose non-zero $\lambda \in \mathbb{Q}$ so that $P$ is also an integer matrix.  
The requirement of full rank in this result is 
the reason we consider type $A_{2N}$ in what follows, since the rank of the cluster algebra is $m=2N$ in that case, and this coincides with the rank of the matrix $B$. (For Dynkin types $A_{2N+1}$, 
on the other hand, the matrix $B$ has a one-dimensional kernel; cf.\ 
the case of $A_3$ in \cite{hone2021deformations, hkm}.) The fact that there is a mutation compatible Poisson bracket in this setting then implies that a cluster map $\varphi$, being formed from a composition of mutations and a permutation, is a Poisson map in the sense of Definition \ref{pm}. 

%
For a general cluster algebra $\mathcal{A}(\vb{x},B)$, when $B$ may be degenerate, 
it was shown in \cite{gekhtman2003cluster} (see also \cite{fockg}) that there is always closed 2-form that is compatible with cluster mutations,  
which in the case of skew-symmetric $B$  
is given by the log-canonical presymplectic form 
\begin{equation}\label{sympform}
    \omega =\sum_{i < j } \frac{b_{ij}}{x_ix_j}\,\dd x_i \wedge \dd x_{j}  = \sum_{i < j }b_{ij}\,\dd \log x_{i} \wedge \dd \log x_{j}. 
\end{equation}
%
Under the action of the cluster mutation 
$\mu_{k}$ sending $(\vb{x},B)$ to $(\vb{x}',B')$, this 2-form is transformed to 
$$
\omega' =
\mu_k^* \,\om
=\sum_{i < j }b'_{ij}\,\dd \log x'_{i} \wedge \dd \log x'_{j},$$
where the new matrix entries are defined as in 
\eqref{matrixmu}. 
Note that if the exchange matrix is nondegenerate then  \eqref{sympform} is an honest symplectic form. 

In what follows, we will start by considering 
cluster maps in the case of Dynkin type $A_{2N}$, 
which preserve a symplectic form given by 
\eqref{sympform} with a nondegenerate exchange matrix $B=(b_{ij})$. However, 
when we consider deformations of these maps 
that preserve the same $\om$, we 
will find that they pull back to a cluster map defined an enlarged algebra, for which the exchange matrix is degenerate. 
Nevertheless, the deformed maps considered are themselves symplectic, so that definition the given by  Maeda \cite{10.3792/pjaa.63.198}  for Liouville integrability of maps is directly applicable; equivalently, in the statement of 
Definition \ref{dis}, this means that $m=2N$ and there are no nontrivial Casimirs.  

For a recent review of integrable systems in the context of cluster algebras, the reader is referred to \cite{GIintegrable}. 

\subsection{Singularity confinement analysis of difference equations}\label{ss:sing-analysis}

In this subsection, we briefly describe a heuristic approach to the analysis of discrete dynamical systems, called the \textit{singularity confinement test}, which was introduced as a tool to detect integrability. We follow the presentation in \cite{Hietarinta_Joshi_Nijhoff_2016}, and use  particular relevant examples to demonstrate the test procedure. 

The singularity confinement test was proposed by Grammaticos, Ramani and Papageorgiou in \cite{Grammaticos1991DoIM} as a means to assess the potential integrability of a discrete dynamical system. The original 
motivation for this test came from the local singularity analysis of the solutions of 
ordinary differential equations (ODEs), 
often referred to as the 
\textit{Kowalevski--Painlev\'{e} test}, which had been used extensively to 
detect ODEs that possess the  \textit{Painlev\'{e} property}, meaning that the solutions are single-valued around movable singular points. (A ``movable'' singularity is one whose position depends on an arbitrary choice of initial data.)  

The Kowalevski--Painlev\'{e} test 
consists of verifying 
necessary local conditions for an ODE to admit a unique 
analytic continuation around any movable singularity, by checking that all such 
singularities are poles at which the Laurent expansion of the solution has the required number of arbitrary constants. 
Motivated by the goal of trying to 
identify the discrete versions of Painlev\'{e} equations, 
Grammaticos, Ramani and Papageorgiou
suggested an analogue of these notions 
in the setting of difference equations: given a 
discrete iteration indexed by an integer $n$, 
singularity confinement requires that if the equation reaches a singularity at some arbitrary 
(``movable'') index value $n_0$, then the solution can be continued beyond this point, 
so that it is defined (non-singular)  for suitably large  indices $n>n_0$. 
A more formal definition is provided below, but to begin with we will illustrate  the procedure with an 
example. 

\begin{exmp}[Lyness recurrence] 
When it is rewritten in the form of 
a recurrence of second order, 
the general Lyness map \eqref{genlyness} is 
given by 
\begin{equation}\label{Lynessrecurrence}
    x_{n+2}x_{n} = a x_{n+1} + b
\end{equation}
where $a$ and $b$ are non-zero parameters. It is clear that 
a singularity can occur if $x_{n} = 0 $ 
for some $n$, and this arises 
whenever it happens that 
$ax_{n-1} +b=0$ at the previous step.  
Let us suppose that at some arbitrary index  $n=n_{0}$ the iteration produces $x_{n_{0}} = 0 $ with $x_{n_{0}+1} = u$ for a generic  regular 
(finite) value $u\in\mathbb{C}$, assumed arbitrary. Then subsequent  iterations give
$$
    x_{n_{0} +2 } = \frac{ax_{n_{0}+1} +b }{x_{n_{0}}} = \infty , \quad 
    x_{n_{0} +3 } = \frac{\infty + b}{u} = \infty , \quad 
    x_{n_{0} +4 } = \frac{\infty+b}{\infty} = \frac{\infty}{\infty}. 
$$
The ratio $\frac{\infty}{\infty}$ is an ambiguous term: it is a true singularity, with loss of information (the dependence on the arbitrary parameter $u$ has disappeared). To resolve it, we 
approach the neighbourhood of the singularity by introducing a small quantity $\epsilon$, and 
letting $x_{n_{0}} = \epsilon$.  Then 
the iteration of the recurrence \eqref{Lynessrecurrence} yields 
\begin{align*}
    & x_{n_{0}} = \epsilon \\
    & x_{n_{0}+1} = u \\
    & x_{n_{0}+2} = (au +b )\epsilon^{-1} \\
    & x_{n_{0}+3} = \frac{a(au +b)}{u}\epsilon^{-1}  + \frac{b}{u} \\ 
    & x_{n_{0}+4} = \frac{a^2}{u} + O(\epsilon) \\
    &x_{n_{0}+5} = \qty(\frac{a^3 + b u }{a(au +b )})\epsilon + O(\epsilon^2) \\
    &x_{n_{0}+6} = \frac{bu}{a^2} + O(\epsilon)
\end{align*}
Notice that, in the limit $\epsilon\to 0$, the value of $x_{n_{0}+4}$ is no longer ambiguous: it is now well defined, and finite values of the iterates are observed for 
$n>n_0+3$ in this limit. We can also reverse the iteration to a regular value of 
$x_{n_0-1}$. When  
$\epsilon \to 0$, the sequence produced by the iteration can be written symbolically as 
%
\begin{align} \label{singp}
    \qty(\dots,R, 0 ,R, \infty, \infty, R, 0,R,\dots ), 
\end{align}
where $R$ denotes a regular non-zero value. 
The symbolic sequence above is referred to as the associated singularity pattern. The iteration enters $0$ and then it passes through simple poles, 
corresponding  to the value $\infty$ with   
terms of $O(\epsilon^{-1})$, and another simple zero,   with the value $0$ appearing at $O(\epsilon)$, after which it returns to a regular value which depends on the prior initial value $u$. In this case, the singularity is said to be confined. 
\end{exmp} 

Note that the difference equation 
\eqref{Lynessrecurrence} above is equivalent to iterating a birational map $\mathbb{C}^2\rightarrow\mathbb{C}^2$,   namely \eqref{genlyness}. We can rewrite this with indices as 
\begin{equation}\label{Lynessclustermap}
    \varphi_{(a,b)} : \,  \mqty(x_{1,n}\\x_{2,n} ) \mapsto 
     \mqty(x_{1,n+1}\\x_{2,n+1} ) = \mqty(x_{2,n} \\[1em] 
    \dfrac{b + a x_{2,n}}{x_{1,n}} ),  
\end{equation}
where, compared with \eqref{Lynessrecurrence}, we have identified 
$x_n = x_{1,n}$ and $x_{n+1}=x_{2,n}$. 

The following definition of singularity confinement for affine maps was given by 
Goriely and Lafortune  in \cite{2004}. 
%
\begin{defn}
For an $m$-dimensional birational map $\psi$, of the  form 
\begin{align*}
    \psi\colon\,\, \vb{x} =   \mqty(x_{1}\\x_{2} \\ \vdots  \\x_{m} ) \mapsto 
    \mqty(f_{1}(\vb{x}_n) \\ f_{2}(\vb{x}_n)  \\ \vdots  \\f_{m}(\vb{x}_n)  )
\end{align*}
for suitable functions $f_{j}\in\mathbb{C}(\vb{x}_n)$, $1\leq j\leq m$, 
the  singularities  of the map are defined as  points $\vb{y}=(y_1,y_2,\dots, y_m)\in\mathbb{C}^m $ where either the right-hand side above is undefined, or the Jacobian determinant 
$\det D\psi$ of $\psi$ vanishes. A singularity is said to be confined if there is some positive integer $K$ such that $\lim_{\vb{x} \to \vb{y}} \psi^{K} (\vb{x})$ exists and 
$\lim_{\vb{x} \to \vb{y}} \det D\psi^{K} (\vb{x})\neq 0$.  
\end{defn}

Now we consider an example which will subsequently be of particular relevance: we will shortly see that it arises as a deformed $A_{2}$ cluster map.
\begin{exmp} \label{exA2} 
Consider
\begin{equation}\label{deformedA2}
    \tilde{\varphi}\colon\,\, \mqty(x_{1}\\x_{2}) \mapsto \mqty(\dfrac{1 + a_{1}x_{2}}{x_{1}} \\[1em] \dfrac{x_1 +a_{2}(1 + a_{1}x_{2})}{x_{1}x_{2}})
\end{equation}
This birational map has singularities when $x_1 =0$, and when $x_2 = 0$. Beginning with the case $x_1 = 0$, we once introduce a small quantity $\epsilon $. After setting the initial data to be $x_{1} = u, \ x_{2} = \dfrac{-1+\epsilon u}{a_{1}}$, the next iteration of the map reaches $\qty(\epsilon, \dfrac{(1 + \epsilon a_{2})a_{1}}{\epsilon u -1})$, where it becomes singular when $\epsilon\to 0$:
\begin{align*}
&\mqty(-(a_{1} - 1)(a_{1} + 1) \epsilon^{-1}  - a_{1}^2(u + a_{2})\\[1em] \dfrac{a_2(a_1 - 1)(a_1 + 1)}{a_1}\epsilon^{-1} + \dfrac{u a_2 + a_2^2 - 1}{a_1}) \to 
    \mqty(-a_{2} + \epsilon (u + a_{2})a_{2} + O(\epsilon^2) \\[1em] \epsilon  \dfrac{-(a_{2} - 1)(a_{2} + 1)a_{1} }{a_{2}(a_{1} - 1)(a_{1} + 1)} + O(\epsilon^2)) \\[1em]
& \to  \mqty( -\dfrac{1}{a_{2}} + O(\epsilon) \\[1em] \dfrac{-a_2^2 + u(a_1^2 - 1)a_2 + a_1^2}{a_1(a_2^2 - 1)} + O(\epsilon))
\end{align*}
As $\epsilon \to 0 $, the sequence becomes 
\begin{equation}\label{singA21}
  \mqty(\infty \\[1em] \infty )  \to  \mqty( - a_{2} \\[1em] 0 
  ) \to  \mqty( -\dfrac{1}{a_{2}} \\[1em]\dfrac{-a_2^2 + u(a_1^2 - 1)a_2 + a_1^2}{a_1(a_2^2 - 1)} )
\end{equation}
Therefore after finite number of steps $K$ there is a non-vanishing 
limit $\lim_{\vb{x} \to \vb{y}} \psi^{K} (\vb{x}) \neq 0 $. Hence the singularity is confined. 

For the case of the singularity $x_2=0$, we take the same procedure with the initial iterates $x_{1} = a_2(1 + a_1 u ) / (\epsilon u - 1)$ and then $\epsilon \to 0 $ give us the 
sequence of iterations  
\begin{equation}\label{singA22}
  \mqty(-a_{2} \\[1em] \infty )  \to  \mqty( \infty \\[1em] (a_1 -a_1a_2^2) / (a_2^2 -1 ) ) \to  \mqty( 0\\[1em] -1/a_1 ) \to \mqty( R\\[1em] R )
\end{equation}
where $R$ is some non-zero regular value. Therefore both singularities are confined. 
\end{exmp}

\begin{rem}\label{A2rem} 
When  $b\neq 0$, any Lyness map $\varphi_{(a,b)}$ 
is conjugate to an ``elementary'' one, of the 
form $\varphi_{(a',1)}$ with a single 
parameter $a'=a/\sqrt{b}$, via rescaling each of the 
coordinates $(x_1,x_2)\in\mathbb{C}^2$ 
by a factor of $\sqrt{b}$. 
The map \eqref{deformedA2} is the composition of two elementary Lyness maps \eqref{genlyness}, with different parameters, namely 
\begin{equation}\label{psicomp}
    \tilde{\varphi} = \varphi_{(a_2,1)}\circ \varphi_{(a_1,1)}. 
\end{equation}
We shall discuss how this composition arises from 
deformed cluster mutations below. 
\end{rem}
Example~\ref{exA2} showed that this particular map possesses the confinement property. In addition to this, one can show that the map \eqref{deformedA2} is Liouville integrable, as there exists a first integral that is invariant under the map, given by
\begin{equation}\label{firstintA2}
  H_{A_2}= x_{1} + \frac{a_1}{a_2}\,x_2 +  
  \frac{1+a_{1}^2}{x_{1}}  + \frac{a_1(1+a_2^2)}{a_2x_2}+ a_1 \left( \frac{x_{1}}{x_{2}} + \frac{x_{2}}{x_{1}} + \frac{1}{x_{1}x_{2}} \right) .
\end{equation}
The level sets of the first integral $H_{A_2}$ form a pencil of biquadratic curves, and  is an example
of a map of Quispel-Roberts-Thompson (QRT) type 
\cite{Duistermaat2010DiscreteIS}. 

The preceding examples, may appear to suggest that the confinement property always leads to integrable maps, but  this is not the case, as was pointed out by  Hietarinta and Viallet  \cite{Hietarinta_1998}, who provided an example of a  non-integrable map in the plane that passes the confinement test. In other words, the 
singularity confinement property is not sufficient for  integrability.

\begin{rem} \label{padic} 
The calculations required to find a 
a singularity confinement pattern can be somewhat 
laborious, but there is an alternative approach based on $p$-adic arithmetic that is usually more convenient in practice. Following  \cite{Kanki_2012}, one can consider 
a dynamical system defined over the field of $p$-adic numbers $\mathbb{Q}_p$, which (for prime $p$) is the completion of $\mathbb{Q}$ with respect to the $p$-adic norm $\abs{\ \cdot \ }_{p}$. If $x_n$ is a coordinate on an orbit at some step $n$, then it admits a well-defined reduction to the finite field $\mathbb{F}_p$ provided that 
$\abs{x_n}_{p}\leq 1$, while a singularity arises at the next step if $\abs{x_{n+1}}_{p} > 1 $. 
Moreover, in practice one can detect singularities by 
computing the prime 
factors of the terms in a rational orbit (defined over $\mathbb{Q}$). 
For instance, taking the Lyness map, given by the recurrence \eqref{Lynessrecurrence}, with a particular set of rational coefficients and initial values $a=1,b=2,x_0=1,x_1=1
\in\mathbb{Q}$, the first few terms of the orbit are factorized as 
$$ 
1,1,3,5,\frac{7}{3},\frac{13}{3\cdot 5},\frac{43}{5\cdot 7}, \frac{3\cdot 113}{7\cdot 13}, \frac{5\cdot 521}{13\cdot 43}, \ldots, 
$$
and for each successive prime that appears, from the numerators/denominators one can read off the same singularity pattern as in \eqref{singp}, wherever there is a zero/pole $\bmod p$, as above for  i.e.\ $p=3,5,7,13,43$ etc. 
\end{rem}

\subsection{Algebraic entropy}\label{s:algentrop}

As noted above, passing the singularity confinement test does not guarantee integrability. However, there is an alternative test which is a stronger indicator, and in some circumstances is known to be sufficient 
for integrability. 
This involves a quantity $\varepsilon$, called \textit{algebraic entropy}, introduced by Bellon and Viallet in \cite{Bellon_1999}, defined by 
\begin{equation}
   \varepsilon =  \lim_{n \to \infty} \frac{\log d_{n} }{n}
\end{equation}
where $d_{n} = \deg(\varphi^{n})$ is the degree 
of the $n$th iterate of the map $\varphi$ (given by the maximum of the degrees of the components of $\varphi^{n}$, written in affine coordinates, or the homogeneous degree 
of all components, for a map given in projective space). 
For a generic map, the degrees of iterates grow exponentially in $n$, so that the entropy is positive.  But in certain special cases, cancellation of common factors can occur under iteration, reducing the degree $d_{n}$. In many examples,  \cite{Bellon_1999, Hietarinta_1998},   
it was observed that
\begin{itemize}
    \item exponential degree growth corresponds to chaotic behaviour, and
    \item polynomial degree growth corresponds to regular behaviour (integrability).
\end{itemize}
Subexponential degree growth is precisely the case of zero entropy, $\varepsilon =0$, and in many examples, this has been shown to coincide with Liouville integrability.

For a general birational map, performing an exact calculation of degree growth, in order to find the algebraic
entropy, can be very complicated. However, it was shown in \cite{2013} that the complexity of this calculation can be significantly reduced in the context of cluster maps, by looking at the tropical analogue of the exchange relation,  
and it was further shown how to find the exact degree growth of cluster maps associated with  period 1 exchange matrices. It turns out that for cluster variables, the total degree is controlled by their denominators, which are given by monomials in the initial cluster, and the maximum  growth rate of degrees of these monomials  
provides a way to determine the algebraic entropy. 

Let us consider a cluster algebra $\mathcal{A}(\vb{x},B)$ of rank $m$. Recall that due to the Laurent phenomenon (Theorem \ref{LP}), any cluster variable 
$\tilde{x}$ generated by a sequence of cluster mutations is expressed as a Laurent polynomial in the initial cluster, taking the form 
\begin{equation}\label{clvar}
    \tilde{x} = \frac{{{\cal N}}(\vb{x})}{\vb{x}^{\vb{d}}}
\end{equation}
where ${{\cal N}}(\vb{x})$ is a polynomial in the initial cluster $\vb{x}=(x_1,\ldots,x_m)$, 
whose  coefficients are integers
(in fact, positive integers  \cite{positiveLS}), not divisible by any of these $x_j$,  
and $\vb{x}^{\vb{d}} = \prod_{j=1}^mx_{j}^{d^{(j)}}$ with \textit{denominator vector} (also called d-vector) $\vb{d}= (d^{(1)},\dots, d^{(m)})^T$. 
To find an iterative relation between d-vectors, it is sufficient to use the fact that each cluster variable $\tilde{x}$ arising from the repeated application of a sequence of exchange relations  is a subtraction-free rational expression in the initial cluster $\vb{x}$. 
Thus, upon substituting ratios of the form \eqref{clvar} directly into the exchange relation \eqref{mu1}, rewritten as 
$$ 
\tilde{x}_k'
 \tilde{x}_k = \prod_{b_{jk}>0} \tilde{x}_j^{b_{jk}} + \prod_{b_{jk}<0} \tilde{x}_j^{-b_{jk}},
$$
and comparing the denominators on each side, it follows  that  $\vb{d}_k'$, the d-vector of the new cluster variable $\tilde{x}_k'$ obtained from a mutation $\mu_k$ applied to a seed $(\tilde{x}_1,\ldots,\tilde{x}_m)$, is related to $\vb{d}_j$ (the d-vectors of the cluster variables 
$\tilde{x}_j$) by 
the  relation 
\begin{equation}\label{maxplusexch}
    \vb{d}_{k}' + \vb{d}_{k} = \max \qty(\sum\limits_{\substack{i=1 \\ b_{ik} > 0}}^mb_{ik}\vb{d}_{i} \ , \ - \sum\limits_{\substack{i=1 \\ b_{ik} < 0}}^m b_{ik} \vb{d}_{i}), 
\end{equation}
which is precisely the tropical analogue of equation 
\eqref{mu1}, in terms of $(\max,+)$  algebra. (For the origin of this tropical relation, see equation (7.7) in \cite{fomin2006cluster}, as well as \cite{FominShapiroThurston}; and for the specific case of period 1 exchange matrices, see Proposition 3.2 in \cite{2013}). 

To relate this to the above notion of algebraic entropy, note that the (total) degree of a general cluster variable $\tilde{x}$ is given by $\max (\deg {\cal N} (\vb{x}),|\vb{d}|)$, where $|\vb{d}|$ denotes the sum of all components of its d-vector, and $\deg \varphi^n$, the degree of an iterate of a  cluster map $\varphi$, is the maximum of the degrees of all the cluster variables that make up its components. With more careful analysis, it can be shown that the numerator degrees $\deg {\cal N}$ have the same asymptotic growth rate as the d-vectors; but in any case, the d-vectors provide a lower bound on   
$\deg \varphi^n$, and hence on the entropy $\varepsilon$.
So henceforth, we will only be concerned with the growth of d-vectors.

\begin{exmp}[Algebraic entropy of a special Somos-7 recurrence] The following exchange matrix
\begin{equation}
    B = \mqty(0 & 1 & 0 & -1 & -1 & 0 & 1 \\ -1 & 0 & 1 & 1 & 0 & -1 & 0 \\ 0 & -1 & 0 & 1 & 1 & 0 & -1 \\ 1 & -1 & -1 & 0 & 1 & 1 & -1\\ 1 & 0 & -1 & -1 & 0 &1 & 0 \\ 0 & 1 & 0 & -1 & -1 & 0 & 1 \\ -1 & 0 & 1 & 1 & 0 & -1 & 0),
\end{equation}
 is cluster mutation periodic with period 1: $\mu_{1}(B) = \rho(B)$ for the cyclic permutation $\rho = (1,2,3,4,5,6,7)$; this gives rise to the cluster map $\varphi=\rho^{-1}\mu_1$, that is 
\begin{equation}\label{somos7clustermap}
\begin{split}
    \varphi : \,(x_1,x_2,x_3,x_4,x_5, x_6, x_7) \mapsto & ( x_2,x_3,x_4,x_5,x_{6}, x_7,x_8) \\
    & = ( x_2,x_3,x_4,x_5,,x_6,x_7,\frac{x_{2}x_{7} + x_{4}x_{5}}{x_{1}})
   \end{split}
\end{equation}
where we renamed the new variable $x_1'$ as $x_{8}$, appearing in the last position after applying a cyclic permutation . If we repeat the same procedure for each iteration of the cluster map $\varphi=\rho^{-1}\mu_1$, labelling the new cluster variable obtained with a shifted index, then we obtain a sequence of cluster variables $(x_n)$ that satisfies the nonlinear recurrence 
\begin{equation} \label{eq:somos-71}
    x_{n+7}x_{n} = x_{n+1}x_{n+6} + x_{n+3}x_{n+4}, 
\end{equation}
 which is  a particular type of \textit{Somos-7 recurrence}.   
  For each $n$, we can write 
  $$ 
   x_n = \frac{{{\cal N}_n}(\vb{x})}{\vb{x}^{\vb{d}_n}}, 
   $$
   where $\vb{x}=(x_1,\ldots,x_7)$ is the initial cluster, 
   and the d-vector 
  $\vb{d}_{n} = (d_{1,n},d_{2,n},\dots,d_{7,n})^T$, whose  integer components $d_{j,n}$ for $1\leq j\leq 7$ corresponding to the exponent of the initial cluster variable $x_j$ in the monomial $\vb{x}^{\vb{d}_n}$. Then the tropical analogue of \eqref{eq:somos-71} is 
\begin{equation}\label{maxexpressionsomos-71}
     \vb{d}_{n + 7} + \vb{d}_{n}  = \max \qty(\vb{d}_{n+6} + \vb{d}_{n+1}\ , \ \vb{d}_{n+4} + \vb{d}_{n+3} ). 
\end{equation}
We are interested in one specific solution of this (max,+) recurrence, corresponding to the initial data specified by the cluster $\vb{x}=(x_1,x_2,\ldots,x_7)$. Upon writing 
$
x_1= \frac{1}{x_1^{-1}}
$ 
we see that the numerator ${\cal N}_1=1$, and 
$\vb{d}_1 =(-1,0,0,0,0,0,0)$, and similarly 
${\cal N}_2=1$, and 
$\vb{d}_2 =(0,-1,0,0,0,0,0)$, etc. Thus we can combine the 7 
initial vectors into a matrix, as follows:
\begin{equation}\label{dinits} 
\Big(\vb{d}_1\,\, \vb{d}_2\,\, \vb{d}_3\,\,\vb{d}_4\,\,\vb{d}_5\,\,\vb{d}_6\,\,\vb{d}_7  \Big) = -\mathbf{I}, 
\end{equation}
i.e.\ minus the $7\times 7$ identity matrix. 
In order to solve the tropical recurrence, and hence determine the growth of the d-vector $\mathbf{d}_n$, we employ a convenient trick, which is to introduce the quantity 
\begin{equation}\label{Usub}
\vb{U}_{n} =  \vb{d}_{n+5} - \vb{d}_{n+3}-\vb{d}_{n+2} +\vb{d}_{n}. 
\end{equation}
Then the key observation is that, after rearranging the terms on each side, \eqref{maxexpressionsomos-71} can be rewritten as a recurrence of second order for $\vb{U}_n$, namely 
\begin{equation}\label{eq:maxsomos7u}
    \vb{U}_{n+2} + \vb{U}_{n} = \max(\vb{U}_{n+1},0) .
\end{equation}
It can be verified directly that the orbit of any initial 
pair $(\vb{U}_1,\vb{U}_2)$ 
has period 5 under iteration of \eqref{eq:maxsomos7u}.
In other words, the sequence  $(\vb{U}_{n})$ satisfies the relation 
$(\cS^{5} - 1)\vb{U}_{n} = 0$,  
where $\cS$ is the shift operator that sends $n \to n+1$. By combining this with the definition of $\vb{U}_n$ in 
\eqref{Usub}, rewritten as 
$\vb{U}_n =(\cS^5-\cS^3-\cS^2+1)\vb{d}_n$, we find that the d-vector sequence satisfies the linear recurrence relation 
\begin{equation}\label{dlin}
(\cS^{5} -1) (\cS^{3} -1)(\cS^{2} -1) \vb{d}_{n} = 0. 
\end{equation}
The characteristic equation for the above recurrence has 1 as a triple root $\cS$, and all other roots have modulus 1, so its general solution has the form 
%
\begin{equation}
\vb{d}_{n} = \vb{a} \,n^2 + O(n)
\end{equation}
for some constant vector $\vb{a} $. This constant is determined from the initial conditions: it follows from 
\eqref{dlin} that 
$ 
    \vb{d}_{n + 9} + \vb{d}_{n + 8} - \vb{d}_{n + 6} - \vb{d}_{n + 5} - \vb{d}_{n + 4} - \vb{d}_{n + 3} + \vb{d}_{n + 1} + \vb{d}_{n}= 60 \, \vb{a}.
$  
With the particular initial conditions \eqref{dinits}, each component of $\vb{d}_n$ is (up to shifting the index) a copy of the same sequence, that is  
$-1,0,0,0,0,0,0,1,1,1,1,2,2,3,3,3,4,5,5,\ldots $, from which any 10 consecutive terms are sufficient to fix the value 
$\vb{a}=\frac{1}{60}(1,1,1,1,1,1,1)^T$. 
The fact that the Somos-7 recurrence relation 
\eqref{eq:somos-71} is homogeneous means that all of the cluster variables $x_n$ can be assigned the same  homogeneous degree 1. Hence,  in this case  all of the numerators ${\cal N}_n(\vb{x})$ are homogeneous polynomials, of total degree one more than that of the corresponding monomial denominator, which 
is $|\vb{d}_n|=\sum_{j=1}^7 d_{j,n}$. 
This implies that the total degree of 
$x_n$, as a rational function of $\vb{x}$, is
$$ 
\deg x_n =\max (\deg{\cal N}_n(\vb{x}),|\vb{d}_n|) 
=|\vb{d}_n|+1\sim \frac{7}{60}n^2 \quad \mathrm{as} \quad n\to\infty. 
$$
Thus we see that $\frac{1}{n}\log\deg x_n\to 0 $, so 
the algebraic entropy of 
the Somos-7 cluster map \eqref{somos7clustermap} is zero. 
\end{exmp}

\begin{rem} The substitution \eqref{Usub} is the (max,+) analogue of the monomial transformation
 \begin{equation}\label{Somos5subs}
    u_{n} = \frac{x_{n+5}x_{n}}{x_{n+3}x_{n+2}},
\end{equation}
which reduces 
the recurrence \eqref{eq:somos-71}  to
\begin{equation}\label{reducedsomos7}
u_{n+1}u_{n-1} = u_{n} + 1. 
\end{equation}
The latter is the Lyness recurrence \eqref{Lynessrecurrence} with $a=b=1$; so the iterates of \eqref{reducedsomos7} are equivalent to those of \eqref{lyness5}. Hence, the fact that 
the orbits of \eqref{eq:maxsomos7u} all have period 5 is the tropical analogue of the Lyness 5-cycle. Applying the substitution \eqref{Somos5subs} to  \eqref{Lynessrecurrence} produces a Somos-7 recurrence with coefficients 
$a,b$, which has the same degree growth as 
\eqref{eq:somos-71}. This is consistent with the fact that the cluster map \eqref{Lynessclustermap} is Liouville integrable.  
\end{rem}

\subsection{Deformation of cluster mutations}\label{ss:deformations} 

In this section, we briefly review the most general 
deformations of cluster mutations which preserve the presymplectic form, introduced by the second author and Kouloukas \cite{hone2021deformations}. Let us consider a generalized  mutation 
$\tilde{\mu}_k$ applied to a cluster $\vb{x}=(x_1,\ldots,x_m)$, of the following form:
\begin{equation}\label{eq:6}
\tilde{\mu}_k: \qquad     x'_{j}=\begin{cases}
      x^{-1}_{k}f_{k}(M^{+}_{k},M^{-}_{k}), & \text{for}\ j=k 
      \\
      x_{j}, & \text{for} \ j \neq k
    \end{cases}
\end{equation}
where 
the exchange relation %
that produces the new variable $x_k'$ is defined by a differentiable 
function $f_{k} : \mathbb{C} \times \mathbb{C} \to \mathbb{C}$,   and 
\begin{equation}\label{Monomials}
    M^{+}_{k} = \prod^{N}_{i=1}x_{i}^{[b_{ik}]_{+}}, \qquad  M^{-}_{k} = \prod^{N}_{i=1}x_{i}^{[-b_{ik}]_{+}}, 
\end{equation}
with $[b]_+:=\max(b,0)$. 
Note that if $f_{k}(M^{+}_{k},M^{-}_{k}) = M^{+}_{k} + M^{-}_{k}$, then the mutation is the ordinary coefficient-free cluster mutation, defined by the usual exchange relation \eqref{mu1}.  In this setting, the key results about preservation of the presymplectic form 
are the following.

\begin{lm}
[\cite{hone2021deformations}] \label{deformth} Let the action of the generalized mutation $\tilde{\mu}_k$ be defined as $(B',\vb{x}') = \tilde{\mu}_{k}(B,\vb{x})$, where $B'$ is given in \eqref{matrixmu} and the components of the new cluster $\vb{x}'=(x_1',\ldots,x_m')$ is defined as in \eqref{eq:6}, then the symplectic form $\omega$ transforms covariantly under the action of mutation, i.e.\
\begin{equation}
    \sum\limits_{i<j} \frac{b'_{ij}}{x'_{i}x'_{j}}\dd x'_{i} \wedge \dd x'_{j}  =  \sum\limits_{i<j} \frac{b_{ij}}{x_{i}x_{j}}\dd x_{i} \wedge \dd x_{j}
\end{equation}
if and only if 
\begin{equation}\label{eq:7}
    f_{k}(M^{+}_{k},M^{-}_{k}) = M^{+}_{k}g_{k}\qty(\frac{M^{-}_{k}}{M^{+}_{k}})
\end{equation}
for some differentiable function $g_{k}:\mathbb{C} \to \mathbb{C}$. 
\end{lm}

%
\begin{thm}[\cite{hone2021deformations}]\label{thmsymp}
\label{thm:1}  Let $\tilde{\mu}_{i_{l}},\dotsc,\tilde{\mu}_{i_{2}}\tilde{\mu}_{i_{1}}$ be a sequence of $l$ generalized mutations, where each $\tilde{\mu}_{i_j}$ is of the form \eqref{eq:6}, 
for some function $f_{i_{j}}$ of the form \eqref{eq:7}, and such that
\[ \tilde{\mu}_{i_{l}}\cdots \tilde{\mu}_{i_{2}}\tilde{\mu}_{i_{1}}(B,\vb{x}) = (B,\tilde{\vb{x}}). \]
Then the map $\varphi\colon\, \vb{x} \mapsto \tilde{\vb{x}}$ given by the composition 
$\varphi=\tilde{\mu}_{i_{l}}\cdots \tilde{\mu}_{i_{2}}\tilde{\mu}_{i_{1}}$ preserves the log-canonical presymplectic form \eqref{sympform}, that is 
$\varphi^{*}\omega = \omega$. 
\end{thm}
The above result says that 
exchange matrices that are invariant under the action of a particular sequence of mutations give rise to generalized cluster maps that preserve the presymplectic form (\ref{eq:6}); and the same result holds for generalizations of cluster maps, in the sense of Definition \ref{cmap}, where the mutations $\mu_{i_j}$ are replaced by generalized mutations of the form \eqref{eq:6}. Our main interest is in how to choose the functions $f_k$ so that the generalized cluster map is integrable. Let us consider a couple of examples of deformed integrable maps.
\begin{exmp}[Integrable deformation of type $A_2$] \label{exmpdefmA2} The type $A_{2}$ quiver, as in Example \ref{QA2}, corresponds to the exchange matrix $B=B_{A_2}$ given by \eqref{A2mu1}. 
The matrix is invariant under the action of the sequence of mutations $\mu_{2}\mu_{1}$, i.e.\ $\mu_{2}\mu_{1}(B_{A_{2}}) = B_{A_{2}}$. This composition of mutations enables us define a cluster map $\varphi_{A_{2}} = \mu_{2}\mu_{1}$ that preserves the symplectic form
\begin{equation}\label{sympA2}
    \omega_{A_2} = \frac{1}{x_1 x_2}\dd x_1 \wedge \dd  x_2. 
\end{equation}
As the map satisfies the condition of Theorem \ref{deformth}, one can define a new symplectic map $\tilde{\varphi}=\tilde{\mu}_2\tilde{\mu}_1$, which  the composition of two deformed mutations in the direction $k=1,2$, given by 
\begin{equation}
    \tilde{\mu}_{k}(x_{k}) = x_{k}^{-1}M_{k}^{+} g_{k}\qty(\frac{M_{k}^{-}}{M_{k}^{+}}) 
\end{equation}
with $M_{k}^{\pm}$ as in \eqref{Monomials}.
By choosing linear functions $g_{k}(x) = a_{k} + b_{k}x$ for $k=1,2$, with non-zero coefficients, the mutated variables $ \tilde{\mu}_{k}(x_{k})$ can be written as 
\begin{equation} \label{eq:deformedmu}
    \tilde{\mu}_{k}(x_{k}) = x^{-1}_{k}\qty(b_{k}M_{k}^{-} + a_{k}M_{k}^{+})
\end{equation}
Also, by rescaling each of the cluster variables, $x_{i} \to \lambda_{i} x_{i}$ for $(\lambda_{1},\lambda_{2}) \in (\mathbb{C}^{*})^2$, we can rescale the parameters $b_{1},b_{2}$ and fix them both to be 1, without loss of generality. Thus we recover the deformed cluster map $\tilde{\varphi}_{A_{2}} = \tilde{\mu}_{2}\tilde{\mu}_{1}$, which coincides with the map \eqref{deformedA2}, whose singularity pattern was  considered  in Example \ref{exA2}. As noted before, since this preserves the symplectic form  \eqref{sympA2} and has the first integral \eqref{firstintA2}, this deformed map 
$\tilde{\varphi}_{A_{2}}$ is Liouville integrable.
\end{exmp}
\begin{exmp}[Integrable deformation of type $A_{4}$] \label{exmpA4} With all the arrows oriented in the same direction, going outward from the first, the   Dynkin quiver of type $A_{4}$ is  
\begin{center}
 \begin{tikzpicture}[every circle node/.style={draw,scale=0.6,thick},node distance=15mm]
  \node [draw,circle,fill=red!50,"$1$"] (a1) at (0,0) {};
  \node [draw,circle,fill=red!50,"$2$"] (a2) [right=of a1] {};
  \node [draw,circle,fill=red!50,"$3$" ] (a3) [right=of a2] {};
  \node [draw,circle,fill=red!50,"$4$"] (a4) [right=of a3] {};
  
  \begin{scope}[>=Latex]
  \draw[-> , thick] (a1) edge (a2);
  \draw[-> , thick] (a2) edge (a3);
  \draw[-> , thick] (a3) edge (a4);
    \end{scope}

\end{tikzpicture}
\end{center}
The corresponding skew-symmetric exchange matrix is 
%
\begin{equation}\label{exchA4}
     B_{A_{4}} =  \mqty(0 & 1 & 0 & 0  \\ -1 & 0 & 1 & 0  \\ 0 & -1 & 0 & 1\\ 0 & 0 & -1 & 0 ), 
\end{equation}
which is invariant under mutation at each of the vertices in turn, that is  
\begin{align*}
    \mu_{4}\mu_{3}\mu_{2}\mu_{1}(B_{A_{4}}) = B_{A_{4}}.
\end{align*}
Due to Zamolodchikov periodicity, every orbit of the corresponding cluster map $\varphi_{A_{4}} = \mu_{4}\mu_{3}\mu_{2}\mu_{1}$ is periodic with period 7.
Once again, by choosing linear functions $g_{k}$,  the associated deformed cluster mutations $\tilde{\mu}_k$  take the form \eqref{eq:deformedmu}. Furthermore, 
by exploiting the freedom to rescale each cluster variable, via $x_{i} \to \lambda_{i}x_{i}$  
for $1\leq i\leq 4$ with $(\lambda_{1},\lambda_{2},\lambda_3,\lambda_4) \in (\mathbb{C}^{*})^4$, 
the parameters can be adjusted so that $b_{2}=b_3=1$ and $a_{2}=a_3=1$. The resulting 
deformed map 
$\tilde{\varphi}_{A_{4}}=\tilde{\mu}_4\tilde{\mu}_3\tilde{\mu}_2\tilde{\mu}_1$ then depends on only 4 parameters, that is 
\begin{align*}
    \tilde{\varphi}_{A_{4}} :\quad  (x_{1},x_{2},x_{3},x_{4}) \to (x'_{1},x_{2}',x'_{3},x'_{4}),
\end{align*}
where the mutated variables $x'_{i}$ are given by the following relations: 
\begin{equation} \label{muA4}
\begin{aligned}
  &\tilde{\mu}_{1}: & x_1 x'_1 &= b_{1} +  a_{1}x_2 \\
  &\tilde{\mu}_{2}: &  x_2 x'_2 &=1 +  x'_1x_{3} \\
  &\tilde{\mu}_{3}: & x_3 x'_3 &= 1 + x'_2x_{4} \\
   &\tilde{\mu}_{4}: & x_4 x'_4 &= b_{4} + a_{4}x'_3 \\
   \end{aligned}
\end{equation}

The periodicity of the original  cluster map $\varphi_{A_{4}}$ 
provides a simple way to construct first integrals: 
the cyclically symmetric functions given by  
\begin{equation}\label{firstintA4}
\begin{split}
    &H_{1} = \sum_{j=0}^{6}L_{j}, \quad H_{2} = \prod_{j=0}^{6}L_{j}, 
\end{split}
\end{equation}
where $L_{i} = (\varphi^{*})^{i}(x_{1}) $, 
not only satisfy $\varphi_{A_{4}}^{*} (H_{i}) = H_{i}$ for $i=1,2$, but also 
$\{H_1,H_2\}=0$ with respect to the Poisson bracket corresponding to the symplectic form \eqref{sympform} associated with the nondegenerate exchange matrix \eqref{exchA4}, i.e.\ 
\begin{equation}\label{poissonbracketA4}
    \qty{x_{i},x_{j}} = P_{ij}x_{i}x_{j}
\end{equation}
where 
\begin{equation}\label{pmatrixA4}
   P= \mqty(0 & 1 & 0 & 1  \\ -1 & 0 & 0 & 0 \\ 0 & 0 & 0 & 1 \\ -1 & 0 & -1 & 0 ). 
\end{equation}
%
Hence the undeformed cluster map $\varphi_{A_{4}}$ is Liouville integrable. 

The approach adopted in 
\cite{hone2021deformations} is to find conditions on the parameters $a_k,b_k$ which ensure integrability of 
the deformed map $\tilde{\varphi}_{A_{4}}$. Since the latter preserves the same Poisson structure  \eqref{poissonbracketA4}, the key is 
to modify the first integrals \eqref{firstintA4}. 
The latter are replaced 
with linear combinations 
$$ 
\tilde{H}_{1} = \sum_{j}\alpha_{j}J_{j},  \quad \tilde{H}_{2} = \sum_{j}\beta_{j}K_{j}, 
$$ 
where $J_j,K_j$ denote 
the same Laurent monomials as appear in the original quantities $H_1,H_2$, respectively. 
The requirement that $\tilde{\varphi}_{A_{4}}^*\tilde{H}_i$ for $i=1,2$ puts conditions on the coefficients $\alpha_{j}$, $\beta_{j}$ which allows them to be completely determined, up to overall scale, subject to the necessary conditions 
$b_1=1=b_4$. Then the first integrals are 
\begin{equation}\label{A4ints}
\begin{split}
    \tilde{H}_{1} &= \frac{1}{x_{1}x_{2}x_{3}x_{4}} (a_{1}a_{4}x_{1}x_{2} + a_{1}a_{4}^2 x_{1}x_{2}x_{3} + a_{1}x_{1}x_{2}x_{3} + a_{1}a_{4}x_{1}x_{2}x_{3}^2 + a_{1}a_{4}x_{1}x_{4} + a_{1}a_{4}x_{1}x_{2}^2x_{4} \\ &+ a_{1}a_{4}x_{3}x_{4} + a_{1}a_{4}x_{1}^2x_{3}x_{4} + a_{4}x_{2}x_{3}x_{4} + a_{1}^2a_{4}x_{2}x_{3}x_{4} + a_{4}x_{1}^2x_{2}x_{3}x_{4} + a_{1}a_{4}x_{2}^2x_{3}x_{4} \\ &+ a_{1}a_{4} x_{1}x_{3}^2x_{4} + a_{1}a_{4}x_{1}x_{2}x_{4}^2 + a_{1}x_{1}x_{2}x_{3}x_{4}^2) \\[1em] 
    \tilde{H}_{2} & = \frac{(a_{1} + x_{2})(x_{1}+x_{3})(x_{1}+x_{3})(a_{4}+x_{3})(x_{1}x_{2} + a_{4}x_{1}x_{2}x_{3} + x_{1}x_{4} + x_{3}x_{4} + a_{1}x_{2}x_{3}x_{4})}{x_{1}x_{2}^2x_{3}^2x_{4}}, 
\end{split}
\end{equation}
and a direct calculation shows that 
$\{\tilde{H}_1,\tilde{H}_2\}=0$. Hence, with arbitrary $a_1,a_4$, the conditions $b_1=1=b_4$ are necessary and sufficient for the 
deformed map $\tilde{\varphi}_{A_{4}}$ to be  Liouville integrable. 
\end{exmp}

The examples above show that $\tilde{\varphi}_{A_{2}}$ and $\tilde{\varphi}_{A_{4}}$ are integrable symplectic maps. However, as a result of applying the deformation, the map no longer generates cluster variables, belonging to a Laurent polynomial ring. Thus in general, the deformed map is not a cluster map. To restore the property, we require a process called \textit{Laurentification}, which will now be introduced. 

\subsection{Laurentification}\label{ss:Laurentification}

In this subsection, we introduce specific 
lifts of the aforementioned deformed maps. 
The point is that, in their original coordinates, these deformed maps do not have a cluster algebra structure, since for generic parameter values their iterates are not given by Laurent polynomial expressions; however, by lifting to a higher-dimensional one which does generate Laurent polynomials, we obtain a cluster structure.   
This lifting is called \emph{Laurentification}, a name that was coined in \cite{hamad2014integrable}, 
although the process of using singularity patterns to obtain such lifts had been previously applied to a variety of difference equations, both integrable \cite{Ramani1995BilinearDP} and non-integrable  ones \cite{Hone_2007}. 
In our setting,  procedure helps us restore a feature of the original periodic maps that is lost through the deformation, namely the Laurent property.

Indeed, recall that one of the key features of a cluster algebra is the Laurent phenomenon, where every variable induced by cluster mutation can be expressed as a Laurent polynomial in the initial cluster variables. This implies that a cluster map, which is composed of certain mutations and permutations, has the \textit{Laurent property}, in the sense that all its iterates are Laurent polynomials, as described in Theorem \ref{LP}. 
\begin{rem}
Note that, for a map $\mathbb{C}^m\to \mathbb{C}^m$ with parameters, we can choose arbitrary values for the parameters in $\mathbb{C}$, and just require that the iterates should be Laurent polynomials with coefficients in $\mathbb{C}$. 
However, in what follows, we will show that we can regard the parameters in the deformed maps as being frozen variables, and an 
appropriate lift will produce Laurent polynomials over $\mathbb{Z}$, as in Theorem \ref{LP}. 
\end{rem}

Notice that the deformed map $\tilde{\varphi}_{A_{2}}$ in Example~\ref{exmpdefmA2} with generic choice of parameters $a_{1}, a_{2}$ cannot be a cluster map. This is because after two iterations of the map, beginning from the initial cluster $(x_{1},x_{2})$, the components are given by 
%
\begin{equation}
   (\tilde{\varphi}_{A_{2}})^{2}: \vb{x} \mapsto \mqty(\dfrac{a_1 a_2 + a_1 x_1 + a_1^2 a_2 x_2 + x_1 x_2}{x_2 (1 + a_1 x_2)} \\[1em] \dfrac{x_1 (a_1 a_2^2 + a_1 a_2 x_1 + x_2 + a_1^2 a_2^2 x_2 + a_2 x_1 x_2 + 
   a_1 x_2^2)}{(1 + a_1 x_2) (a_2 + x_1 + a_1 a_2 x_2)}\\  ) , 
\end{equation}
which consists of rational expressions whose denominator is no longer monomial, as the parameters prevent the cancellation with the numerator. Thus this deformation of the $A_2$ cluster map destroys the Laurent property that is present in its undeformed counterpart. 

To recover the Laurent property, one must try to lift the map to a higher-dimensional space, where the Laurent property is restored.

\begin{defn}[Laurentification] Let $\varphi \colon\, \mathbb{C}^{m} \to \mathbb{C}^{m}$ be a birational map.  A birational map $\psi \colon \mathbb{C}^{k} \to \mathbb{C}^{k}$ (for some $k\geq m$) is said to be a \emph{Laurentification} of $\varphi$ if the map $\psi$ is a lift that possesses the Laurent property. In other words, there exists a rational map $\pi : \mathbb{C}^{k} \to \mathbb{C}^{m}$ such that $\varphi \circ \pi = \pi \circ \psi$, and all of the iterates of $\psi$ are 
Laurent polynomials in the $k$ initial data.
\end{defn}
Note that there are several methods which fit the definition of Laurentification. One such, called recursive factorization, was introduced by Hamad and Kamp  \cite{hamad2014integrable}, who showed that certain QRT maps, whose iterates do not lie in a  Laurent polynomial ring, can be transformed into Somos-4 and Somos-5 recurrence relations with periodic coefficients, for which the Laurent property holds.

Here our approach is to find an appropriate rational map $\pi$ defined by dependent variable transformations, that is, expressed by new variables called \textit{tau functions}.  The transformations we use are suggested by the singularity confinement patterns induced by a deformed integrable map. This method was applied to several examples in \cite{kar64512}, \cite{hone2021deformations}. To see the approach in detail in small examples, let us consider the Laurentification of the deformed maps $\tilde{\varphi}_{A_{2}}$ and $\tilde{\varphi}_{A_{4}}$, which are shown in Example \ref{LaurentA2} and Example \ref{LaurentA4} as below.

\begin{exmp}[Laurentification of $\tilde{\varphi}_{A_{2}}$]\label{LaurentA2} 
    The singularity confinement pattern displayed by $\tilde{\varphi}_{A_{2}}$, as shown in Example~\ref{exA2}, suggests the introduction of  a rational map  
\begin{align*}
    \pi :\,  \qty(\tau_{-1},\tau_{0},\tau_{1}, \sigma_{0}, \sigma_{1}, \sigma_{2}, \sigma_{3}) \mapsto  \qty(x_{1},x_{2}) 
\end{align*}
which (for $n=0$)  is defined by the dependent variable transformation
\begin{equation}\label{vartrans}
    \begin{split}
    x_{1,n} = \frac{\sigma_{n}\tau_{n+1}}{\sigma_{n+1}\tau_{n}}, \quad x_{2,n} =\frac{\sigma_{n+3}\tau_{n-1}}{\sigma_{n+2}\tau_{n}}, 
    \end{split}
\end{equation}
where the index $n$ labels the iterates of the map $\tilde{\varphi}=\tilde{\varphi}_{A_{2}}$ in   \eqref{deformedA2},  
and the tau functions $\tau_n$ and $\sigma_n$ are associated with the singularity patterns  \eqref{singA21} and \eqref{singA22}, respectively. Substituting \eqref{vartrans} into the each iteration of \eqref{deformedA2} leads to a lift $\psi$ of the deformed map to the space of tau functions, 
such 
that 
each iteration of $\psi$ sends 
$$ 
\psi: \quad \qty(\tau_{n-1},\tau_{n},\tau_{n+1}, \sigma_{n}, \sigma_{n+1}, \sigma_{n+2}, \sigma_{n+3}) \mapsto\qty(\tau_{n},\tau_{n+1},\tau_{n+2}, \sigma_{n+1}, \sigma_{n+2}, \sigma_{n+3}, \sigma_{n+4}), 
$$
and $\pi \circ \psi = \tilde{\varphi} \circ \pi $. Then $\psi$ is defined by 
the following pair of bilinear equations:
\begin{subequations}
    \begin{equation}\label{exA21}
        \tau_{n+2}\sigma_{n} = \sigma_{n+2}\tau_{n} + a_{1}\sigma_{n+3}\tau_{n-1} 
    \end{equation}
    \begin{equation}\label{exA22}
        \sigma_{n+4}\tau_{n-1} = \sigma_{n+2}\tau_{n+1} + a_{2}\sigma_{n+1}\tau_{n+2}
    \end{equation}
\end{subequations}
If the relations above are to be regarded as exchange relations, then there must be initial data formed by an initial cluster and an exchange matrix. The initial clusters can be extracted from \eqref{exA21} and \eqref{exA22}, as follows. Let us denote 
\begin{equation}\label{seed1}
    \qty(\tau_{-1},\tau_{0},\tau_{1}, \sigma_{0}, \sigma_{1}, \sigma_{2}, \sigma_{3}) = \qty(\tilde{x}_{1},\tilde{x}_{1},\dots, \tilde{x}_{7})
\end{equation}
The presymplectic form $\tilde{\omega}_{A_2}$ on the space of tau functions can be obtained as the pullback 
\begin{equation}
    \tilde{\omega}_{A_{2}} = \pi^{*} \omega_{A_2} = \sum_{i<j} \tilde{b}_{ij} \dd \log \tilde{x}_i \wedge \dd \log \tilde{x}_j
\end{equation}
where the $\tilde{b}_{ij}$ are entries of a new exchange matrix, that is 
\begin{equation} \label{defmexchmA2}
    \tilde{B}_{A_2} = \mqty(0 & 1 & -1 & -1 & 1 & 0 & 0 \\ -1 & 0 & 1 & 1 & -1 & 1 & -1 \\ 1 & -1  & 0 & 0 & 0 & -1 & 1 \\ 1 & -1 & 0 & 0 & 0 & -1 & 1 \\ -1 & 1 & 0 & 0 & 0 & 1 & -1 \\ 0 & -1 & 1 & 1 & -1 & 0 & 0 \\ 0 & 1 & -1 & -1 & 1 & 0 & 0 )
\end{equation}
Next consider an extended initial cluster, with the coefficients $a_1,a_2$ added as a pair of frozen variables, namely  $\tilde{\vb{x}}=\qty(\tau_{-1},\tau_{0},\tau_{1}, \sigma_{0}, \sigma_{1}, \sigma_{2}, \sigma_{3},a_1,a_2)$, and introduce an extended  exchange matrix, given by 
\begin{equation} \label{exddefmexchmA2}
    \hat{B}_{A_2} = \mqty(0 & 1 & -1 & -1 & 1 & 0 & 0 \\ -1 & 0 & 1 & 1 & -1 & 1 & -1 \\ 1 & -1  & 0 & 0 & 0 & -1 & 1 \\ 1 & -1 & 0 & 0 & 0 & -1 & 1 \\ -1 & 1 & 0 & 0 & 0 & 1 & -1 \\ 0 & -1 & 1 & 1 & -1 & 0 & 0 \\ 0 & 1 & -1 & -1 & 1 & 0 & 0 \\ 0 & 1 & 1 & -1 & -1 & 0 & 0 \\ -1 & -1 & 0 & 0 & 0 & 1 & 1)
\end{equation}
The associated quiver may been found in Figure~\ref{QA2first}.

\begin{figure}[h!]
\begin{center}
\resizebox{0.6\textwidth}{!}{%
\begin{tikzpicture}[every circle node/.style={draw,scale=0.6,thick},node distance=15mm]
  \node [draw,circle,fill=red!50,"$5$"] (5) at (0,0) {};
  \node [draw,circle,fill=red!50,"$6$"] (6) [right=of 5] {};
  \node [draw,circle,fill=red!50,"$7$" below] (7) [below right=of 6] {};
  \node [draw,circle,fill=red!50,"$4$" below] (4) [below left=of 5] {};
    \node [draw,circle,fill=red!50,"$3$" below] (3) [below=of 4] {};
      \node [draw,circle,fill=red!50,"$1$" below] (1) [below=of 7] {};
      \node (a) [left=of 5]{};
      \node (b) [right=of 6]{};
      \node  [draw,circle,fill=blue!50,"$8$" left]  (8) [below left=of a]{};
      \node [draw,circle,fill=blue!50,"$9$" right] (9) [below right=of b]{};

 \node [draw,circle,fill=red!50,"$2$" below] (2) at (1,-4) {};
  
  \begin{scope}[>=Latex]
  \draw[-> , thick]  (5) edge (6);  
   \draw[-> , thick]  (4) edge (7); 
    \draw[-> , thick]  (3) edge (1); 
    
     \draw[-> , thick]  (4) edge (1); 
      \draw[-> , thick]  (1) edge (5); 
       \draw[-> , thick]  (1) edge (2); 
       
        \draw[-> , thick]  (7) edge (5); 
         \draw[-> , thick]  (7) edge (2); 
          \draw[-> , thick]  (4) edge (7); 
          \draw[-> , thick]  (3) edge (7); 
          
          \draw[-> , thick]  (6) edge (4); 
          \draw[-> , thick]  (6) edge (3); 
          \draw[-> , thick]  (2) edge (6); 
          
          \draw[-> , thick]  (5) edge (2); 
          \draw[-> , thick]  (2) edge (3); 
          \draw[-> , thick]  (2) edge (4); 
          
          \draw[-> , thick]  (8) edge (2); 
          \draw[-> , thick]  (8) edge (3); 
          \draw[-> , thick]  (4) edge (8); 
          \draw[-> , thick]  (5) edge (8); 
          
          \draw[-> , thick]  (9) edge (6); 
           \draw[-> , thick]  (9) edge (7);
            \draw[-> , thick]  (2) edge (9);  
             \draw[-> , thick]  (1) edge (9);

    \end{scope}

\end{tikzpicture}
}
\end{center}
\caption{The extended quiver $Q_{A_{2}}$ constructed through Laurentification.\label{QA2first}}
\end{figure}
For the above quiver, the cluster mutation in direction 4 acting on the initial clusters \eqref{seed1} produces a new cluster:
$$
   \hat{ \mu}_{4}:\qty(\tau_{-1},\tau_{0},\tau_{1}, \sigma_{0}, \sigma_{1}, \sigma_{2}, \sigma_{3},a_1,a_2) \mapsto \qty(\tau_{-1},\tau_{0},\tau_{1}, \tau_{2}, \sigma_{1}, \sigma_{2}, \sigma_{3},a_1,a_2),
$$ 
where the fourth component of the tuple sees $\sigma_{0}$ replaced by the new variable $\tau_{2}$, defined by the 
exchange relation
\begin{equation}\label{tauA21}
\begin{split}
 \tau_{2}\sigma_{0} = \tau_{-1} \sigma_{3} + a_{1}\tau_{0}\sigma_{2} .
\end{split}
\end{equation} (Here $\hat{\mu}_{j}$ denotes a cluster mutation in the cluster algebra associated with \eqref{exddefmexchmA2}, obtained by Laurentification of the deformed map.) 
Next, applying the mutation in the direction 1, that is  
$$
    \hat{\mu}_{1}:\qty(\tau_{-1},\tau_{0},\tau_{1}, \tau_{2}, \sigma_{1}, \sigma_{2}, \sigma_{3},a_1,a_2) \mapsto \qty(\sigma_4,\tau_{0},\tau_{1}, \tau_{2}, \sigma_{1}, \sigma_{2}, \sigma_{3},a_1,a_2), 
$$ 
produces the new cluster variable $\sigma_4$,  obtained via the exchange relation 
\begin{equation}\label{tauA22}
   \sigma_{4}\tau_{-1} = \sigma_{2} \tau_{1} + a_{2} \sigma_{1}\tau_{2} . 
\end{equation}
Notice that equations \eqref{tauA21} and \eqref{tauA22} are  just the bilinear relations \eqref{exA21} and \eqref{exA22} with $n=0$ respectively. Furthermore, we have that 
$\hat{\mu}_1\hat{\mu}_4 (\hat{B}_{A_2}) = \rho(\hat{B}_{A_2})$, where 
$\rho=(1,2,3,4,5,6,7)$ is a cyclic permutation of the nodes corresponding to the mutable (non-frozen) variables in the cluster $\tilde{\vb{x}}$. Then  we see that 
\begin{equation}\label{seqA2}
\rho^{-1} \hat{\mu}_1\hat{\mu}_4
(\tilde{\vb{x}},\hat{B}_{A_{2}}) = (\psi(\tilde{\vb{x}}),\hat{B}_{A_{2}}) ,
\end{equation} 
so $\psi=\rho^{-1} \hat{\mu}_1\hat{\mu}_4$ is a cluster map. Thus all of the tau functions $\sigma_n,\tau_n$ generated by iteration of $\psi$ are cluster variables, which implies that they are Laurent polynomials in the initial data. Hence the lift $\psi$ possesses the Laurent property, as required. 
\end{exmp}

\begin{exmp}[Laurentification of $\tilde{\varphi}_{A_{4}}$]\label{LaurentA4} 
Paraphrasing the calculations in \cite{hone2021deformations}, the empirical $p$-adic method described in Remark \ref{padic} can be carried out, in order to detect  
four types of singularity patterns 
for the map $\tilde{\varphi}_{A_{4}}=\tilde{\mu}_4\tilde{\mu}_3\tilde{\mu}_2\tilde{\mu}_1$ defined by \eqref{muA4}, as follows: 
\begin{align*}
    (1): \dots &\to \qty(\epsilon,R,R,R) \to \qty(\epsilon^{-1},\epsilon^{-1},\epsilon^{-1},\epsilon^{-1}) \to \qty(R,R,R,\epsilon) \to \dots \\[1em] 
   (2) : \dots &\to \qty(R,R,R,\epsilon) \to \qty(R,R,R,\epsilon^{-1}) \to \qty(R,R,\epsilon^{-1},R) \to \qty(R,\epsilon^{-1},R,R)  \\
     &\to \qty(\epsilon^{-1},R,R,R) \to \qty(\epsilon,R,R,R) \to \dots \\[1em] 
   (3): \dots & \to (R,\epsilon,R,R) \to \dots \\[1em] 
   (4): \dots & \to (R,R,\epsilon, R) \to \dots 
\end{align*}
where the limit $\epsilon\to 0$ in the patterns (3) and (4) corresponds to vanishing modulo primes which can be seen only in the components $x_{2,n}$ and $x_{3,n}$ of the $n$th iterate, respectively. 
Then we define a rational map $\pi$ such that  
\begin{align*}
    \pi: \, \tilde{\vb{x}} = (q_{0},\tau_{-1},\tau_{0},\tau_{1}, \sigma_{0},\sigma_{1},\sigma_{2},\sigma_{3},\sigma_{4},\sigma_{5},p_{0}) \mapsto \vb{x} = (x_{1},x_{2},x_{3},x_{4}), 
\end{align*}
which is defined by the case $n=0$ of the following dependent variable transformation: 
\begin{equation}\label{vartransA4}
    \begin{split}
        x_{1,n} = \frac{\sigma_{n}\tau_{n+1}}{\sigma_{n+1}\tau_{n}} \quad & x_{2,n} = \frac{p_{n}}{\sigma_{n+2}\tau_{n}} \quad  x_{3,n} =\frac{q_{n}}{\sigma_{n+3}\tau_{n}} \quad x_{4,n} =\frac{\sigma_{n+5}\tau_{n-1}}{\sigma_{n+4}\tau_{n}}
    \end{split}
\end{equation}
where tau functions $\sigma_{n}$, $\tau_{n}$, $p_{n}$, $q_{n}$ correspond to the appearance of a zero/pole in the singularity patterns $(1), (2), (3), (4)$,  respectively. With the addition of the 
frozen variables $a_{1}$ and $a_{4}$ to $\tilde{\vb{x}}$, we define a new initial seed $(\hat{\vb{x}},\hat{B}_{A_{4}})$, containing the extended cluster 
\begin{align*}
    \hat{\vb{x}} = (q_{0},\tau_{-1},\tau_{0},\tau_{1}, \sigma_{0},\sigma_{1},\sigma_{2},\sigma_{3},\sigma_{4},\sigma_{5},p_{0},a_{1},a_{4})
\end{align*}
and the 
exchange matrix $\hat{B}_{A_{4}}$ 
corresponding to the quiver depicted  in Figure \ref{fig:ladderA4}. 
  \begin{figure}[h]
\centering
\resizebox{0.6 \textwidth}{!}{%
 \begin{tikzpicture}[every circle node/.style={draw,scale=0.6,thick},node distance=15mm]

  \node [draw,circle,fill=blue!50,"$12$"] (12) at (0,0) {};
  
     \node [draw,circle,fill=red!50,"$6$"] (6) [right= of 12] {};
      \node [draw,circle,fill=red!50,"$7$"] (7) [right=of 6] {};
      \node [draw,circle,fill=red!50,"$8$"] (8) [right=of 7] {};
      \node [draw,circle,fill=red!50,"$9$"] (9) [right=of 8] {};
      \node [draw,circle,fill=red!50,"$10$"right] (10) [below right=of 9] {};
      \node [draw,circle,fill=red!50,"$5$" left] (5) [below left=of 6] {};

      \node [draw,circle,fill=blue!50,"$13$"] (13) [right=of 9] {};

       \node [draw,circle,fill=red!50,"$4$"below] (4) [below right=of 5] {};
       \node [draw,circle,fill=red!50,"$11$"below] (11) [right=of 4] {};
       \node [draw,circle,fill=red!50,"$1$"below] (1) [right=of 11] {};
       \node [draw,circle,fill=red!50,"$2$"below] (2) [right=of 1] {};
       
        \node [draw,circle,fill=red!50,"$3$" below] (3) at (4.4,-3.5) {};
       

  \begin{scope}[>=Latex]
  
  \draw[-> , thick]  (1) edge (2); 
  \draw[-> , thick]  (1) edge (7); 
  \draw[-> , thick]  (9) edge (1);
  \draw[-> , thick]  (11) edge (1);
   \draw[-> , thick]  (1) edge (10);

 \draw[-> , thick]  (2) edge (3);
 \draw[-> , thick]  (2) edge (8);
 \draw[-> , thick]  (2) edge (13);

 \draw[-> , thick]  (4) edge (11);
  \draw[-> , thick]  (12) edge (4);
   \draw[-> , thick]  (7) edge (4);
    \draw[-> , thick]  (3) edge (4);
    
 \draw[-> , thick]  (5) edge (12);
  \draw[-> , thick]  (5) edge (11);
   \draw[-> , thick]  (3) edge (5);
    \draw[-> , thick]  (7) edge (5);
   
   \draw[-> , thick]  (6) edge (12);
    \draw[-> , thick]  (6) edge (7);
     \draw[-> , thick]  (11) edge (6);
      \draw[-> , thick]  (6) edge[bend left= 15] (3);
      
    \draw[-> , thick]  (7) edge (8);
    
     \draw[-> , thick]  (8) edge (11);
      \draw[-> , thick]  (8) edge (9);
       \draw[-> , thick]  (10) edge (8);
       
        \draw[-> , thick]  (13) edge (9);
         \draw[-> , thick]  (3) edge[bend left=15] (9);
         
          \draw[-> , thick]  (13) edge (10);
           \draw[-> , thick]  (10) edge (3);
           
            \draw[-> , thick]  (3) edge (13);
            
             \draw[-> , thick]  (12) edge (3);
    
    \end{scope}
\end{tikzpicture}
}
\caption{Type $A_{4}$ deformed quiver  }
\label{fig:ladderA4}
\end{figure}
Then, with a suitable compostion of mutations 
$\hat{\mu}_j$ applied to this seed, the deformed map $\tilde{\varphi}_{A_{4}}$ is Laurentified by lifting it to the cluster map 
\begin{align*}
    \psi_{A_{4}}  =  \hat{\rho}_{A_{4}}^{-1} \hat{\mu}_{2}\hat{\mu}_{1}
    \hat{\mu}_{11}\hat{\mu}_{5}, \qquad \text{with} \quad \hat{\rho}_{A_{4}} =(2,3,4,5,6,7,8,9,10) ,
\end{align*}
%
such that $\pi\circ  \psi_{A_{4}}=\tilde{\varphi}_{A_{4}}\circ\pi$, 
whose iterates generate cluster variables obtained from the following bilinear relations:
\begin{equation}
    \begin{aligned}
 \tau_{n+2}\sigma_{n} & = \sigma_{n+2}\tau_{n} + a_{1}p_{n} \\
    p_{n+1}p_{n} & = \sigma_{n+3}\sigma_{n+2}\tau_{n}\tau_{n+1} + q_{n} \sigma_{n+1}\tau_{n+2} \\ 
    q_{n+1}q_{n} & = \sigma_{n+4}\sigma_{n+3}\tau_{n}\tau_{n+1} + p_{n+1}\sigma_{n+5}\tau_{n-1}\\
    \sigma_{n+6}\tau_{n-1} & = \sigma_{n+4}\tau_{n+1} + a_{1}q_{n+1}        
    \end{aligned}
\end{equation}
\end{exmp}
%


\section{Liouville integrability of periodic maps of type $A_{2N}$}
\label{s:type-A2N}
\setcounter{equation}{0}

In this section, we begin by introducing the properties of the periodic cluster map of type 
$A_{2N}$, and explain the connection with Coxeter's frieze patterns \cite{MGfriezes}. By 
considering the map from the Poisson point of view, we are able to adapt the argument from  \cite{Fordy_2011} (which concerned affine type algebras) in order to prove that this cluster map 
is Liouville integrable.


By taking an orientation of the  Dynkin diagram of type $A_{2N}$ so that all arrows point in the same direction, a quiver is obtained whose 
associated exchange matrix 
is 
\begin{equation}\label{exchA2N}
    B= \mqty(0 & 1 & 0 & 0 & 0  & \cdots & 0 \\ -1 & 0 & 1 & 0 & 0 & \cdots & 0 \\0 & -1 & 0 & 1 & 0 & \cdots & 0  \\ \vdots & \ddots & \ddots & \ddots & \ddots & \ddots &\vdots \\ 0 & \cdots & 0  & -1 & 0 & 1 & 0 \\ 0 & \cdots &  0 & 0 & -1 & 0 & 1 \\  0 & \cdots &  0 & 0 & 0 & -1& 0 \\ ) .
\end{equation}
This exchange matrix is mutation periodic with period $2N$ under the particular sequence of mutations $\mu_{2N}\mu_{2N-1}\cdots \mu_{2}\mu_{1}$. To begin with, 
our main object of interest is the cluster map 
\begin{equation}\label{muprod}
\varphi_{A_{2N}}= \mu_{2N}\mu_{2N-1}\cdots \mu_{2}\mu_{1}, 
\end{equation} 
which leaves the quiver invariant, and hence also the exchange matrix, so that 
\begin{equation} 
\label{Binv} 
\varphi_{A_{2N}}(B) = B, \end{equation} 
but has a nontrivial action on the cluster 
$\vb{x} = (x_j)$. For convenience, we drop the suffix, so in the rest of this subsection we set 
$\varphi_{A_{2N}}=\varphi$, and we have the 
birational map 
\begin{equation} \label{phiA2N} 
\varphi:\qquad \vb{x} \mapsto \vb{x}'
\end{equation}
defined by the following exchange relations: 
\begin{equation}\label{mueq1}
\begin{aligned}
    x'_{1}x_{1} &= 1 + x_{2},\\ 
    x'_{2}x_{2} &= 1 + x'_{1}x_{3},\\ 
    x'_{3}x_{3} &= 1 + x'_{2}x_{4},\\ 
    & \vdots\\
    x'_{2N-1}x_{2N-1} &= 1 + x'_{2N-2}x_{2N},\\ 
    x'_{2N}x_{2N} &= 1 + x'_{2N-1} .\\
    \end{aligned}
\end{equation}
%


It is well known that the map (\ref{phiA2N}) 
exhibits  
Zamolodchikov periodicity. In the case of 
Y-systems associated with a product  of Dynkin diagrams of finite type, the period is given by the sum of the (dual) Coxeter numbers, which was originally proved in \cite{volkov} for the case of $A_r \times A_s$. For the case at hand, we have $r=2N$, $s=1$, so the Coxeter numbers are 
$h=2N+1$, $h=2$, respectively, giving the period $p=h+h'=2N+3$. However, rather than appealing to the connection with Y-systems in this more general setting, there is a more direct way to 
obtain the periodicity of the map $\varphi$, based on Coxeter's frieze patterns (see \cite{MGfriezes} for a review). 

The sequence of exchange relations 
\eqref{mueq1} defines the first part of a 
frieze of type $A$, of width $2N$ and period $2N+3$.  The shape of the frieze pattern is 
$$ \begin{array}{cccccccccccc}
0 & & 0 & & 0 & & 0 & & 0 & & 0 & \cdots \\ 
& 1 &  & 1 & & 1 & & 1 & & 1 & \cdots & \\
& &  x_1 & & x_1' & & x_1'' & & x_1''' & \cdots && \\
& & &  x_2 & & x_2' & & x_2'' & \ddots & & &\\
& & & &  x_3 & & \ddots && \ddots &&& \\ 
&&&&& \ddots && \ddots &&&& \\
&&&&&& x_{2N} && x_{2N}' && \cdots & \\ 
&&&& \cdots & 1 && 1 && 1 && \cdots \\
&&&&& \cdots & 0 && 0 && 0 & \cdots 
\end{array}
$$
Note that this is an $SL_2$ frieze, so that 
we have the determinant condition 
$ad-bc=1$ 
on each $2\times 2$ diamond 
$\begin{array}{ccc} & b & \\ 
a & & d \\ 
& c& \end{array} $ 
in the pattern.

The $SL_2$ relation defines the action of the map \eqref{phiA2N} on each diagonal in the frieze, sending $\vb{x}\mapsto\vb{x}'$ from the first to the second diagonal shown on the left above, then  $\vb{x}'\mapsto\vb{x}''$ mapping the second diagonal to the third, and so on. 
The manifestation of Zamolodchikov periodicity here is the fact that these diagonals repeat every $2N+3$ steps \cite{cc1, cc2}. Hence, the map \eqref{phiA2N} defined by \eqref{mueq1} is periodic with period $2N+3$, i.e.\ $\varphi^{2N+3} =\mathrm{id}$. As a particular numerical example, here we present the non-zero part of the $A_4$ frieze ($N=2$) obtained by fixing the  
values of the initial cluster variables as 
$\vb{x} = (x_1,x_2,x_3,x_4) = (1,1,1,1)$: 
$$ \begin{array}{cccccccccccccccccccc}
\cdots & 1 &  & 1 & & 1 & & 1 & & 1 && 1 && 1 && 1 && 1&\cdots \\
&& 5 && 1 && 2 && 2 && 2 && 2 && 1 && 5 && \\
&&& 4 && 1 && 3 && 3 && 3 && 1 && 4 &&& \\ 
&&&& 3 && 1 && 4 && 4 && 1 && 3 &&&& \\ 
&&&&& 2 && 1 && 5 && 1 && 2 &&&&& \\ 
&&&&&\cdots & 1 && 1 && 1 && 1 &\cdots&&&&& 
\end{array}
$$
By continuing the above frieze to the left or right using the $SL_2$ rule, it is easily verified that the pattern repeats with period 7. 


For what follows, we need to consider the sequence consisting of the cluster variables which appear on the first nontrivial row of the frieze 
(immediately below the top row of 1s). These are  generated by the action of 
$\varphi$ on $x_1$. So we set 
$L_j = (\varphi^{*})^{j}x_{1}$, and find  
\begin{equation}\label{Lsfunctions}
\begin{split}
    &L_{0} = x_{1}, \quad L_{1} = \frac{1+x_2}{x_1}, \quad L_{2} = \frac{x_{1} + x_{3}}{x_{2}}, \ \cdots, \ L_{2N-1} = \frac{x_{2N-2} + x_{2N}}{x_{2N-1}}, \quad L_{2N}= \frac{1 + x_{2N-1}}{x_{2N}}, \\ &L_{2N+1}= x_{2N},\\
    & L_{2N+2} = \frac{\prod^{2N-1}_{i=1}x_{i} + \prod^{2N-2}_{i=1}x_{i} + \qty(\qty(\prod^{2N-3}_{i=1}x_{i})+ \qty(\prod^{2N-4}_{i=1}x_{i} + \cdots \qty(1+x_2)x_{3})\cdots)x_{2N-1})x_{2N}}{\prod^{2N}_{i=1}x_{i}},  \\
\end{split}
\end{equation}
while from one more application of $\varphi$ to  $L_{2N+2}=x_{2N}'$ we have 
$$ 
L_{2N+3} = (\varphi^*)^{2N+3} x_1 = x_1 = L_0, 
$$
by Zamolodchikov periodicity. 
The fact that $L_{2N+1}=x_{2N}$ implies that these same periodic quantities $L_j$ appear in sequence along the last nontrivial row of the frieze pattern (above the bottom row of 1s). 

The formulae \eqref{Lsfunctions} define a rational map 
\begin{equation}\label{chimap} 
\begin{array}{lrcl}
\chi: & \mathbb{C}^{2N} & \longrightarrow & \mathbb{C}^{2N+3} \\ 
    & \vb{x} & \longmapsto & \vb{L} = (L_0,L_1,\ldots,L_{2N+2}) . 
\end{array} 
\end{equation}
To consider the image of the map $\chi$, it is convenient to write the formulae defining the $L_j$ as a system 
of linear equations, namely 
\begin{equation}\label{Llinsys}
\begin{array}{rcl} 
x_1 - L_0 & = & 0, \\ 
x_2 -L_1 x_1 +1 & = & 0 \\ 
& \cdots & \\ 
x_{2N} - L_{2N-1}x_{2N-1} +x_{2N-2} & = & 0, \\ 
-L_{2N}x_{2N} +x_{2N-1}+1 & = & 0, \\ 
x_{2N} - L_{2N+1} & = & 0  
\end{array} 
\end{equation}
(where we have omitted an equation for $L_{2N+2}$). Regarding the first $2N$ of these equations as an inhomogeneous linear system for the $x_j$, we can easily find the solution by Cramer's rule, expressing the initial cluster variables in terms of the first $2N$ of the $L_j$, as 
\begin{equation}\label{xfromL} 
x_j = (-1)^j \left| 
\begin{array}{ccccc}
-L_0 & 1 & && \\ 
1 & -L_1 & 1 && \\ 
 & \ddots & \ddots & \ddots & \\ 
 && 1 & -L_{j-2} & 1 \\ 
 &&& 1 & -L_{j-1}
\end{array}
\right| , 
\end{equation}
where all omitted matrix entries are zero. 
(This is equivalent to Coxeter and Conway's 
formula in \cite{cc2}, up to changing signs along the diagonal and replacing $L_j\to a_j$.) 

We can further adjoin the next equation in the list \eqref{Llinsys}, and consider it now as being a homogeneous linear system, extending $\vb{x}$ to a vector in $\mathbb{C}^{2N+1}$ by adding an initial 1 at the top (corresponding to the top row of 1s in the frieze). Thus we obtain the polynomial relation 
$E_0 (\bf L)= 0$, where for each $j\in\mathbb{Z}$ we can define the $(2N+1)\times(2N+1)$ determinant 
\begin{equation}\label{Ejdet}
    E_j := 
    \left| 
\begin{array}{ccccc}
-L_j & 1 & && \\ 
1 & -L_{j+1} & 1 && \\ 
 & \ddots & \ddots & \ddots & \\ 
 && 1 & -L_{j+2N-1} & 1 \\ 
1 &&& 1 & -L_{j+2N}
\end{array}
\right| . 
\end{equation}
Instead, if we omit the first equation in the list \eqref{Llinsys}, and consider only the final $2N+1$ linear relations, then we have another 
homogeneous equation for the vector $\vb{x}$  extended with a 1 at the end, which gives 
$E_1(\vb{L}) =0$. Furthermore, by the action of $\varphi$, we have $\varphi^* E_j = E_{j+1}$,  
and hence we find the $2N+3$ polynomial relations 
\begin{equation}\label{affineq}
    E_j (\vb{L}) =0, \qquad j=0,\ldots, 2N+2. 
\end{equation}
However, the first three of these relations define $L_{2N}$, $L_{2N+1}$ and $L_{2N+2}$ 
in terms of $2N$ arbitrary values of $L_0,\ldots,L_{2N-1}$; so there are only three independent relations, which define the image of the map $\chi$ as an affine variety in $\mathbb{C}^{2N+3}$. 

The linear relations \eqref{Llinsys} are also naturally extended to give the solutions of a
 second order linear difference equation with periodic coefficients, satisfied by infinite sequences $({v}_j)_{j\in\mathbb{Z}}$, namely  
\begin{equation}\label{sturml}
    v_{j+1} - L_j \,v_j + v_{j-1} = 0,  
\end{equation}
which can be regarded as a discrete Sturm--Liouville equation \cite{henry}. Indeed, extending the initial cluster $\vb{x}$ along a diagonal of the frieze by the 1s and 0s at either end, we see that a consistent solution of this system is obtained by setting 
\begin{equation}
\label{sol1}
v_{-1}=0, \qquad v_0=1, \qquad  v_j = x_j \quad \mathrm{for}\quad j=1,\ldots, 2N, \qquad v_{2N+1}=1, 
\end{equation}
with the rest of the sequence being determined by the antiperiodic boundary condition 
\begin{equation}\label{antip} 
v_{j+2N+3} = - v_j, \qquad j\in\mathbb{Z}. 
\end{equation}
Moreover, from the action of $\varphi$ we see that $$\varphi^*v_{j+1}-L_{j+1}\, \varphi^* v_j +
\varphi^*v_{j-1}=0,$$ 
which provides another, linearly independent solution of \eqref{sturml} by setting 
$v^\dagger_j=\varphi^* v_{j-1}$, so that 
\begin{equation}
\label{sol2}
v^\dagger_{0}=0, \qquad v^\dagger_1=1, \qquad  v^\dagger_{j+1} = x_j' \quad \mathrm{for}\quad j=1,\ldots, 2N, \qquad v^\dagger_{2N+2}=1, 
\end{equation}
which extends to a sequence $(v_j^\dagger)_{j\in\mathbb{Z}}$ satisfying the same antiperiodic boundary condition.  

In order to obtain a polynomial identity involving all $2N+3$ of the $L_j$,  we take the solution \eqref{sol1} and see that the equation \eqref{sturml} becomes a trivial relation when $j=2N+2$. Adjoining this extra relation to the list \eqref{Llinsys} produces a linear system of size $2N+3$, which can be interpreted as a homogeneous equation for the vector 
$ \vb{v}_0 = (v_0,v_1,\ldots, v_{2N+2})^T
=(1,x_1,x_2,\ldots,x_{2N}, 1,0)^T$, that is 
${\bf M}_0 \, \vb{v}_0 =\mathbf{0}$, with the 
matrix 
\begin{equation}\label{M0}
 {\bf M}_0 = \left(
 \begin{array}{ccccc}
-L_0 & 1 & && -1 \\ 
1 & -L_{1} & 1 && \\ 
 & \ddots & \ddots & \ddots & \\ 
 && 1 & -L_{2N+1} & 1 \\ 
-1 &&& 1 & -L_{2N+2}
\end{array}
 \right) .
\end{equation}
The desired polynomial relation, which we will make use of shortly, is the vanishing of the determinant of this matrix, that is 
\begin{equation}\label{Mdet}
\mathrm{det}\,{\bf M}_0 =0.     
\end{equation}
In fact, this matrix (generically) has a two-dimensional kernel, since 
${\bf M}_0\vb{v}^\dagger_{0}=\mathbf{0} $, 
where 
$ \vb{v}_0^\dagger = (v_0^\dagger,v_1^\dagger,\ldots, v_{2N+2}^\dagger)^T
=(0, 1,x_1',x_2',\ldots,x_{2N}',1)^T$.

We now wish to consider the map \eqref{chimap} from a Poisson point of view, which will lead to a proof of Liouville integrability of the cluster map $\varphi$.  Since $\varphi=\varphi_{A_{2N}}$ is given by the composition of mutations 
\eqref{muprod}, 
and the 
symplectic form (\ref{sympform}) constructed from the exchange matrix \eqref{exchA2N} 
is covariant under the action of mutations, it follows from the invariance property \eqref{Binv} 
that 
$$ 
\varphi^* \om =\om, 
$$
or in other words the cluster map $\varphi$ is a symplectic map. Equivalently, $\varphi$ preserves the log-canonical Poisson  bracket 
\begin{equation}\label{poissA2N}
    \qty{x_{i},x_{j}} = P_{ij} \, x_{i}x_{j}, \qquad P_{ij} = (B^{-1})_{ij}, 
\end{equation} 
which is specified by 
$$ 
P_{2k,2j+1}=-P_{2j+1,2k} = 1 \qquad \mathrm{for} 
\quad k=j+1,j+2,\ldots, 2N, 
$$
with all other matrix entries being zero.

%
 %

Using \eqref{poissA2N}, we can compute the brackets between the functions $L_j$ defined by \eqref{Lsfunctions}. 
To begin with, we have 
$$\qty{L_0, L_1 } = 
\qty{x_1, \frac{x_2+1}{x_1} } = \frac{1}{x_1} \, \qty{x_1,x_2} =-x_2, 
$$
which we can rewrite in terms of $L_0$ and $L_1$, and the calculation of the other brackets 
$\qty{L_0,L_j}$ can be performed similarly, to yield the relations 
\begin{equation}\label{rel1}
\begin{aligned}
    \qty{L_{0},L_{1}} &= - L_{0}L_{1} +1 &&\\
    \qty{L_{0},L_{2j}} &=  L_{0}L_{2j} && \text{for} \ 1\leq j \leq N \\
    \qty{L_{0},L_{2j+1}} &= - L_{0}L_{2j+1} && \text{for} \ 1\leq j \leq N .  
    \end{aligned}
\end{equation}
To find further relations, we can use the fact that $\varphi$ is a Poisson map, and hence 
$$ 
\varphi^* \qty{ L_j,L_k } = \qty{\varphi^* L_j , \varphi^* L_k} = \qty{L_{j+1},L_{k+1}}, 
$$
so proceeding by induction it is straightforward 
to show that the brackets between these quantities are completely specified by 
\begin{equation}\label{Lalg}
\qty{L_j,L_k} 
=(-1)^{j+k}L_jL_k + 
\delta_{k-j-1,0}-\delta_{k-j-2N-2,0} \qquad 
\mathrm{for}\quad 0\leq j<k\leq 2N+2. 
\end{equation}

The formulae \eqref{Lalg} can be regarded as defining a Poisson bracket $\{\ ,\ \}_{\vb{L}} $ on $\mathbb{C}^{2N+3}$, which has the form of a sum  
$\{\ ,\ \}_{\vb{L}} = \{\ ,\ \}_2 + \{\ ,\ \}_0 $ 
of two pieces of homogeneous degrees 2 and 0, respectively. The brackets $\{\ ,\ \}_2$ and 
$\{\ ,\ \}_0$ are compatible: they 
form a bi-Hamiltonian pair,  in the sense that 
each piece separately satisfies the Jacobi identity, as does their sum (or any linear combination). Equivalently, we can write the Poisson bivector defined by the right-hand side  
of \eqref{Lalg} as a sum 
$$ 
\vb{P}^{(\vb{L})}=\vb{P}^{(2)} + \vb{P}^{(0)}, 
$$
where each of the homogeneous parts 
$\vb{P}^{(2)},\vb{P}^{(0)}$ is itself a Poisson 
bivector. 

The preceding observations 
show that the rational map (\ref{chimap}) is also a map of Poisson varieties, which satisfies the 
requirement to be a Poisson morphism, that is 
\begin{equation}\label{pmorphism}
\qty{\chi^* f,\chi^* g } = \chi^* \qty{ f, g}_{\vb{L}}   
\end{equation}
for any pair of regular functions $f,g$ on 
$\mathbb{C}^{2N+3}$. 
However, this is not a regular morphism, because in $\mathbb{C}^{2N}$ the expressions \eqref{Lsfunctions} become singular on the hyperplanes $x_j=0$. Restricting $\chi$ to the complement of these hyperplanes, we have the following result. 

\begin{lm}\label{PoissLA2N} 
Let $\{\ ,\ \}_{\vb{L}}$ be the Poisson bracket 
\begin{equation}\label{pois1}
    \qty{L_{j},L_{k}}_{\vb{L}} 
    =\qty{L_{j},L_{k}}_{2}+\qty{L_{j},L_{k}}_{0} 
\end{equation}
on $\mathbb{C}^{2N+3}$ given 
by the right-hand side  
of \eqref{Lalg}. 
Then (the restriction of) the map $\chi$ defined 
by  
the expressions \eqref{Lsfunctions} is a Poisson 
morphism from 
$\big( (\mathbb{C}^*)^{2N}, \{\ ,\ \}\big) $ 
to 
$\big( \mathbb{C}^{2N+3}, \{\ ,\ \}_{\vb{L}}\big) $.
\end{lm}

We now recall some arguments from 
\cite{Fordy_2011}, further developed in \cite{2013}, where  (up to rescaling) the same Poisson brackets \eqref{Lalg} 
were found to arise for 
periodic quantities 
appearing as coefficients in linear relations for cluster algebras of affine type 
$A^{(1)}_{2N+3}$.  
Moreover, this is a 
particular instance of the  Poisson algebra 
appearing in the (odd) dressing chain for 
Schr\"odinger operators, a well known integrable 
system \cite{Veselov1993DressingCA}. The 
bi-Hamiltonian nature of this algebra means that 
the standard machinery of the Lenard--Magri scheme can be applied \cite{Magri1978ASM, Olver1993},  leading to the construction of a sequence of 
Poisson-commuting first integrals ${\cal I}_j$, 
satisfying the recursive chain of relations 
\begin{equation} \label{MLR}
    \vb{P}^{(2)}\iprod \dd {\cal I}_{k+1} = \vb{P}^{(0)}\iprod \dd {\cal I}_{k}.
\end{equation} 

In the case at hand, the Lenard--Magri chain begins with the Casimir of $\Pbkt{\ }{\ }_2$, that is 
$$
{\cal I}_0 = \prod_{j=0}^{2N+2}L_j, 
$$
and ends with  the Casimir of $\{\ ,\ \}_0$, namely 
$$ 
{\cal I}_{N+1} = \sum_{j=0}^{2N+2}L_j. 
$$
The combined bracket $\Pbkt{\ }{\ }_{\vb{L}}$ 
on $\mathbb{C}^{2N+3}$ is also degenerate, of rank $2N+2$, with a single polynomial Casimir function given (up to scale) by 
\begin{equation}\label{casK} 
{\cal K} = \sum_{k=0}^{N+1} (-1)^k\,{\cal I}_k,  
\end{equation}
where the sum breaks up into  
components with different odd homogeneous degrees:
$$\deg {\cal I}_k=2N-2k+3.$$
The recursion 
\eqref{MLR} then results by extracting the relations of different homogeneity from the identity 
$$ \vb{P}^{(\vb{L})} \iprod \dd {\cal K}=0,$$ 
with the first and last relations being 
$$ 
\vb{P}^{(\vb{2})} \iprod \dd {\cal I}_0=0,
\qquad 
\vb{P}^{(\vb{0})} \iprod \dd {\cal I}_{N+1}=0, 
$$
corresponding to the Casimirs 
of $\{\ ,\ \}_2$ and $\{\ ,\ \}_0$, respectively. 
Using this recursion, a standard argument then shows that all of the quantities ${\cal I}_k$ Poisson commute with one another with respect to both  $\{\ ,\ \}_2$ and $\{\ ,\ \}_0$, and hence also 
with respect to their sum. 

It is convenient to paraphrase the main observations from \cite{Fordy_2011, 2013} about the bracket $\{\ ,\ \}_{\vb{L}}$, including an explicit expression for ${\cal K}$, as follows.  
\begin{thm}\label{thm1.1}
The Poisson bracket $\{,\}_{\vb{L}}$ on 
$\mathbb{C}^{2N+3}$ defined by the right-hand side of \eqref{Lalg} has rank $2N+2$, with the Casimir \eqref{casK} given explicitly by 
\begin{equation}\label{expK}
{\cal K} = 
\prod_{i=0}^{2N+2}\left(1-\frac{\partial^2}{\partial L_i\partial L_{i+1}}\right)\,\prod_{j=0}^{2N+2}L_j
\end{equation}
(with indices read mod $2N+3$). The $N+2$ homogeneous components of ${\cal K}$ are in involution with one another, that is 
$$ 
\{ {\cal I}_j,{\cal I}_k\}_{\vb{L}} =0 
$$
for all $j,k$.
\end{thm}
\begin{rem}
The formula for ${\cal K}$ in \eqref{expK} is the special case $\lambda=-1$ of a more general formula for the 
polynomial 
$$ 
{\cal P} (\lambda) = \tr \mathbf{L}(\lambda) = 
\sum_{k=0}^{N+1} \lambda^k\,{\cal I}_k, 
$$
where (up to conjugation) ${\mathbf L}$ is given by the matrix product 
\begin{equation}\label{mdy}
\mathbf{L}(\lambda) = 
\left(\begin{array}{cc} L_0 & 1 \\ 
\lambda & 0 \end{array}\right) \left(\begin{array}{cc} L_1 & 1 \\ 
\lambda & 0 \end{array}\right) \cdots \left(\begin{array}{cc} L_{2N+2} & 1 \\ 
\lambda & 0 \end{array}\right). 
\end{equation} 
For a proof see Lemma 4.3 in 
\cite{honeward}.  
\end{rem}
\begin{cor}\label{chiI}
Under the map \eqref{chimap}, the pullbacks of the coefficients of the polynomial ${\cal P} (\lambda)$, which are the homogeneous components of $\cal K$, provide a set of first integrals in involution for the $A_{2N}$ cluster map $\varphi$, that is 
$$ 
\{ \chi^*{\cal I}_j,\chi^*{\cal I}_k\} =0 
$$
for all $j,k$.
\end{cor}
\begin{proof} 
By Zamolodchikov periodicity, the action of $\varphi$ on 
the quantities $L_k=(\varphi^*)^k x_1$ sends 
$L_j\to L_{j+1}$, with indices read modulo the period $p=2N+3$. Each of the coefficients of ${\cal P} (\lambda)$ is a   cyclically symmetric homogenous polynomial in the $L_k$,  hence from \eqref{Lsfunctions} the pullbacks $\chi^* {\cal I}_j$ are rational functions of the cluster variables $x_k$ that are first integrals for the map $\varphi$. 
Since $\chi$ is a Poisson morphism, the involutivity then follows immediately from Theorem \ref{thm1.1}.
\end{proof}

The Liouville integrability of each  $A_{2N}$ cluster map $\varphi$ would now appear to be obvious, given that we have a large set of first integrals in involution. However, 
there is a slight subtlety here, because the  Poisson bracket \eqref{poissA2N} %
on $\C^{2N}$ is nondegenerate, so does not admit more than $N$ independent functions in involution. Hence, after pulling back the $N+2$ polynomials ${\cal I}_k$, they can no longer be functionally independent, and there must be at least two relations between them; moreover, for Liouville integrability, exactly $N$ independent functions are required. There is one relation between the first integrals that takes the same form for all $N$, and follows from  the fact that $\cal K$ is a Casimir of the bracket $\{\ ,\ \}_{\bf L}$, while the bracket $\{\ ,\ \}$ has no non-trivial Casimir functions, so $\chi^*\cal K$ must be a constant. A second relation between these invariants has a form that depends explicitly on $N$. 

\begin{lm} \label{trdet}
The image of 
the map $\chi$ lies on the level set ${\cal K}=-2$, fixing the value of the Casimir of 
the bracket $\{\ ,\ \}_{\bf L}$. Equivalently, the pullbacks of the 
functions ${\cal I}_k$ satisfy the relation 
\begin{equation}\label{trm1}
\sum_{k=0}^{N+1} (-1)^k\,\chi^*{\cal I}_k = -2. 
\end{equation} 
In addition, another independent  linear combination of their pullbacks is constant, namely 
\begin{equation}\label{trm2}
\sum_{k=1}^{N+1} k(-1)^{k-1}\,\chi^*{\cal I}_k = 2N+3. 
\end{equation} 
\end{lm}
\begin{proof} 
The linear relation \eqref{sturml} can be rewritten in matrix form as 
$$ 
(v_{j+1}\,\, v_j ) = (v_j\,\, v_{j-1} )\left(\begin{array}{cc} L_j & 1 \\ 
-1 & 0 \end{array}\right).
$$
Starting with $j=0$ and combining $p=2N+3$ iterations of this relation, we obtain 
$$ 
(v_{2N+3}\,\,\, v_{2N+2} ) = (v_0\,\,\, v_{-1} )\, {\bf L}(-1), 
$$
where ${\bf L}(\lambda)$ is the matrix product in 
\eqref{mdy}. The fact that \eqref{sturml} has the particular solution \eqref{sol1} satisfying the antiperiodic boundary condition \eqref{antip} implies that 
$(1\,\, 0 )$ is a left eigenvector of ${\bf L}(-1)$ with eigenvalue $-1$, while the other solution \eqref{sol2} provides a second, independent eigenvector with the same eigenvalue. Hence 
\begin{equation}
\label{lm1}    
{\bf L}(-1) =  - \mathbf{I}  = \left(\begin{array}{cc} -1 & 0 \\ 
0 & -1  \end{array}\right), 
\end{equation}
from which it follows that ${\cal K}={\cal P}(-1)=\tr\mathbf{L}(-1)= -2$ when it is pulled back under $\chi$ to $(\C^*)^{2N}$, 
the phase space of the cluster map $\varphi$; this proves the relation \eqref{trm1}.  For the second relation,  note that the matrix product \eqref{mdy}  has the general structure 
$${\bf L}(\lambda) 
= \left(\begin{array}{cc} U_{N+1}(\lambda) & V_{N+1}(\lambda) \\ 
\lambda W_{N+1}(\lambda) &  \lambda X_{N}(\lambda) \end{array}\right),$$ 
where the subscripts on $U,V$ etc.\ denote the degrees of these polynomials, with $V_{N+1}$ and $W_{N+1}$ both being monic; this is easily shown by induction on $N$. Furthermore, due to \eqref{lm1}, the entries of this matrix can be rewritten in terms of some polynomials of lower degree, in the form  
$${\mathbf{L}}(\lambda) = - \mathbf{I} + (\lambda+1)\, 
\hat{\mathbf{L}}(\lambda)
= \left(\begin{array}{cc} (\lambda+1)\hat{U}_{N}(\lambda)-1 & (\lambda+1)\hat{V}_{N}(\lambda) \\ 
\lambda (\lambda+1)\hat{W}_{N}(\lambda) & (\lambda+1) \hat{X}_{N}(\lambda)-1 \end{array}\right),$$ 
with  $\hat{V}_{N}$ and $\hat{W}_{N}$ both being monic, 
and $\hat{X}_{N}(0)=1$. Upon taking the trace and determinant of the above, we have the two equations 
$$ 
\begin{array}{rcl} 
{\cal P}(\lambda) & = & -2 + (\lambda+1)
\big(\hat{U}_N(\lambda)+\hat{X}_N(\lambda)\big), \\
\det \mathbf{L}(\lambda) & = & (\lambda+1)^2\det \hat{\mathbf{L}}(\lambda) -(\lambda+1)\big(\hat{U}_N(\lambda)+\hat{X}_N(\lambda)\big) +1. 
\end{array} 
$$
On the other hand, by taking the determinant of the 
product in \eqref{mdy}, we have $\det \mathbf{L}(\lambda)=-\lambda^{2N+3}$, so combining this with the above two relations yields the formula 
\begin{equation}\label{Pform} 
{\cal P}(\lambda) = (\lambda+1)^2\det \hat{\mathbf{L}}(\lambda) +\lambda^{2N+3} -1. 
\end{equation}
Then differentiating this with respect to $\lambda$ and evaluating at $\lambda=-1$ gives 
${\cal P}'(-1) = 2N+3 $, which is equivalent to the required 
second relation for (the pullbacks of) the coefficients of the generating polynomial $\cal P$. 
\end{proof} 

\begin{rem}
An alternative way to derive the result of the preceding lemma is by considering the determinant 
$$ 
{\cal D}(\lambda) = \left|
 \begin{array}{rrrll}
-L_0 & 1 & && \lambda \\ 
-\lambda & -L_{1} & 1 && \\ 
 & \ddots & \ddots & \ddots & \\ 
 && -\lambda & -L_{2N+1} & 1 \\ 
-1 &&& -\lambda & -L_{2N+2}
\end{array}
 \right| .
$$
Upon expanding this in powers of $\lambda$, it is found to be 
$$ 
{\cal D}(\lambda) = \lambda^{2N+3} -1 - {\cal P}(\lambda). 
$$
Then the first relation ${\cal P}(-1)=-2$ just follows from 
${\cal D}(-1) =0 $, which coincides with 
\eqref{Mdet}, while the second relation 
${\cal P}'(-1) = 2N+3 $ is a consequence of 
${\cal D}'(-1) =0 $, which holds because the latter derivative can be expressed as a linear combination of minors of \eqref{M0} of size $2N+2$, which all vanish since the matrix $\mathbf{M}_0$ has a two-dimensional kernel. 
\end{rem}

Before proceeding to our conclusion about the integrability of 
the $A_{2N}$ cluster map $\varphi$, 
it is instructive to consider the two examples of lowest dimension, as we now do. 


\begin{exmp}
In the simplest case $N=1$, the $A_2$ cluster map is 
$$ 
\varphi_{A_2}: \qquad \left(\begin{array}{c} x_1 \\ x_2 \end{array}\right) 
\longmapsto 
 \left(\begin{array}{c} x_1' \\ x_2' \end{array}\right)  = 
\left(\begin{array}{c} \frac{x_2+1}{x_1} \\ \frac{x_1'+1}{x_2} \end{array} \right). 
$$
Every orbit has period $p=5$, and 
this is the square of the Lyness map, i.e.\ $\varphi_{A_2}=\varphi\cdot\varphi=\varphi^2$ with $\varphi$ given by  \eqref{lyness5}. Hence it preserves the same invariant $H$, given by the formula \eqref{hlyness}.  
From $L_j=(\varphi_{A_2}^*)^j x_1$ we find the five Laurent polynomials 
$$ 
L_0=x_1, \quad 
L_1= \frac{x_2+1}{x_1}, 
\quad 
L_2= \frac{x_1+1}{x_2}, 
\quad 
L_3 = x_2, 
\quad L_4 = 
 \frac{x_1+x_2+1}{x_1x_2}, 
$$
which in this case provide a complete set of generators for the $A_2$ cluster algebra. The first three of the equations  
\eqref{Ejdet} satisfied by the $L_j$ are 
\begin{equation}\label{A2var}
\begin{array}{rcl}
1+L_0+L_2-L_0L_1L_2 & = & 0, \\ 
1+L_1+L_3-L_1L_2L_3 & = & 0, \\ 
1+L_2+L_4-L_2L_3L_4 & = & 0,
\end{array}
\end{equation}
which define an affine variety in $\C^5$, and there are two more equations obtained from these via shifting indices $\bmod\, 5$, but these are not algebraically independent of \eqref{A2var}. This defines the image of the map $\chi$ in this case.  
The dressing chain first integrals have 
generating polynomial 
$$ 
\begin{array}{rcl}
\lambda^2 {\cal I}_2 +\lambda {\cal I}_1+{\cal I}_0 
& = &   ( L_0 + L_1 + L_2 +L_3 + L_4)\lambda^2  \\ && +(L_0L_1L_2 + 
L_1L_2L_3 +L_2L_3L_4 +L_3L_4L_0 +L_4L_0L_1
) \lambda \\ 
&& + L_0 L_1L_2L_3L_4, 
\end{array}
$$
and, upon pulling this back to the space of cluster variables, the coefficients satisfy the pair of constraints 
$$ 
\begin{array}{rcr}
L_0 + L_1 + L_2 +L_3 + L_4 - (L_0L_1L_2 + 
L_1L_2L_3 +L_2L_3L_4 +L_3L_4L_0 +L_4L_0L_1
) && \\ 
+   L_0L_1L_2L_3L_4 &=&  -2, \\
- 2(L_0 + L_1 + L_2 +L_3 + L_4) + L_0L_1L_2 + 
L_1L_2L_3 +L_2L_3L_4 +L_3L_4L_0 +L_4L_0L_1 & = & 5. 
\end{array}
$$
Under the pullback, the first integrals are given by 
$\chi^*{\cal I}_2 =H$, $\chi^*{\cal I}_1 =2H+5$, $\chi^*{\cal I}_0 =H+3$, and the generating polynomial can be written in terms of the single first integral \eqref{hlyness}, as 
$$ 
{\cal P}(\lambda) = (\lambda+1) \Big( H(\lambda+1)+5\Big) -2. 
$$
\end{exmp}

\begin{exmp} 
When $N=2$, 
the (undeformed) $A_4$ cluster map is given by the 
composition of the four mutations \eqref{muA4} with all the parameters set to 1 (i.e.\ with $a_1=a_4=b_1=b_4=1$). 
Every orbit has period $p=7$. 
In terms of the 7 Laurent polynomials 
$$
L_0=x_1, L_1 = \frac{x_2+1}{x_1}, L_2=\frac{x_1+x_3}{x_2}, L_3 = \frac{x_2+x_4}{x_3}, L_4 = \frac{x_3+1}{x_4}, L_5=x_4
$$
and 
$$ 
L_6 = \frac{x_1x_2x_3 + x_2x_3x_4 + x_1x_2 + x_1x_4 + x_3x_4}{x_1x_2x_3x_4}, 
$$
it has two independent first integrals $H_1,H_2$ given by \eqref{firstintA4}, which can be found  
explicitly in terms of the cluster variables $x_j$ by setting $a_1=a_4=1$ in the 
expressions \eqref{A4ints}. The
two quantities are involution, so $\{H_1,H_2\}=0$ with respect to the invariant log-canonical Poisson structure defined by 
\eqref{poissonbracketA4}
with \eqref{pmatrixA4}. 
The map $\chi$ from the cluster variables to the $L_j$ defines a four-dimensional affine variety in $\C^7$, defined by the vanishing of $E_0,E_1,E_2$ as in \eqref{Ejdet}, while the other 4 conditions $E_j=0$ for $3\leq j\leq 6$ are algebraic consequences of these. 
Moreover, on this variety, the coefficients of the generating polynomial $\lambda^3 {\cal I}_3+\lambda^2 {\cal I}_2 +\lambda {\cal I}_1+{\cal I}_0 $ satisfy the constraints 
$$
\begin{array}{rcr}
-{\cal I}_3 +{\cal I}_2 - {\cal I}_1+{\cal I}_0 & = & -2, \\
3{\cal I}_3-2{\cal I}_2 + {\cal I}_1 & = & 7. 
\end{array} 
$$
From \eqref{firstintA4}, we see that after pulling back to the cluster variables we may identify 
$\chi^* {\cal I}_3 = H_1$, $\chi^* {\cal I}_0 = H_2$, 
and the above constraints then give 
$\chi^* {\cal I}_1 = H_1+2H_2-3$, $\chi^* {\cal I}_2 = 2H_1+H_2-5$. 
Hence, under the pullback, the generating polynomial can be expressed in terms of the two independent first integrals 
$H_1,H_2$, as 
$$ 
{\cal P}(\lambda) = (\lambda+1) \Big( (\lambda+1)\big(H_1\lambda+H_2-5\big)+7\Big) -2.
$$
\end{exmp}

\begin{thm}\label{liouvilleA2N} 
The periodic $A_{2N}$ cluster map $\varphi=\varphi_{A_{2N}}$, as defined by 
\eqref{mueq1}, is  integrable in the 
Liouville sense, having $N$ Poisson-commuting first integrals that are obtained as independent coefficients in 
the equation for its invariant spectral curve, which is the hyperelliptic curve of genus $N$ defined by the singular affine model 
\begin{equation}\label{spec}
\mu^2 -{\cal P}(\lambda )\mu-\lambda^{2N+3}=0. 
\end{equation}
\end{thm}

\begin{proof}
The spectral curve is defined by the equation 
$$ 
\det \big(\mathbf{L}(\lambda ) -\mu \mathbf{I}\big) =0, 
$$
in terms of the matrix product \eqref{mdy} with the quantities $L_j =(\varphi^*)^jx_1$. Under the action of $\varphi$, this product transforms as 
$$
\varphi^* \mathbf{L}(\lambda ) = 
\left(\begin{array}{cc} L_0 & 1 \\ 
\lambda & 0 \end{array}\right)^{-1}\, \mathbf{L}(\lambda )\, \left(\begin{array}{cc} L_0 & 1 \\ 
\lambda & 0 \end{array}\right),  
$$
so the spectral curve is clearly invariant, 
and in terms of its trace and its determinant as found in the proof of Lemma \eqref{trdet}, it is given by the 
formula  \eqref{spec}. 
The non-trivial first integrals appearing in the formula are found from the $N+2$ coefficients of the polynomial ${\cal P}$, which (pulled back by the map $\chi$) are subject to the two linear relations \eqref{trm1} and  \eqref{trm2}. After eliminating $\chi^*{\cal I}_{N+1}$ and $\chi^*{\cal I}_{N}$, say, this means that all the coefficients can be written as linear combinations of the remaining $N$ first integrals 
$\chi^*{\cal I}_{j}$ for $0\leq j \leq N-1$, which Poisson commute by Corollary \ref{chiI}. These $N$ quantities are required to be functionally independent, meaning that 
$\chi^* (\rd {\cal I}_0\wedge \rd {\cal I}_1 \wedge \cdots \wedge \rd {\cal I}_{N-1})\neq 0$: for any fixed $N$ the non-vanishing of this form can be checked at any particular (generic) point in phase space, whence it holds on an open set.  (For instance, it is convenient to take the point 
$(x_1,x_2,\ldots,x_{2N}) =(1,1,\ldots,1)$, where the 
quantities \eqref{Lsfunctions} in the first non-trivial row of the frieze take the values $(L_0,L_1,\ldots,L_{2N+2}) =(1,2,\ldots,2,1,2N+1)$.) The affine curve defined by \eqref{spec} is singular: due to the conditions 
${\cal P}(-1)=-2$ and ${\cal P}'(-1)=2N+3$, the derivatives of this equation with respect to $\lambda$ and $\mu$ both vanish at the point $(-1,-1)$, where the curve has a node. To resolve the singularity, we introduce a new coordinate 
$\hat{\mu}$ and define a monic polynomial $\hat{\cal P}$ of degree $2N+1$ via 
$$
\mu = (\lambda+1)\hat{\mu} +\frac{1}{2}{\cal P}(\lambda), 
\qquad 
(\lambda+1)^2 \hat{\cal P}(\lambda) = 
\lambda^{2N+3}+\frac{1}{4}{\cal P}(\lambda)^2. 
$$
Then the original spectral curve \eqref{spec} is birationally
equivalent to 
$$ 
\hat{\mu}^2 = \hat{\cal P}(\lambda), 
$$
which (at generic points in phase space) defines a smooth hyperelliptic curve of genus $N$, as required. 
\end{proof} 


\begin{rem}
There is a lot more that can be said about the geometry of the periodic cluster map $\varphi_{A_{2N}}$. For instance, the generic level set of the first integrals is an affine component in the Jacobian of the spectral curve \eqref{spec}, on which the map acts as a translation by a torsion point of order $2N+3$. Moreover, the matrix product  \eqref{mdy} reveals the connection with the periodic continued fraction 
$$ 
L_0 + \cfrac{\lambda}{L_1 +\cfrac{\lambda}{\ddots\,  +\,  \cfrac{\lambda}{L_{2N+2}+\cfrac{\lambda}{L_0 +\cfrac{\lambda}{\ddots}}}}}. $$
The latter is a degenerate case of the periodic continued fractions, defined by functions on hyperelliptic curves, that were considered in \cite{grossetveselov}.  
\end{rem}

Having understood the Liouville integrability of the 
$\varphi_{A_{2N}}$ cluster maps, in the next section we seek integrable deformations of them. 

\section{Integrable deformations of $A_{2N}$ cluster maps}
\label{A2Ndef}
\setcounter{equation}{0} 

In order to obtain integrable deformations of the periodic 
cluster maps of type $A_{2N}$, 
we begin with a construction for the case of $A_{6}$, and consider how this generalizes the cases of $A_{2}$ and $A_{4}$ seen earlier.

In the previous examples of deformations, 
when we discussed 
the cluster maps  of types $A_2$ and $A_4$,  
we explained in each case how the singularity confinement patterns of these 
deformed integrable maps suggested   
a way to define a lift to a space of tau functions. 
The tau functions are cluster variables in an enlarged phase space where the deformed
map acts as a cluster map, and the 
Laurent property is restored. 
This begs the question as to 
whether this Laurentification procedure can be successfully applied to the deformed cluster maps associated with general type $A$.  We are able to answer this question positively for deformed cluster algebras of type $A_{2N}$, beginning with the case of $A_6$, and then obtaining the general case by a  method of local expansion applied successively to the quivers of the enlarged algebra.

\subsection{Integrable deformation of the $A_{6}$ cluster map}

\label{ss:type-A6}

In this subsection, we consider the deformation of the periodic cluster map of type $A_{6}$.  The exchange matrix 
is given by 
\begin{equation}\label{exchA6}
     B_{A_{6}} = \mqty(0 & 1 & 0 & 0 & 0 & 0 \\ -1 & 0 & 1 & 0 & 0 & 0 \\ 0 & -1 & 0 & 1 & 0 & 0 \\ 0 & 0 & -1 & 0 & 1 & 0 \\ 0 & 0 & 0 & -1 & 0 & 1 \\ 0 & 0 & 0 & 0 & -1 & 0)
\end{equation}
The corresponding matrix is invariant under a   sequence of 6 cluster mutations, that is  
\[    \mu_{6}\mu_{5}\mu_{4}\mu_{3}\mu_{2}\mu_{1}(B_{A_{6}}) = B_{A_{6}},  \]

Given the initial cluster $\vb{x} = (x_{1},x_{2},x_{3},x_{4},x_{5},x_{6})$, let us denote by $\varphi_{A_{6}}$ the composition of mutations above, i.e.\ $\varphi_{A_{6}} = \mu_{6}\mu_{5}\mu_{4}\mu_{3}\mu_{2}\mu_{1}$. 
By Zamolodchikov periodicity, every orbit of the map $\varphi_{A_{6}}$ has period 9. 
Once more, with $f_{k}(x)=b_{k}M_{k}^{-} + a_{k}M_{k}^{+}$ in \eqref{eq:6}, we define the modified mutations $\tilde{\mu}_{k}(x_{k}) = x^{-1}_{k}\qty(b_{k}M_{k}^{-} + a_{k}M_{k}^{+})$ for $k =1,\dots,6$, which yields deformed map $\tilde{\varphi}_{A_{6}}= \tilde{\mu}_{6}\tilde{\mu}_{5}\tilde{\mu}_{4}\tilde{\mu}_{3}\tilde{\mu}_{2}\tilde{\mu}_{1}$ equivalent to the following exchange relations, 
\begin{equation}\label{deformedmuA6}
\begin{aligned}
    \tilde{ \mu}_{1}: &&  x_1 x'_1 & = b_{1} +  a_{1}x_2 && \\
  \tilde{ \mu}_{i} : && x_i x'_i &=  1 +  x'_{i-1}x_{i+1} &&   (2\leq i \leq 4)\\
   \tilde{ \mu}_{5} : && x_{5} x'_{5} &= b_{5} + a_{5}x'_4x_{6} &&\\
    \tilde{ \mu}_{6} : && x_{6} x'_{6} &= b_{6} + a_{6}x'_5 &&
\end{aligned}
\end{equation}
where the parameters $b_{i}$ and $a_{i}$ for $i=2,3,4$ are rescaled to 1 via  $x_{i} \to \lambda_{i}x_{i}$ $\lambda_{i} \in \mathbb{C}^{*}$ for each cluster variable. 

For the original undeformed cluster map $\varphi_{A_{6}}$, with all coefficients $a_j,b_j$ above equal to 1, we know from 
the previous section that it is Liouville integrable, with three independent Poisson-commuting first integrals obtained as a subset of the homogeneous components of the function ${\cal K}$ in \eqref{casK} with $N=3$. 
In particular, we can choose 
\begin{equation}\label{firstintA61}
    H_{1}  = \sum_{j=0}^{8}L_{j}, \quad H_{2}  = \prod_{j=0}^{8}L_{j}, \quad H_{3}  = \sum_{j=0}^{8}L_{j}L_{j+1}(L_{j+2} + L_{j+4} + L_{j+6}) + \sum_{j=0}^{2}L_{j}L_{j+3}L_{j+6}
\end{equation}
where $L_{i} = (\varphi_{A_{6}}^{*})^{i}(x_{1})$. These $H_j$ for $j=1,2,3$ correspond to the quantities ${\cal I}_4,{\cal I}_2,{\cal I}_0$ appearing in ${\cal K}$, pulled back to the cluster variables. 
To obtain deformed first integrals $\tilde{H}_j$ 
for the map defined by \eqref{deformedmuA6}, we adopt the same approach as in Example~\ref{exmpA4}, attaching an arbitrary  coefficient to each Laurent monomial appearing in $H_j$. Then the invariance requirement 
$\tilde{\varphi}_{A_{6}}^{*}\tilde{H}_j = \tilde{H}_j$ for $j=1,2,3$ determines 
a system of conditions on these coefficients and on the parameters $a_j,b_j$, which can be solved 
by using a computer algebra system (we used Maple\texttrademark\ \cite{Maple}). 
We find that the system has a solution if and only if the parameters are fixed as $b_{1} = 1 = b_{5} $ and $b_{6}a_{5}^2=1$, in which case three independent first integrals are given as follows:
\begin{equation}
	\begin{aligned}
   \tilde{H}_{1} &= \frac{1}{a_{5}^3 a_{6} x_{1} x_{2}x_{3}x_{4}x_{5}x_{6}}\qty(\begin{aligned}a_{1}&a_{5}^2a_{6}x_{1}x_{2}x_{3}x_{4} + a_{1}x_{1}x_{2}x_{3}x_{4}x_{5} + a_{1}a_{5}^4a_{6}^2 x_{1}x_{2}x_{3}x_{4}x_{5} \\ &+ a_{1}a_{5}^2a_{6}x_{1}x_{2}x_{3}x_{4}x_{5}^2 + a_{1}a_{5}^3a_{6}x_{1}x_{2}x_{3}x_{6} + a_{1}a_{5}^3a_{6}x_{1}x_{2}x_{3}x_{4}^2x_{6} \\
    &+a_1 a_5^3 a_6 x_1 x_2 x_5 x_6 + a_1 a_5^3 a_6 x_1 x_2 x_3^2 x_5 x_6 + 
 a_1 a_5^3 a_6 x_1 x_4 x_5 x_6 \\
 &+ a_1 a_5^3 a_6 x_1 x_2^2 x_4 x_5 x_6 + 
 a_1 a_5^3 a_6 x_3 x_4 x_5 x_6 + a_1 a_5^3 a_6 x_1^2 x_3 x_4 x_5 x_6 \\
 &+ a_5^3 a_6 x_2 x_3 x_4 x_5 x_6 + a_1^2 a_5^3 a_6 x_2 x_3 x_4 x_5 x_6 + 
 a_5^3 a_6 x_1^2 x_2 x_3 x_4 x_5 x_6 \\
 &+ a_1 a_5^3 a_6 x_2^2 x_3 x_4 x_5 x_6 + 
 a_1 a_5^3 a_6 x_1 x_3^2 x_4 x_5 x_6 + a_1 a_5^3 a_6 x_1 x_2 x_4^2 x_5 x_6 \\
 &+ a_1 a_5^3 a_6 x_1 x_2 x_3 x_5^2 x_6 + a_1 a_5^4 a_6 x_1 x_2 x_3 x_4 x_6^2 + 
 a_1 a_5^2 x_1 x_2 x_3 x_4 x_5 x_6^2\end{aligned}) \\ & \\
    \tilde{H}_{2} & = \qty(a_{1} + x_{2})\qty(\frac{x_{1} + x_{3}}{x_{2}})\qty(\frac{x_{2} + x_{4}}{x_{3}})\qty(\frac{x_{3} + x_{5}}{x_{4}})\qty(\frac{x_{4} + a_{5}x_{6}}{x_{5}})\qty(\frac{x_{5} + a_{5}^{2}a_{6}}{a_{5}}) \\
     & \qquad \cdot \qty(\frac{\splitfrac{a_{5}^2a_{6}x_{1}x_{2}x_{3}x_{4}x_{5} + a_{1}a_{5}x_{2}x_{3}x_{4}x_{5}x_{6} + a_{5}x_{1}x_{2}x_{3}x_{6} }{+ a_{5}x_{1}x_{2}x_{5}x_{6} + a_{5}x_{1}x_{4}x_{5}x_{6} + a_{5}x_{3}x_{4}x_{5}x_{6} + x_{1}x_{2}x_{3}x_{4}}}{a_{5}x_{1}x_{2}x_{3}x_{4}x_{5}x_{6}}) \\ & \\
  \tilde{H}_{3} & =\frac{P}{a_{1} x_{2}^{2} x_{4}^{2} x_{5}^{2} a_{5}^{3} x_{6} x_{3}^{2} x_{1} a_{6}}
  \end{aligned}
  \end{equation}
  where
  \begin{equation*}
  \begin{aligned}
  	P & = x_{2} x_{3}^{2} x_{6} a_{6}^{3} (a_{1} (x_{2}+x_{4}) x_{3}+x_{1} (a_{1}+x_{2}) x_{4}+a_{1} x_{2} x_{5}) x_{5} x_{4} x_{1} a_{5}^{7}+x_{3} ((a_{1} (x_{2}+x_{4}) x_{3}^{2} \\
  	&+((x_{2} x_{1}+a_{1} (x_{1}+x_{5})) x_{4}+x_{2}^{2} x_{5}) x_{3}+x_{1} x_{5} (x_{2}+x_{4}) (a_{1}+x_{2})) x_{2} x_{5} x_{4}^{2} x_{1} a_{6} \\
  	&+x_{6}^{2} ((x_{5} (a_{1}+x_{2}) (a_{1} x_{2}+1) x_{4}+a_{1} x_{2} x_{1}) (x_{2}+x_{4}) x_{3}^{2}+(((a_{1}^{2}+2) x_{2}+2 a_{1}) x_{5} x_{1} x_{4}^{2} \\
  	&+x_{2} (a_{1} x_{5} (x_{1}+x_{5}) x_{2}^{2} +((a_{1}^{2}+1) x_{5}^{2}+x_{1}^{2}) x_{2}+a_{1} (x_{1}+x_{5})^{2}) x_{4}+2 a_{1} x_{1} x_{2}^{2} x_{5}) x_{3} \\
  	&+((a_{1} x_{2}^{2} x_{5}+a_{1} x_{1}+x_{2} x_{1}) x_{4}+a_{1} x_{2} x_{5}) x_{5} x_{1} (x_{2}+x_{4}))) a_{6}^{2} a_{5}^{6}\\
  	&+x_{6} a_{6}^{2} x_{4} ((x_{2}+x_{4})(x_{5} (a_{1}+x_{2}) (a_{1} x_{2}+1) x_{4}+ (x_{2} x_{5}^{2}+2 a_{1}) x_{2} x_{1}) x_{3}^{3} \\
  	&+((a_{1} x_{2}^{2} x_{5}+(a_{1}^{2} x_{1}+x_{5} a_{1}^{2}+2 x_{1}+x_{5}) x_{2}+(2x_{1}+x_{5}) a_{1}) x_{5} x_{4}^{2} + x_{2} (a_{1} x_{2}^{2} x_{5}+2x_{2} x_{1}\\
  	&+a_{1} ( x_{5}^{3}+2 x_{1}+3 x_{5})) x_{1} x_{4}+x_{2}^{2} ((x_{1} x_{5}+x_{5}^{2}+1) x_{2}+a_{1} (x_{1} x_{5}+2)) x_{5} x_{1}) x_{3}^{2} \\
  	&+(((x_{5} a_{1}^{2}+x_{1}+2 x_{5}) x_{2}+a_{1} (x_{1}+2 x_{5})) x_{4} \\
  	&+x_{2} ((x_{1} x_{5}^{2}+x_{1}+x_{5}) x_{2}+a_{1} x_{1} (x_{5}^{2}+1))) x_{5} x_{1} (x_{2}+x_{4}) x_{3}+x_{1}^{2} x_{5}^{2} (x_{2}+x_{4})^{2} (a_{1}+x_{2})) a_{5}^{5} \\
  	&+x_{3} a_{6} (x_{2} (a_{1} (x_{5}^{2}+1) (x_{2}+x_{4}) x_{3}^{2}+((x_{5}^{2}+1) (x_{2} x_{1}+a_{1} (x_{1}+x_{5})) x_{4}\\
  	&+(x_{5}^{3}+x_{5}) x_{2}^{2}+a_{1}^{2} x_{2} x_{5}^{2}) x_{3}+x_{1} x_{5} (x_{5}^{2}+1) (x_{2}+x_{4}) (a_{1}+x_{2})) x_{4}^{2} x_{1} a_{6}\\
  	&+x_{6}^{2} x_{5} ((x_{5} (a_{1}+x_{2}) (a_{1} x_{2}+1) x_{4}+a_{1} x_{2} x_{1}) (x_{2}+x_{4}) x_{3}^{2} \\
  	&+(x_{1} (x_{2}^{2} a_{1}^{2}+(a_{1}^{2}+2) x_{5} x_{2}+2 a_{1} x_{5}) x_{4}^{2}\\
  	&+x_{2} (a_{1} x_{5} (x_{1}+x_{5}) x_{2}^{2} +((a_{1}^{2}+1) x_{5}^{2}+x_{1}^{2}) x_{2}+a_{1} (x_{1}+x_{5})^{2}) x_{4}+2 a_{1} x_{1} x_{2}^{2} x_{5}) x_{3} \\
  	&+((a_{1} x_{2}^{2} x_{5}+a_{1} x_{1}+x_{2} x_{1}) x_{4}+a_{1} x_{2} x_{5}) x_{5} x_{1} (x_{2}+x_{4}))) a_{5}^{4} +(a_{1} x_{1} (x_{2}+x_{4}) (a_{1} x_{5}+2) x_{3}^{2}\\
  	&+(a_{1} x_{1} x_{2} (a_{1}+x_{5}) x_{4}^{2} +(a_{1} x_{5} (a_{1}+x_{5}) x_{2}+x_{5}^{2}+a_{1} (x_{1}^{2}+1) x_{5}+2 x_{1}^{2}) (a_{1}+x_{2}) x_{4}\\
  	&+a_{1} x_{1} x_{2} (a_{1} x_{5}^{2}+a_{1}+3 x_{5})) x_{3}\\
  	&+a_{1} x_{1} x_{5} (x_{2}+x_{4}) (x_{2} x_{4}+1) (a_{1}+x_{5})) x_{2} x_{3} x_{6} a_{6} x_{5} x_{4} a_{5}^{3}+x_{2} x_{3} ((a_{1} (x_{2}+x_{4}) x_{3}^{2}\\
  	&+((x_{2} x_{1}+a_{1} (x_{1}+x_{5})) x_{4}+x_{2} (x_{2} x_{5}+a_{1}^{2} (x_{5}^{2}+1))) x_{3}+x_{1} x_{5} (x_{2}+x_{4}) (a_{1}+x_{2})) a_{6}\\
  	&+a_{1}^{2} x_{2} x_{3} x_{5} x_{6}^{2}) x_{5} x_{4}^{2} x_{1} a_{5}^{2}+a_{1}^{2} x_{2}^{2} x_{4}^{2} x_{5}^{2} x_{3}^{2} x_{1}
  \end{aligned}
  \end{equation*}

A computer-aided calculation enables us to verify the following
\begin{thm}\label{A6thm} The conditions $b_{1} = 1 = b_{5} $ and $b_{6}a_{5}^2=1$ on the parameters are necessary and sufficient conditions for 
$\tilde{H}_{1},\tilde{H}_{2}$ and  $\tilde{H}_{3}$ to be first integrals that are preserved by the type $A_{6}$ deformed map, i.e.\ $\tilde{\varphi}_{A_{6}}^{*}(\tilde{H}_{j}) = \tilde{H}_{j}$ for $j=1,2,3$, and these quantities are in involution with respect to the Poisson bracket \eqref{poissA2N}. Hence $\tilde{\varphi}_{A_{6}}$ is a Liouville integrable map whenever these conditions on the parameters hold. 
\end{thm}

The deformed map $\tilde{\varphi}_{A_{6}}$ is a Liouville integrable symplectic map, but it cannot be a cluster map because the generated variables stop being Laurent polynomials after some iterations. Therefore, once again, we look to lift the deformed map to a higher dimensional space by Laurentification. Following the same process as in the previous examples, we study the singularity structures of the deformed map $\tilde{\varphi}_{A_{6}}$ by observing the $p$-adic properties of orbits defined over $\mathbb{Q}$ (see Remark \ref{padic}). Then we observe the following singularity patterns: 
\begin{equation}\label{singpA6}
    \begin{aligned}
         \text{(1)} : &&\dots &\to (\epsilon, R,R,R,R,R)\to (\epsilon^{-1},\epsilon^{-1},\epsilon^{-1},\epsilon^{-1},\epsilon^{-1},\epsilon^{-1}) \to (R,R,R,R,R,\epsilon) \\ 
         \text{(2)} :&&\dots &\to ( R,R,R,R,R,\epsilon) \to  ( R,R,R,R,R,\epsilon^{-1}) \to ( R,R,R,R,\epsilon^{-1},R) \\
        &&&\to ( R,R,R,\epsilon^{-1},R,R) \to ( R,R,\epsilon^{-1},R,R,R) \to ( R,\epsilon^{-1},R,R,R,R)\\
        &&& \to ( \epsilon^{-1},R,R,R,R,R) \to ( \epsilon,R,R,R,R,R) \to \dots \\
        \text{(3)} :&&\dots &\to (R, \epsilon,R,R,R,R) \to \dots \\ 
        \text{(4)} :&&\dots &\to (R, R , \epsilon,R,R,R)\to \dots \\
        \text{(5)} :&&\dots &\to (R, R , R, \epsilon,R,R)\to \dots \\
        \text{(6)} :&&\dots &\to (R, R ,R, R, \epsilon,R)\to \dots \\
    \end{aligned}
\end{equation}
%
By introducing six different tau functions $\tau_{n}$, $\sigma_{n}$, $p_{n}$, $r_{n}$, $q_{n}$, $w_{n}$, one corresponding to 
each pattern, 
we construct a rational map $\pi_{A_{6}}: \mathbb{C}^{15} \to \mathbb{C}^{6} $, analogous to \eqref{vartransA4} in Example~\ref{LaurentA4}. 
The map is defined by  the dependent variable transformation  
\begin{equation}\label{vartransA6}
    \begin{split}
 \pi_{A_{6}}:\,\,   x_{1,n} = \frac{\sigma_{n}\tau_{n+1}}{\sigma_{n+1}\tau_{n}} \quad & x_{2,n} = \frac{p_{n}}{\sigma_{n+2}\tau_{n}} \quad  x_{3,n} =\frac{r_{n}}{\sigma_{n+3}\tau_{n}} \quad  x_{4,n} =\frac{q_{n}}{\sigma_{n+4}\tau_{n}}\\[1em]
    & x_{5,n} =\frac{w_{n}}{\sigma_{n+5}\tau_{n}} \quad x_{6,n} =\frac{\sigma_{n+7}\tau_{n-1}}{\sigma_{n+6}\tau_{n}}, 
    \end{split}
\end{equation}
where $\tau_n$ and $\sigma_n$ represent patterns (1) and (2) respectively, while the variables $p_n,r_n,q_n,w_n$ correspond to patterns (3)-(6).

Upon substituting \eqref{vartransA6} directly into the components \eqref{deformedmuA6} of $\tilde{\varphi}_{A_{6}}$ with the conditions $b_{6}a_{5}^{2} = 1 $, $b_{i} = 1 = a_{j} $ for $i=1,\dotsc, 5$ and $j = 2,3,4$, one obtains the following system of equations: 
\begin{equation}\label{systmA6}
\begin{aligned}
    \tau_{n+2}\sigma_{n} &= \sigma_{n+2}\tau_{n} + a_{1}p_{n} \\
    p_{n+1}p_{n} &= \sigma_{n+3}\sigma_{n+2}\tau_{n}\tau_{n+1} + r_{n} \sigma_{n+1}\tau_{n+2} \\ 
    r_{n+1}r_{n} &= \sigma_{n+4}\sigma_{n+3}\tau_{n}\tau_{n+1} + q_{n}p_{n+1}\\
    q_{n+1}q_{n} &= \sigma_{n+5}\sigma_{n+4}\tau_{n}\tau_{n+1} + w_{n}r_{n+1} \\
    w_{n+1}w_{n} &= \sigma_{n+6}\sigma_{n+5}\tau_{n}\tau_{n+1} + a_{5}\sigma_{n+7}q_{n+1}\tau_{n-1} \\
    \sigma_{n+8}\tau_{n-1} &= b_{6}\sigma_{n+6}\tau_{n+1} + a_{6}w_{n+1}
    \end{aligned}
\end{equation}
Once again, we begin by presenting the initial data as 
\begin{align*}
(\tilde{x}_{1},\tilde{x}_{2},&\tilde{x}_{3},\tilde{x}_{4},\tilde{x}_{5},\tilde{x}_{6},\tilde{x}_{7},\tilde{x}_{8},\tilde{x}_{9},\tilde{x}_{10},\tilde{x}_{11},\tilde{x}_{12},\tilde{x}_{13},,\tilde{x}_{14},\tilde{x}_{15}) \\ 
&= (q_{0}, w_{0},\tau_{-1},\tau_{0},\tau_{1}, \sigma_{0},\sigma_{1},\sigma_{2},\sigma_{3},\sigma_{4},\sigma_{5},\sigma_{6},\sigma_{7}, p_{0},r_{0}) 
\end{align*}
Then a new exchange matrix can be found by reading off the coefficients of the presymplectic form $\pi_{A_{6}}^{*}\omega_{A_{6}}$ (the pullback of the symplectic form under $\pi_{A_{6}}$). 
This matrix is given by 
\begin{equation}\label{exchmA6}
    \tilde{B}_{A_6}=  {\footnotesize\left(
\begin{array}{*{15}c}
 0 & 1 & 0 & 0 & 0 & 0 & 0 & 0 & 1 & 0 & -1 & 0 & 0 & 0 & -1\\ 
 -1 & 0 & 1 & 0 & 0 & 0 & 0 & 0 & 0 & 1 & 0 & -1 & 1 & 0 & 0 \\
  0 &  -1 & 0 & 1 & 0 & 0 & 0 & 0 & 0 & 0 & 1 & 0 & 0 & 0 & 0\\ 
   0 &  0 &  -1& 0 & 1 & 1 & -1 & 0 & 0 & 0 & 0 & 1 & -1 & 0 & 0\\
 0 & 0 & 0 & -1 & 0 & 0 & 0 & -1 & 0 & 0 & 0 & 0 & 0 & 1 & 0\\
 0 & 0 & 0 & -1& 0 & 0 & 0 & -1 & 0 & 0 & 0 & 0 & 0 & 1 & 0 \\ 
 0 &  0 & 0 & 1 & 0  & 0 & 0 & 1 & 0 & 0 & 0 & 0 & 0 & -1 & 0 \\ 
 0 & 0 & 0 & 0 & 1 & 1 & -1  & 0 & 1 & 0 & 0 & 0 & 0 & 0& -1 \\ 
 -1 & 0 & 0 & 0 & 0& 0 & 0 & -1& 0 & 1 & 0 & 0 & 0 & 1 & 0 \\ 
 0 & -1 & 0 & 0 & 0 & 0 & 0 & 0 & -1  & 0 & 1 & 0 & 0 & 0 & 1 \\ 
 1 & 0 & -1 & 0 & 0 & 0 & 0 & 0 & 0 & -1 & 0 & 1 & -1 & 0 & 0 \\ 
 0 & 1 & 0 & -1 & 0 & 0 & 0 & 0 & 0 & 0 & -1 & 0  & 0 & 0 & 0\\
 0 & -1 & 0 & 1 & 0 & 0 & 0 & 0 & 0 & 0 & 1 & 0 & 0 & 0 & 0 \\ 
 0 & 0 & 0 & 0 & -1 & -1 & 1 & 0 & -1 & 0 & 0 & 0 & 0 & 0 & 1 \\
  1 & 0 & 0 & 0 & 0 & 0 & 0 & 1 & 0 & -1 & 0 & 0 & 0 & -1& 0 \\
\end{array}
\right)}
\end{equation}
If this case were similar to that of type $A_{2}$ and $A_{4}$, then one would expect to be able to find an extended exchange matrix which contained entries corresponding to frozen variables $a_{1}, a_{5}$ and $a_{6}$. Then it should be invariant under a certain sequence of mutations, generating cluster variables expressed by the relations \eqref{systmA6}. However, after a few iterations of the recurrence \eqref{systmA6}, one finds a variable whose denominator is given by 
\begin{align*}
    a_{5}^2\sigma_{0}\tau_{-1}p_{0}r_{0}q_{0}w_{0}
\end{align*}
which shows that the denominator contains the ``frozen'' variable $a_5$. 
Under cluster mutation, the frozen variables  only appear in the numerators of Laurent polynomials arising as cluster variables. This indicates that with the condition $b_{6}a_{5}^{2} = 1 $, we cannot generate cluster variables corresponding to the recurrence \eqref{systmA6}, without putting a further constraints on the parameters: if we set $a_5=1$, which implies $b_{6} = 1$, then we are 
led to the following theorem.
\begin{thm}\label{thmforA6}
    The sequence of mutations in the cluster algebra defined by the exchange matrix \eqref{exchmA6} with two frozen variables $a_{1},a_{6}$ generates the  sequences of tau functions $(\sigma_{n})$, $(p_{n})$, $(r_{n})$, $(w_{n})$, $(q_{n})$, $(\tau_{n})$ satisfying   the relations 
    \begin{equation}\label{eq:systm1}
\begin{aligned}
    \tau_{n+2}\sigma_{n} &= \sigma_{n+2}\tau_{n} + a_{1}p_{n} \\
    p_{n+1}p_{n} &= \sigma_{n+3}\sigma_{n+2}\tau_{n}\tau_{n+1} + r_{n} \sigma_{n+1}\tau_{n+2} \\ 
    r_{n+1}r_{n} &= \sigma_{n+4}\sigma_{n+3}\tau_{n}\tau_{n+1} + q_{n}p_{n+1}\\
    q_{n+1}q_{n} &= \sigma_{n+5}\sigma_{n+4}\tau_{n}\tau_{n+1} + w_{n}r_{n+1} \\
    w_{n+1}w_{n} &= \sigma_{n+6}\sigma_{n+5}\tau_{n}\tau_{n+1} + \sigma_{n+7}q_{n+1}\tau_{n-1} \\
    \sigma_{n+8}\tau_{n-1} &= \sigma_{n+6}\tau_{n+1} + a_{6}w_{n+1}, 
    \end{aligned}
\end{equation}
and belonging to the Laurent polynomial ring \[ \mathbb{Z}\qty[a_1,a_6,\sigma_{0}^{\pm},\sigma_{1}^{\pm},\sigma_{2}^{\pm},\sigma_{3}^{\pm},\sigma_{4}^{\pm},\sigma_{5}^{\pm},\sigma_{6}^{\pm},\sigma_{7}^{\pm},\tau_{-1}^{\pm},\tau_{0}^{\pm},\tau_{1}^{\pm}, p_{0}^{\pm},r_{0}^{\pm},w_{0}^{\pm},q_{0}^{\pm} ] \] 
\end{thm}
%
\begin{proof} Let us extend the initial data by inserting the frozen variables:  
\begin{align*}
&(\tilde{x}_{1},\tilde{x}_{2},\tilde{x}_{3},\tilde{x}_{4},\tilde{x}_{5},\tilde{x}_{6},\tilde{x}_{7},\tilde{x}_{8},\tilde{x}_{9},\tilde{x}_{10},\tilde{x}_{11},\tilde{x}_{12},\tilde{x}_{13},\tilde{x}_{14},\tilde{x}_{15},\tilde{x}_{16},\tilde{x}_{17})\\
&=(q_{0},w_{0},\tau_{-1},\tau_{0},\tau_{1},\sigma_{0},\sigma_{1},\sigma_{2},\sigma_{3},\sigma_{4},\sigma_{5},\sigma_{6},\sigma_{7},p_{0},r_{0},a_{1},a_{6})    
\end{align*} 
We add two new rows, whose entries correspond to the frozen variables, to the exchange matrix \eqref{exchmA6} to define the extended exchange matrix
    \begin{equation} \label{extdefexchA6}
         \hat{B}_{A_{6}}=  {\footnotesize\left(
\begin{array}{*{15}c}
 0 & 1 & 0 & 0 & 0 & 0 & 0 & 0 & 1 & 0 & -1 & 0 & 0 & 0 & -1\\ 
 -1 & 0 & 1 & 0 & 0 & 0 & 0 & 0 & 0 & 1 & 0 & -1 & 1 & 0 & 0 \\
  0 &  -1 & 0 & 1 & 0 & 0 & 0 & 0 & 0 & 0 & 1 & 0 & 0 & 0 & 0\\ 
   0 &  0 &  -1& 0 & 1 & 1 & -1 & 0 & 0 & 0 & 0 & 1 & -1 & 0 & 0\\
 0 & 0 & 0 & -1 & 0 & 0 & 0 & -1 & 0 & 0 & 0 & 0 & 0 & 1 & 0\\
 0 & 0 & 0 & -1& 0 & 0 & 0 & -1 & 0 & 0 & 0 & 0 & 0 & 1 & 0 \\ 
 0 &  0 & 0 & 1 & 0  & 0 & 0 & 1 & 0 & 0 & 0 & 0 & 0 & -1 & 0 \\ 
 0 & 0 & 0 & 0 & 1 & 1 & -1  & 0 & 1 & 0 & 0 & 0 & 0 & 0& -1 \\ 
 -1 & 0 & 0 & 0 & 0& 0 & 0 & -1& 0 & 1 & 0 & 0 & 0 & 1 & 0 \\ 
 0 & -1 & 0 & 0 & 0 & 0 & 0 & 0 & -1  & 0 & 1 & 0 & 0 & 0 & 1 \\ 
 1 & 0 & -1 & 0 & 0 & 0 & 0 & 0 & 0 & -1 & 0 & 1 & -1 & 0 & 0 \\ 
 0 & 1 & 0 & -1 & 0 & 0 & 0 & 0 & 0 & 0 & -1 & 0  & 0 & 0 & 0\\
 0 & -1 & 0 & 1 & 0 & 0 & 0 & 0 & 0 & 0 & 1 & 0 & 0 & 0 & 0 \\ 
 0 & 0 & 0 & 0 & -1 & -1 & 1 & 0 & -1 & 0 & 0 & 0 & 0 & 0 & 1 \\
  1 & 0 & 0 & 0 & 0 & 0 & 0 & 1 & 0 & -1 & 0 & 0 & 0 & -1& 0 \\
  0 & 0 & 0 & 1 & 1 & -1 & -1 & 0 & 0 & 0 & 0 & 0 & 0 & 0 & 0 \\
   0 & 0 & -1 & -1 & 0 & 0 & 0 & 0 & 0 & 0 & 0 & 1 & 1 & 0 & 0 \\
\end{array}
\right)}
        \end{equation}
which can be depicted by the quiver in Figure \ref{fig:MuQuiverA6}. 

\begin{figure}[h]
\centering
\resizebox{0.7 \textwidth}{!}{%
 \begin{tikzpicture}[every circle node/.style={draw,scale=0.6,thick},node distance=15mm]

 \node [draw,circle,fill=blue!50,"$12$"] (12) at (12,0) {};
  
     \node [draw,circle,fill=red!50,"$7$"] (6) [right= of 12] {};
      \node [draw,circle,fill=red!50,"$8$"] (7) [right=of 6] {};
       \node [draw,circle,fill=green!50,"$9$"] (a1) [right=of 7] {};
        \node [draw,circle,fill=green!50,"$10$"] (a2) [right=of a1] {};
      
      \node [draw,circle,fill=red!50,"$11$"] (8) [right=of a2] {};
      \node [draw,circle,fill=red!50,"$12$"] (9) [right=of 8] {};
      \node [draw,circle,fill=red!50,"$13$"right] (10) [below right=of 9] {};
      \node [draw,circle,fill=red!50,"$6$" left] (5) [below left=of 6] {};

      \node [draw,circle,fill=blue!50,"$17$"] (13) [right=of 9] {};

       \node [draw,circle,fill=red!50,"$5$"below] (4) [below right=of 5] {};
       \node [draw,circle,fill=red!50,"$14$"below] (11) [right=of 4] {};
        \node [draw,circle,fill=green!50,"$15$"below] (b1) [right=of 11] {};
         \node [draw,circle,fill=green!50,"$1$"below] (b2) [right=of b1] {};
       
       \node [draw,circle,fill=red!50,"$2$"below] (1) [right=of b2] {};
       \node [draw,circle,fill=red!50,"$3$"below] (2) [right=of 1] {};
       
        \node [draw,circle,fill=red!50,"$4$" below] (3) at (18,-3.5) {};

 \node [draw,circle,fill=blue!50,"$12$"] (12) at (12,0) {};
  
     \node [draw,circle,fill=red!50,"$7$"] (6) [right= of 12] {};
      \node [draw,circle,fill=red!50,"$8$"] (7) [right=of 6] {};
       \node [draw,circle,fill=red!50,"$9$"] (a1) [right=of 7] {};
        \node [draw,circle,fill=red!50,"$10$"] (a2) [right=of a1] {};
      
      \node [draw,circle,fill=red!50,"$11$"] (8) [right=of a2] {};
      \node [draw,circle,fill=red!50,"$12$"] (9) [right=of 8] {};
      \node [draw,circle,fill=red!50,"$13$"right] (10) [below right=of 9] {};
      \node [draw,circle,fill=red!50,"$6$" left] (5) [below left=of 6] {};

      \node [draw,circle,fill=blue!50,"$17$"] (13) [right=of 9] {};

       \node [draw,circle,fill=red!50,"$5$"below] (4) [below right=of 5] {};
       \node [draw,circle,fill=red!50,"$14$"below] (11) [right=of 4] {};
        \node [draw,circle,fill=red!50,"$15$"below] (b1) [right=of 11] {};
         \node [draw,circle,fill=red!50,"$1$"below] (b2) [right=of b1] {};
       
       \node [draw,circle,fill=red!50,"$2$"below] (1) [right=of b2] {};
       \node [draw,circle,fill=red!50,"$3$"below] (2) [right=of 1] {};
       
        \node [draw,circle,fill=red!50,"$4$" below] (3) at (18,-3.5) {};
       
       \node (a) at (18,-5.5) { (b) $Q_{A_{6}}$};

  \begin{scope}[>=Latex]
  
  \draw[-> , thick]  (1) edge (2); 
 
  \draw[-> , thick]  (9) edge (1);

   \draw[-> , thick]  (1) edge (10);

 \draw[-> , thick]  (2) edge (3);
 \draw[-> , thick]  (2) edge (8);
 \draw[-> , thick]  (2) edge (13);

 \draw[-> , thick]  (4) edge (11);
  \draw[-> , thick]  (12) edge (4);
   \draw[-> , thick]  (7) edge (4);
    \draw[-> , thick]  (3) edge (4);
    
 \draw[-> , thick]  (5) edge (12);
  \draw[-> , thick]  (5) edge (11);
   \draw[-> , thick]  (3) edge[bend left=35] (5);
    \draw[-> , thick]  (7) edge (5);
   
   \draw[-> , thick]  (6) edge (12);
    \draw[-> , thick]  (6) edge (7);
     \draw[-> , thick]  (11) edge (6);
      \draw[-> , thick]  (6) edge[bend left= 25] (3);

      \draw[-> , thick]  (8) edge (9);
       \draw[-> , thick]  (10) edge (8);
       
        \draw[-> , thick]  (13) edge (9);
         \draw[-> , thick]  (3) edge[bend left=25] (9);
         
          \draw[-> , thick]  (13) edge (10);
           \draw[-> , thick]  (10) edge[bend left= 35] (3);
           
            \draw[-> , thick]  (3) edge (13);
            
             \draw[-> , thick]  (12) edge (3);
        

 \draw[-> , thick]  (7) edge (a1);
  \draw[-> , thick]  (a1) edge (a2);
   \draw[-> , thick]  (a2) edge (8);
    \draw[-> , thick]  (11) edge (b1);
     \draw[-> , thick]  (b1) edge (b2);
      \draw[-> , thick]  (b2) edge (1);
      
      \draw[-> , thick]  (b1) edge (7);
     \draw[-> , thick]  (b2) edge (a1);
      \draw[-> , thick]  (1) edge (a2);

      \draw[-> , thick]  (a1) edge (11);
	\draw[-> , thick]  (a2) edge (b1);  
	\draw[-> , thick]  (8) edge (b2);

    \end{scope}

\end{tikzpicture}
}

\caption{Quiver corresponding to $\hat{B}_{A_{6}}$  }
\label{fig:MuQuiverA6}
\end{figure}

We then apply the mutation sequence $\mu_{3}\mu_{2}\mu_{1}\mu_{15}\mu_{14}\mu_{6}$ to this quiver.  If we arrange the nodes and edges of the mutated quiver as per Figure \ref{fig:arrangedQuiverA6}, then one can see that this is identical to the initial quiver except that specific labels are shifted by 1. Thus the block mutation $\mu_{3}\mu_{2}\mu_{1}\mu_{15}\mu_{14}\mu_{6}$ is equivalent to permuting the labels of the nodes in $Q_6$ and hence
\begin{equation}\label{eq:Qinv}
    \mu_{3}\mu_{2}\mu_{1}\mu_{15}\mu_{14}\mu_{6}(Q_{6}) = \rho_{6} (Q_{6})
\end{equation}
where $\rho_{6} = (3,4,5,6,7,8,9,10,11,12,13)$ is permutation of the labellings. Thus if we apply the mutations in the same order as previously, once again the structure of the quiver remains the same except the labels of the nodes are shifted. Thus if we take the inverse of the permutation on each side of \eqref{eq:Qinv}, then one has
\begin{equation}
\psi_{A_{6}}:=\rho_{6}^{-1}\mu_{3}\mu_{2}\mu_{1}\mu_{15}\mu_{14}\mu_{6} (Q_{A_{6}}) = Q_{A_{6}}
\end{equation}
and it is clear that the composition of mutations on the left-hand side is a cluster map. The first iteration of the map gives rise to a new seed, which contains cluster variables that are expressed by \eqref{eq:systm1} with $n=0$: 

\begin{equation}
\begin{split}
    &\psi_{A_{6}} : (q_{0},w_{0},\tau_{-1},\tau_{0},\tau_{1},\sigma_{0},\sigma_{1},\sigma_{2},\sigma_{3},\sigma_{4},\sigma_{5},\sigma_{6},\sigma_{7},p_{0},r_{0},a_{1},a_{6}) \\ 
    & \qquad \to (q_{1},w_{1},\tau_{0},\tau_{1},\tau_{2},\sigma_{1},\sigma_{2},\sigma_{3},\sigma_{4},\sigma_{5},\sigma_{6},\sigma_{7},\sigma_{8},p_{1},r_{1},a_{1},a_{6})    
\end{split}
\end{equation}
Notice that the subscript of the variables $q, w, \tau, \sigma$ is shifted by 1. Therefore successive applying the map $\psi_{A_{6}}$ will induce a series of seeds that consist of the cluster variables 
\begin{equation}(q_{n},w_{n},\tau_{n-1},\tau_{n},\tau_{ n+1},\sigma_{n} , \sigma_{n+1},\sigma_{n+2},\sigma_{n+3},\sigma_{n+4},\sigma_{n+5},\sigma_{n+6},\sigma_{n+7}, p_{n},r_{n},a_{1},a_{6})    
\end{equation}
for $n\in \mathbb{Z}$, satisfying \eqref{eq:systm1}.  Therefore every tau function, generated by the system of recurrences in \eqref{eq:systm1}, can be obtained by applying $\psi_{A_{6}}$ repeatedly. Hence, for each $n$, the tau functions $\tau_{n},\sigma_{n}, p_{n}, w_{n}, q_{n}, r_{n}$ are cluster variables belonging to the stated Laurent polynomial ring, by 
Theorem \ref{LP}. 

  \begin{figure}[H]
\centering
\resizebox{0.7 \textwidth}{!}{%
 \begin{tikzpicture}[every circle node/.style={draw,scale=0.6,thick},node distance=15mm]

 \node [draw,circle,fill=blue!50,"$16$"] (12) at (12,0) {};
  
     \node [draw,circle,fill=red!50,"$8$"] (6) [right= of 12] {};
      \node [draw,circle,fill=red!50,"$9$"] (7) [right=of 6] {};
       \node [draw,circle,fill=red!50,"$10$"] (a1) [right=of 7] {};
        \node [draw,circle,fill=red!50,"$11$"] (a2) [right=of a1] {};
      
      \node [draw,circle,fill=red!50,"$12$"] (8) [right=of a2] {};
      \node [draw,circle,fill=red!50,"$13$"] (9) [right=of 8] {};
      \node [draw,circle,fill=red!50,"$3$"right] (10) [below right=of 9] {};
      \node [draw,circle,fill=red!50,"$7$" left] (5) [below left=of 6] {};

      \node [draw,circle,fill=blue!50,"$17$"] (13) [right=of 9] {};

       \node [draw,circle,fill=red!50,"$6$"below] (4) [below right=of 5] {};
       \node [draw,circle,fill=red!50,"$14$"below] (11) [right=of 4] {};
        \node [draw,circle,fill=red!50,"$15$"below] (b1) [right=of 11] {};
         \node [draw,circle,fill=red!50,"$1$"below] (b2) [right=of b1] {};
       
       \node [draw,circle,fill=red!50,"$2$"below] (1) [right=of b2] {};
       \node [draw,circle,fill=red!50,"$4$"below] (2) [right=of 1] {};
       
        \node [draw,circle,fill=red!50,"$5$" below] (3) at (18,-3.5) {};

  \begin{scope}[>=Latex]
  
  \draw[-> , thick]  (1) edge (2); 
 
  \draw[-> , thick]  (9) edge (1);

   \draw[-> , thick]  (1) edge (10);

 \draw[-> , thick]  (2) edge (3);
 \draw[-> , thick]  (2) edge (8);
 \draw[-> , thick]  (2) edge (13);

 \draw[-> , thick]  (4) edge (11);
  \draw[-> , thick]  (12) edge (4);
   \draw[-> , thick]  (7) edge (4);
    \draw[-> , thick]  (3) edge (4);
    
 \draw[-> , thick]  (5) edge (12);
  \draw[-> , thick]  (5) edge (11);
   \draw[-> , thick]  (3) edge[bend left=35] (5);
    \draw[-> , thick]  (7) edge (5);
   
   \draw[-> , thick]  (6) edge (12);
    \draw[-> , thick]  (6) edge (7);
     \draw[-> , thick]  (11) edge (6);
      \draw[-> , thick]  (6) edge[bend left= 25] (3);

      \draw[-> , thick]  (8) edge (9);
       \draw[-> , thick]  (10) edge (8);
       
        \draw[-> , thick]  (13) edge (9);
         \draw[-> , thick]  (3) edge[bend left=25] (9);
         
          \draw[-> , thick]  (13) edge (10);
           \draw[-> , thick]  (10) edge[bend left= 35] (3);
           
            \draw[-> , thick]  (3) edge (13);
            
             \draw[-> , thick]  (12) edge (3);
        

 \draw[-> , thick]  (7) edge (a1);
  \draw[-> , thick]  (a1) edge (a2);
   \draw[-> , thick]  (a2) edge (8);
    \draw[-> , thick]  (11) edge (b1);
     \draw[-> , thick]  (b1) edge (b2);
      \draw[-> , thick]  (b2) edge (1);
      
      \draw[-> , thick]  (b1) edge (7);
     \draw[-> , thick]  (b2) edge (a1);
      \draw[-> , thick]  (1) edge (a2);

      \draw[-> , thick]  (a1) edge (11);
	\draw[-> , thick]  (a2) edge (b1);  
	\draw[-> , thick]  (8) edge (b2);

    \end{scope}

\end{tikzpicture}
}
\caption{ Mutated quiver $Q'_{A_{6}} = \mu_{3}\mu_{2}\mu_{1}\mu_{15}\mu_{14}\mu_{6} (Q_{A_{6}})$. It has the same structure as Figure \ref{fig:MuQuiverA6} with permuted labellings.}
\label{fig:arrangedQuiverA6}
\end{figure}

\end{proof}

\subsection{Local expansion}\label{ss:local-expansion}

To investigate generalizing the integrable cluster map to arbitrary even rank, we begin by exploring the relation between deformed quivers/exchange matrices of type $A_{4}$ and type $A_{6}$.  Recall that in Example \ref{LaurentA4}, we showed that Laurentification of deformed type $A_{4}$ lead us to a new cluster algebra defined by pair of the initial cluster  $(q_{0},\tau_{-1},\tau_{0},\tau_{1},\sigma_{0},\sigma_{1},\sigma_{2},\sigma_{3},\sigma_{4},\sigma_{5},p_{0},a_{1},a_{4})$ and the exchange matrix illustrated by Figure \ref{fig:ladderA4}. 

Comparison between $Q_{A_{4}}$ and $Q_{A_{6}}$ indicates that $ Q_{A_6}$ can be obtained from $Q_{A_{4}}$ by local expansion, as illustrated in Figure~\ref{fig:Localexpansion}, that is by removing edges between the four-cycle formed by the nodes $1$, $7$, $8$ and $11$ in $Q_{A_{4}}$ and including new nodes and edges in the quiver as per Figure~\ref{fig:Q4toQ6}.
\begin{figure}[H]
    \centering
    \begin{subfigure}{0.4\textwidth}
        \includegraphics[width=\hsize, valign=m]{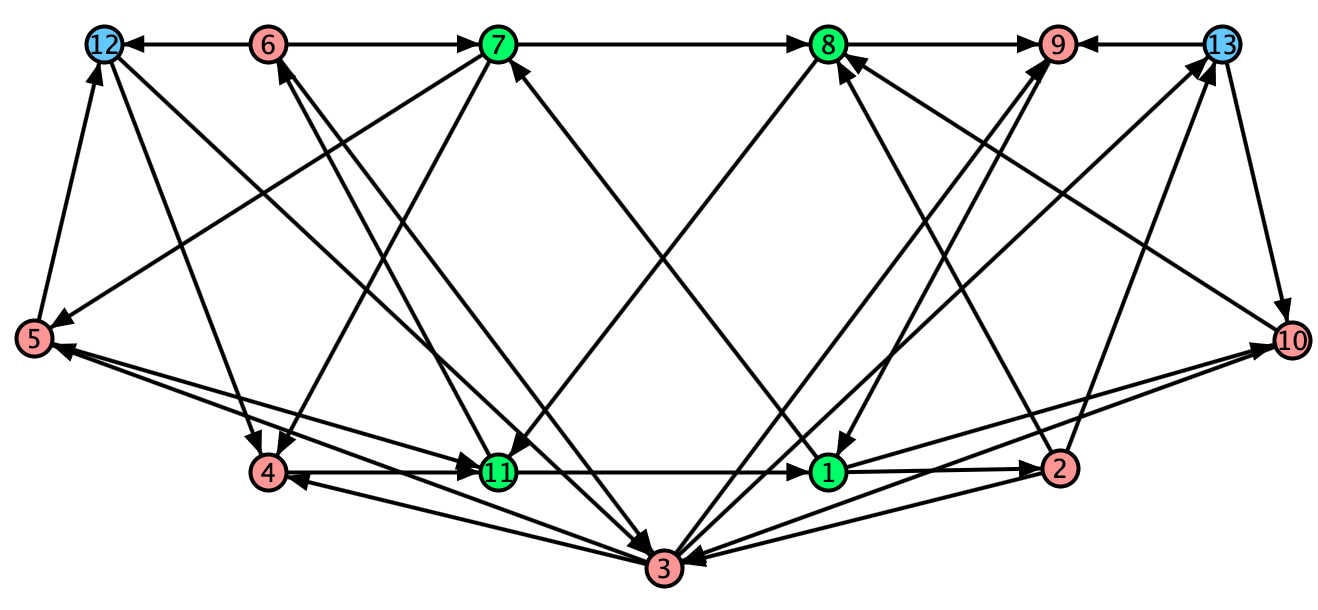}
        \caption{$Q_{A_{4}}$}
    \end{subfigure}
\qquad\tikz[baseline=-\baselineskip]\draw[thick,->] (0,1) -- ++ (1,0);\qquad
    \begin{subfigure}{0.4\textwidth}
    \includegraphics[width=\hsize, valign=m]{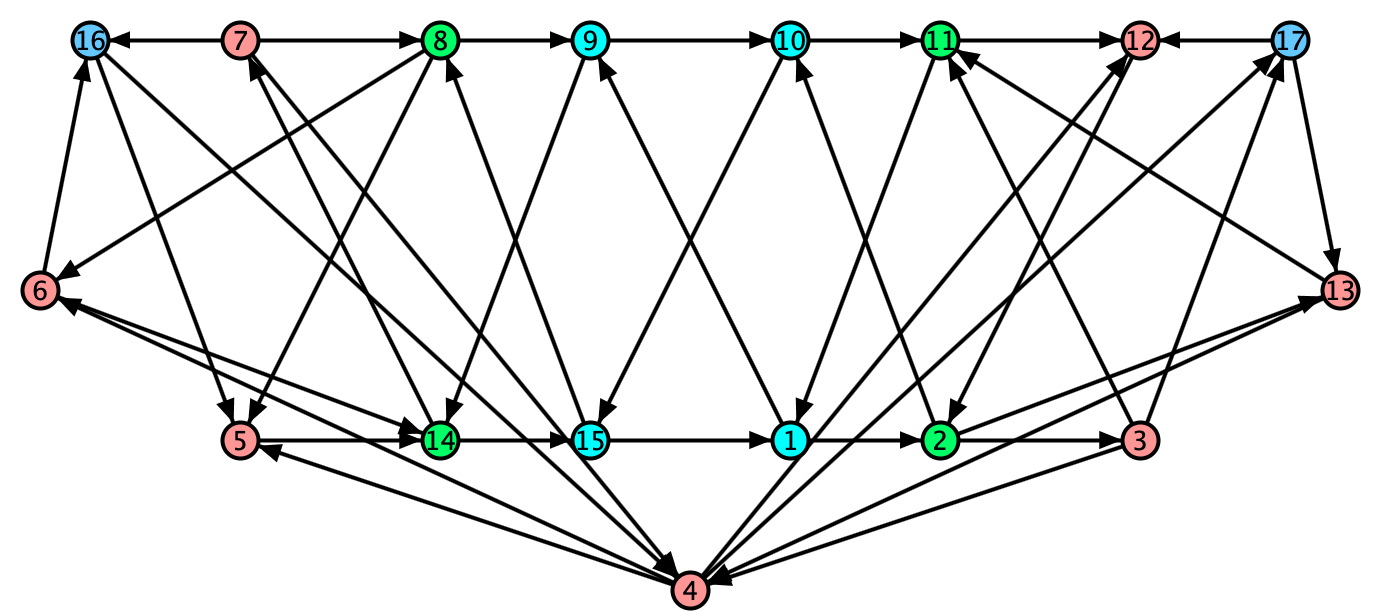}
    \caption{$Q_{A_{6}}$}
\end{subfigure}
    \caption{Extension from $Q_{A_{4}}$ to $Q_{A_{6}}$}
    \label{fig:Localexpansion}
\end{figure}
\begin{figure}[H]
    \centering
    \begin{subfigure}{0.34\textwidth}
        \includegraphics[width=\hsize, valign=m]{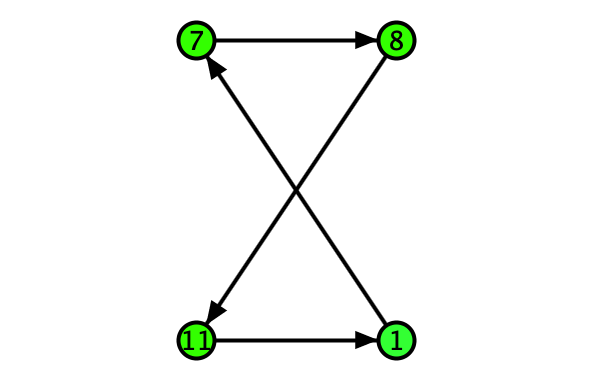}
        \caption{Subquiver in $Q_{A_{4}}$}
    \end{subfigure}
\qquad\tikz[baseline=-\baselineskip]\draw[ultra thick,->] (0,2) -- ++ (1,0);\qquad
    \begin{subfigure}{0.4\textwidth}
    \includegraphics[width=\hsize, valign=m]{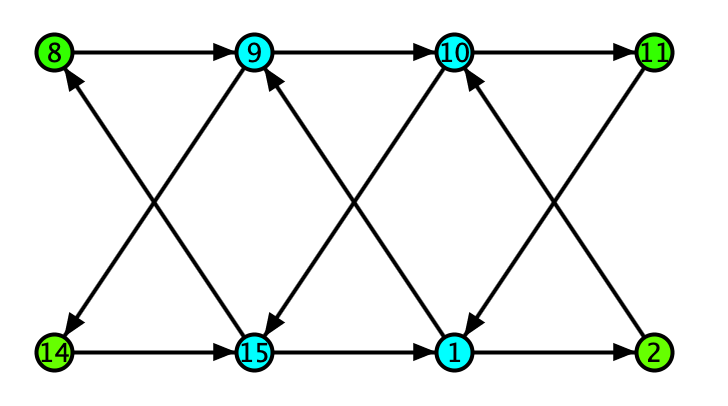}
    \caption{Subquiver in $Q_{A_{6}}$}
\end{subfigure}
    \caption{Local expansion of the subquiver in $Q_{A_{4}}$}
    \label{fig:Q4toQ6}
\end{figure}

Recall that each node in the deformed quiver corresponds to a tau function, e.g. for $Q_{A_{4}}$, the sequence of nodes $(1,2,3,4,5,6,7,8,9,10,11,12,13)$ corresponds to the sequence of functions $(q_{0},\tau_{-1},\tau_{0},\tau_{1},\sigma_{0},\sigma_{1},\sigma_{2},\sigma_{3},\sigma_{4},\sigma_{5},p_{0},a_{1},a_{4})$. Figure \ref{fig:Q4toQ6} shows that  $Q_{A_{6}}$ can be built from $Q_{A_{4}}$ by carrying out the local expansion on the four-cycle subquiver with nodes corresponding to the functions $\sigma_{3}$, $\sigma_{4}$, $q_{0}$ and $p_{0}$. From the cluster point of view, the local expansion is equivalent to relabelling $\sigma_3$, $\sigma_4$, $\sigma_5$ as $\sigma_{5},\sigma_{6},\sigma_{7}$ respectively and inserting  $\sigma_3$, $\sigma_4$, $p_{1}$ and $q_{1}$ in a way that the cluster becomes $(q_{1}, q_{0},\tau_{-1},\tau_{0},\tau_{1},\sigma_{0},\sigma_{1},\sigma_{2},\sigma_{3},\sigma_{4},\sigma_{5}\sigma_{6},\sigma_{7},p_{0}, p_{1}, a_{1},a_{4}) $.

We will show that this pattern continues: one can recursively apply the same local expansion by a four-cycle quiver to obtain the deformed quiver  $Q_{A_{2N}}$ with nodes corresponding to 
\begin{align*}
    (q_{N-2}\dots&,q_{1},q_{0},\tau_{-1},\tau_{0},\tau_{1},\sigma_{0},\sigma_{1},\sigma_{2},\sigma_{3},\dots, \sigma_{2N+1},p_{0},p_{1},\dots, p_{N-2},a_{1},a_{2N}) \\
    &=(1,2,3,\dots, 4N+3,4N+4,4N+5)
\end{align*}

What does this expansion tell us? The local expansion above gives insight into the structure of the tau functions in the $x_{i}$ variables. Let us compare the tau functions in type $A_4$ and type $A_6$ cases. In the setting $(\tilde{x}_{1},\tilde{x}_{2},\dots,\tilde{x}_{11}) = (q_{0},\tau_{-1},\tau_{0},\tau_{1},\sigma_{0},\dots,\sigma_{5},p_{0})$, the variables $x_{i,n}$, induced by the deformed map associated to type $A_4$, are defined as 
\begin{align*}
    x_{1,n} = \frac{\tilde{x}_{5}\tilde{x}_{4}}{\tilde{x}_{6}\tilde{x}_{3}}, \quad  x_{2,n} = \frac{\tilde{x}_{11}}{\tilde{x}_{7}\tilde{x}_{3}}, \quad  x_{3,n} = \frac{\tilde{x}_{1}}{\tilde{x}_{8}\tilde{x}_{3}}, \quad  x_{4,n} = \frac{\tilde{x}_{10}\tilde{x}_{2}}{\tilde{x}_{9}\tilde{x}_{3}}
\end{align*}
The local expansion above $(q_{0},\tau_{-1},\tau_{0},\tau_{1},\sigma_{0},\dots,\sigma_{5},p_{0}) \to \qty(q_{1},q_{0},\tau_{-1},\tau_{0},\tau_{1},\sigma_{0},\dots,\sigma_{7},p_{0},p_{1})$ is equivalent to shifting the subscript of the variables $\tilde{x}_{i} \to \tilde{x}_{i+1}$ for $i = 1,2,\dots,7,8$ and $\tilde{x}_{j} \to \tilde{x}_{j+3}$ for $i = 8,9, 10,11$ and imposing the new variables 
\begin{align*}
    x_{3} = \frac{\tilde{x}_{15}}{\tilde{x}_{9}\tilde{x}_{4}}, \quad x_{4} = \frac{\tilde{x}_{1}}{\tilde{x}_{10}\tilde{x}_{4}}
\end{align*}
Then one obtains the variable transformation in \eqref{vartransA6} whose tau functions are denoted as $\qty( \tilde{x}_{1} , \tilde{x}_{2},\dots,\tilde{x}_{15}) = \qty(q_{1},q_{0},\tau_{-1},\tau_{0},\tau_{1},\sigma_{0},\dots,\sigma_{7},p_{0},p_{1})$. The recursive local expansion constructs the following $x_i$ variables associated to the type $A_{2N}$ deformed map, 
\begin{equation}\label{vartransA2N}
    \begin{split}
    &x_{1} = \frac{\tilde{x}_{N+3}\tilde{x}_{N+2}}{\tilde{x}_{N+4}\tilde{x}_{N+1}}, \quad x_{2} = \frac{\tilde{x}_{3N+5}}{\tilde{x}_{N+5}\tilde{x}_{N+1}},\quad x_{3} = \frac{\tilde{x}_{3N+6}}{\tilde{x}_{N+6}\tilde{x}_{N+1}}, \dots, \\ & x_{N} = \frac{\tilde{x}_{4N+3}}{\tilde{x}_{2N+3}\tilde{x}_{N+1}}, x_{N+1} = \frac{\tilde{x}_{1}}{\tilde{x}_{2N+4}\tilde{x}_{N+1}}, \quad x_{N+2} = \frac{\tilde{x}_{2}}{\tilde{x}_{2N+5}\tilde{x}_{N+1}},\dots , \\ & x_{2N-1} = \frac{\tilde{x}_{N-1}}{\tilde{x}_{3N+2}\tilde{x}_{N+1}}, \quad x_{2N} = \frac{\tilde{x}_{3N+4}\tilde{x}_{N}}{\tilde{x}_{3N+3}\tilde{x}_{N+1}}
    \end{split}
    \end{equation}
where \[ \qty(\tilde{x}_{1},\tilde{x}_{2},\dots ,\tilde{x}_{4N+3}) =\qty(q_{N-2}\dots,q_{1},q_{0},\tau_{-1},\tau_{0},\tau_{1},\sigma_{0},\sigma_{1},\sigma_{2},\sigma_{3},\dots, \sigma_{2N+1},p_{0},p_{1},\dots, p_{N-2},a_{1},a_{2N}). \]

The symplectic form associated to $A_{2N}$ is defined by 
\begin{align*}
    \omega = \sum_{i<j} b_{ij} \dd \log x_{i} \wedge \dd \log x_{j} 
\end{align*}
where $b_{ij}$ are entries of the exchange matrix $B_{A_{2N}}$ \eqref{exchA2N}, defined by the following
\begin{align*}
    (B_{A_{2N}})_{ij} = \begin{cases}
        1 & \text{if} \ j=i+1 \\
        -1& \text{if} \ i=j+1 \\
        0 & \text{otherwise}
    \end{cases} 
\end{align*}

Now,
\begin{align*}
    \pi^{*}\omega 
    &= \tilde{\omega} \\
    &= \dd \log \qty(\frac{\tilde{x}_{N+3}\tilde{x}_{N+2}}{\tilde{x}_{N+4}\tilde{x}_{N+1}}) \wedge \dd \log \qty(\frac{\tilde{x}_{3N+5}}{\tilde{x}_{N+5}\tilde{x}_{N+1}}) + \dd \log\qty(\frac{\tilde{x}_{3N+5}}{\tilde{x}_{N+5}\tilde{x}_{N+1}}) \wedge  \log \qty(\frac{\tilde{x}_{3N+6}}{\tilde{x}_{N+6}\tilde{x}_{N+1}})\\ 
    &  + \sum_{l=6}^{N+2} \dd \log \qty(\frac{\tilde{x}_{3N+l}}{\tilde{x}_{N+l}\tilde{x}_{N+1}}) \wedge \dd \log \qty(\frac{\tilde{x}_{3N+(l+1)}}{\tilde{x}_{N+(l+1)}\tilde{x}_{N+1}}) + \dd \log \qty(\frac{\tilde{x}_{4N+3}}{\tilde{x}_{2N+3}\tilde{x}_{N+1}}) \wedge \dd \log \qty(\frac{\tilde{x}_{1}}{\tilde{x}_{2N+4}\tilde{x}_{N+1}})\\
    &+ \sum_{m}\dd \log \qty(\frac{\tilde{x}_{m}}{\tilde{x}_{2N+3 + m}\tilde{x}_{N+1}}) \wedge \dd \log \qty(\frac{\tilde{x}_{m+1}}{\tilde{x}_{2N+3 + (m+1)}\tilde{x}_{N+1}}) \\ 
    &+ \dd \log \qty(\frac{\tilde{x}_{N-1}}{\tilde{x}_{3N+2}\tilde{x}_{N+1}}) \wedge \dd \log \qty(\frac{\tilde{x}_{3N+4}\tilde{x}_{N}}{\tilde{x}_{3N+3}\tilde{x}_{N+1}})
\end{align*}
To simplify the calculation,  let us define $\alpha_{i} = \dd \log \tilde{x}_{i}$ and $f_{j} = \alpha_{3N+j} - \alpha_{N+j}$ and $g_{k} = \alpha_{k} - \alpha_{k+2N+3}$. Then the pre-symplectic form can be re-written as
\begin{align*}
    \tilde{\omega} &= (\alpha_{N+3} + \alpha_{N+2} - \alpha_{N+4} - \alpha_{N+1}) \wedge \qty(f_{5} - \alpha_{N+1}) \\ 
     &+ \qty(\sum_{l=5}^{N+2} f_{l} \wedge f_{l+1} - f_{l} \wedge \alpha_{N+1} - \alpha_{N+1} \wedge f_{l+1}) + f_{4N+3} \wedge g_{1}- \alpha_{N+1} \wedge g_{1} \\ 
     &+ \qty(\sum_{m=1}^{N-2} g_{m} \wedge g_{m+1} - g_{m} \wedge \alpha_{N+1} - \alpha_{N+1} \wedge g_{m+1}) \\ 
     &+ (g_{N-1} + \alpha_{N+1})\wedge (g_{N} - g_{N+1})
\end{align*}
Combining and cancelling, we obtain 
\begin{align*}
    \tilde{\omega} &= (\alpha_{N+3} + \alpha_{N+2} - \alpha_{N+4}) \wedge \qty(f_{5} - \alpha_{N+1})  \\ 
     &+ \qty(\sum_{l=5}^{N+2} f_{l} \wedge f_{l+1} ) +  f_{4N+3} \wedge g_{1}+ \qty(\sum_{m=1}^{N-2} g_{m} \wedge g_{m+1} ) \\ 
     &+ g_{N-1}\wedge \qty(\alpha_{3N+4} + g_{N}) - \alpha_{N+1}\wedge \qty(\alpha_{3N+4} + g_{N}) - g_{N-1} \wedge \alpha_{N+1}
\end{align*}      
Therefore $\tilde{\omega}$ is expressed as 
\begin{align*}
    \sum_{r<s} \tilde{b}_{rs} \alpha_{r} \wedge \alpha_{s}
\end{align*}
whose coefficients are entries of the $(4N+3) \times (4N+3)$ exchange matrix 
\begin{equation}\label{deformmatrixA2N}
    B_{A_{2N}} =
  \left(\begin{array}{@{}c|c@{}}
  \vb{A}_{2N} & \vb{B}_{2N} \\\hline
   -\vb{B}_{2N}^{T}& \vb{C}_{2N} 
  \end{array}\right)
\end{equation}
which is composed of four block skew-symmetric matrices: a $(2N+3) \times (2N+3)$ matrix $\vb{A}_{2N}$, a $(2N + 3)\times 2N$ matrix $\vb{B}_{2N}$ and a $2N\times 2N$ matrix $\vb{C}_{2N}$ where
\begin{align*}
    &\scalemath{0.8}{\vb{A}_{2N}/ \vb{C}_{2N} =
\left(\begin{array}{@{}ccc|cccc|ccc@{}}
   0 & 1 &0  & 0& & & & \cdots& 0& 1  \\
    -1 & \ddots &\ddots  & & & & &  \\& 
    \ddots & 0 & 1 & 0 & \cdots  & 0 & 0 \\ \hline &  0 & -1 &  & &  & &   0 \\ &  0 & 0 &  & &   \vb{A}_{4} / \vb{C}_{4} & & 0 \\ & & \vdots & &  & & & \vdots \\  & & 0 & & & & & 1 \\ \hline  & & 0 &0 & 0& \cdots& -1 &0 &\ddots & 0 \\
    0 & & 0 &0 & 0& \cdots& 0 &\ddots & \ddots & 1\\
      -1& & 0 &0 & 0& \cdots& 0 &0 & -1 & 0  \\
  \end{array}\right)},\quad 
  \scalemath{0.8}{\vb{B}_{2N} =
  \left(\begin{array}{@{}ccc|cccc|ccc@{}}
   0 & -1 &0  & 0& & & & \cdots& 0& -1  \\
    1 & \ddots &\ddots  & & & & &  \\& 
    \ddots & 0 & -1 & 0 & \cdots  & 0 & 0\\ \hline &  0 & 1 &  & &  & &   0 \\ &  0 & 0 &  & &   \vb{B}_{4} & & 0 \\ & & \vdots & &  & & & \vdots \\  & & 0 & & & & & -1 \\ \hline  & & 0 &0 & 0& \cdots& 1 &0 &\ddots & 0 \\ 0
     & & 0 &0 & 0& \cdots& 0 &\ddots & \ddots & -1\\1
      & & 0 &0 & 0& \cdots& 0 &0 & -1 & 0  \\
  \end{array}\right)} \\[1em]
  \end{align*}
Here, the block matrix has $(\vb{A}_{2N})_{ij}=(\vb{A}_{4})_{ij}$ for $N-2 < i,j < N+6 $, $(\vb{B}_{2N})_{mn}=(\vb{B}_{4})_{mn}$ for $N < m < N+6 $ and $3(N-2) + 7 < n < 3(N-2) + 12$ and $(\vb{C}_{2N})_{rs}=(\vb{C}_{4})_{rs}$ for $3(N-2) +7 <r,s < 3(N-2) + 12$, and
\begin{align*}
\vb{A}_{4} = \mqty(0 & 1 & 0 & 0 & 0 & 0 & 0 \\ -1 & 0 & 1 & 0 & 0 & 0 & 0 \\ 0 & -1 & 0 & 1 & 1 & -1 & 0 \\ 0 & 0 & -1 & 0 & 0 & 0 & -1 \\ 0 & 0 & -1 & 0 & 0 & 0 & -1 \\ 0 & 0 & 1 & 0 & 0 & 0 & 1 \\ 0 & 0 & 0 & 1 & 1 & -1 & 0 ), \qquad \vb{B}_{4} = \mqty(0 & -1 & 1 & 0 \\ 1 & 0 & 0 & 0 \\ 0 & 1 & -1 & 0 \\ 0 & 0 & 0 & 1 \\ 0 & 0 & 0 & 1\\ 0 & 0 & 0 & -1\\ 0 & 0 & 0 & 0), \qquad \vb{C}_{4} = \mqty(0 & 1 & -1 & 0 \\ -1 & 0 & 0 & 0\\ 1 & 0 & 0 & 0\\ 0 & 0 & 0 & 0 )
\end{align*}
Based on the pattern of local expansion from $Q_{A_{4}}$ to $Q_{A_{6}}$, we introduce frozen variables and extend the matrix by extra rows $\vb{b}_{1}$ and $\vb{b}_{2}$ in which entries are 
\begin{equation}
\begin{split}
     &(\vb{b}_{1})_{i}= \delta_{i,N+1} + \delta_{i,N+2} - \delta_{i,N+3} - \delta_{i,N+4}\\[1em]
     &(\vb{b}_{2})_{i}= -\delta_{i,N} - \delta_{i,N+1} + \delta_{i,3N+3} + \delta_{i,3N+4}
    \end{split}
\end{equation}
so that
\begin{equation}\label{exchmA2N}
\begin{split}
    \tilde{B}_{A_{2N}} = \includegraphics[scale=0.45,height=7cm,valign=c]{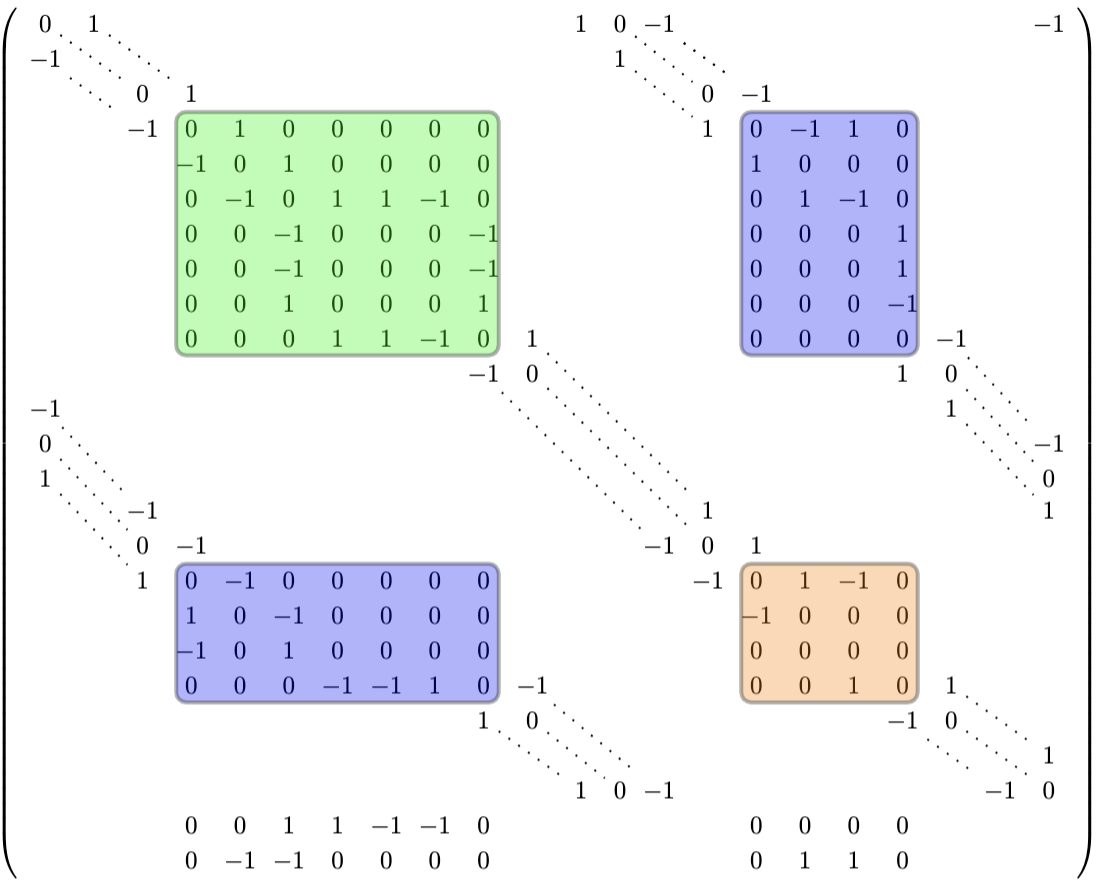}
\end{split}
  \end{equation}
  

By this calculation and formalisation, we have obtained a description of local expansion in terms of exchange matrices, corresponding to the graphical description in terms of quivers.

For the example of $A_{4}$ to $A_{6}$, the following exchange matrix 
\begin{equation}
    \tilde{B}_{A_{4}} =
  \left(\begin{array}{@{}c|c@{}}
   \vb{A}_{4} & \vb{B}_{4} \\\hline
   -\vb{B}_{4}^{T}& \vb{C}_{4} 
  \end{array}\right)
  \end{equation}
represents a quiver in which the edges of a subquiver (which is the four-cycle) have been removed i.e.\ the beginning and end of off-diagonal entries in each block matrix are set to be zero.  The local expansion in Figure~\ref{fig:Localexpansion} is equivalent to inserting new columns and rows in each block matrix, i.e. 
\begin{align*}
    &\vb{A}_{6} =
  \left(\begin{array}{@{}c|cccc|c@{}}
    0 & 1 & 0 & \cdots  & 0 & 1 \\ \hline  -1 &  & &  & &   0 \\ 0 & & &   \vb{A}_{4} & & 0 \\ \vdots & &  & & & \vdots \\ 0 & & & & & 1 \\ \hline -1 &0 & 0& \cdots& -1&0
  \end{array}\right), \qquad 
   \vb{B}_{6} =
  \left(\begin{array}{@{}c|cccc|c@{}}
    0 & -1 & 0 & \cdots  & 0 & -1 \\ \hline  1 &  & &  & &   0 \\ 0 & & \vb{B}_{4} &   & & 0 \\ \vdots & &  & & & \vdots \\ 0 & & & & & -1 \\ \hline 1 &0 & 0& \cdots& 1&0
  \end{array}\right) \\[1em]
  & \vb{C}_{6} =
  \left(\begin{array}{@{}c|cccc|c@{}}
    0 & 1 & 0 & \cdots  & 0 & 1 \\ \hline  -1 &  & &  & &   0 \\ 0 & & &   \vb{C}_{4} & & 0 \\ \vdots & &  & & & \vdots \\ 0 & & & & & 1 \\ \hline -1 &0 & 0& \cdots& -1&0
  \end{array}\right)
  \end{align*}
Such extension leads to the exchange matrix \eqref{exchmA6}, now written as 
\begin{equation}
    \tilde{B}_{A_{6}} =
  \left(\begin{array}{@{}c|c@{}}
   \vb{A}_{6} & \vb{B}_{6} \\\hline
   -\vb{B}_{6}^{T}& \vb{C}_{6} \\
  \end{array}\right)
  \end{equation}
%

\subsection{The deformed type $A_{2N}$  cluster map}
\label{ss:type-A2N-Laurent}

Earlier we proved that the Laurent property of type $A_{6}$ deformed map can be restored by lifting it to a higher-dimensional cluster map defined on the space of tau functions, which is done by finding the particular sequence of mutations preserving the structure of the quiver up to a permutation of the labels. In this section, we use a similar procedure to show that there exists a sequence of mutations such that the structure of the candidate deformed quiver $Q_{A_{2N}}$ (or exchange matrix $\tilde{B}_{A_{2N}}$, as constructed in the previous section) is preserved and show that the corresponding cluster variables can be produced by a two-parameter family of deformed cluster maps corresponding to type $A_{2N}$.

%
\begin{exmp}[Deformed quiver $Q_{A_{8}}$]\label{exmptypeA8}
Let us consider the initial seed formed by the initial cluster  $\qty(\tilde{x}_{1},\tilde{x}_{2},\dots ,\tilde{x}_{4N+3}) =\qty(q_{2},q_{1},q_{0},\tau_{-1},\tau_{0},\tau_{1},\sigma_{0},\sigma_{1},\dots, \sigma_{9},p_{0},p_{1},p_{2},a_{1},a_{2N})$ and exchange matrix
\begin{equation}\label{exchmA8}
    \tilde{B}_{A_{8}} =
  \left(\begin{array}{@{}c|c@{}}
   \vb{A}_{8} & \vb{B}_{8} \\\hline
   -\vb{B}_{8}^{T}& \vb{C}_{8} \\\hline
   \vb{b}^{(1)}_{8} &\vb{b}^{(2)}_{8}
  \end{array}\right) = \includegraphics[scale=0.45,height=7cm,valign=c]{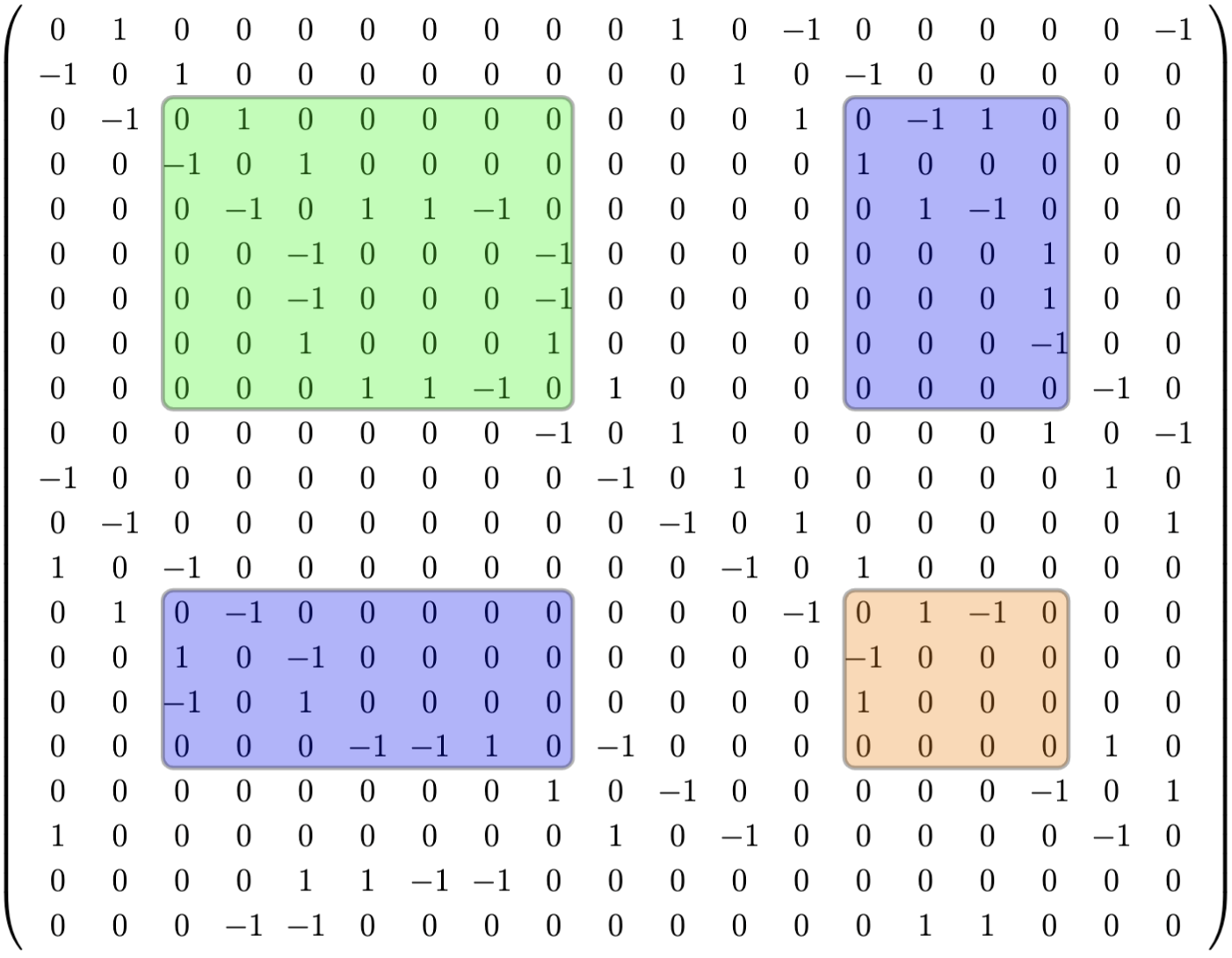}
  \end{equation}
where 
\begin{align*}
    &\scalemath{0.8}{\vb{A}_{8} / \vb{C}_{8} =
  \left(\begin{array}{@{}cc|cccc|cc@{}}
  0 & 1 & 0 & 0 & \cdots  & 0 & 0&  1 \\
    -1 & 0 & 1 & 0 & \cdots  & 0 & 0&  0 \\ \hline  0 & -1 &  & &  & & 0&  0 \\ 0 & 0 & & &   \vb{A}_{4} / \vb{C}_{4} & & 0& 0 \\ 0 & \vdots & & & & & \vdots & \vdots\\0 & 0 & & & & & 1 &  0\\ \hline 0 & 0 &0 & 0& \cdots& -1& 0& 1\\
    -1& 0 & 0 & 0 & \cdots  & 0 &  -1 &  0 \\
  \end{array}
  \right)}, \qquad 
  \scalemath{0.8}{ \vb{B}_{8} =
  \left(\begin{array}{@{}cc|cccc|cc@{}}
  0 & -1 & 0 & 0 & \cdots  & 0 & 0&  -1 \\
    1 & 0 & -1 & 0 & \cdots  & 0 & 0&  0 \\ \hline  0 & 1 &  & &  & & 0&  0 \\ 0 & 0 & &\vb{B}_{4} &    & & 0& 0 \\ 0 & \vdots & \phantom{-1+1} &  & & & \vdots & \vdots\\0 & 0 & & & & & -1 &  0\\ \hline 0 & 0 &0 & 0& \cdots& 1& 0& -1\\
    1& 0 & 0 & 0 & \cdots  & 0 &  1 &  0 \\
  \end{array}\right)} \\[1em]
  & \vb{b}^{(1)}_{8} + \vb{b}^{(2)}_{8} = \left( \begin{array}{c c c c c c c c c c c c c c  c c c c c }0 & 0&  0& 0& 1 & 1 & -1 & -1& 0 & 0 & 0 & 0 & 0 & 0 & 0 & 0 & 0 & 0 & 0 \\ 0 & 0 & 0 & -1 & -1& 0 & 0 & 0 & 0 & 0 & 0 & 0 & 0 & 0 & 1 & 1 & 0& 0 & 0 
  \end{array} \right)\\
  \end{align*}
Reading off the exchange matrix,  the corresponding deformed quiver $Q_{A_{8}}$ can be drawn and this is depicted in Figure~\ref{fig:QuiverA8}. 
 \begin{figure}[h]
\centering
\includegraphics[scale=0.4]{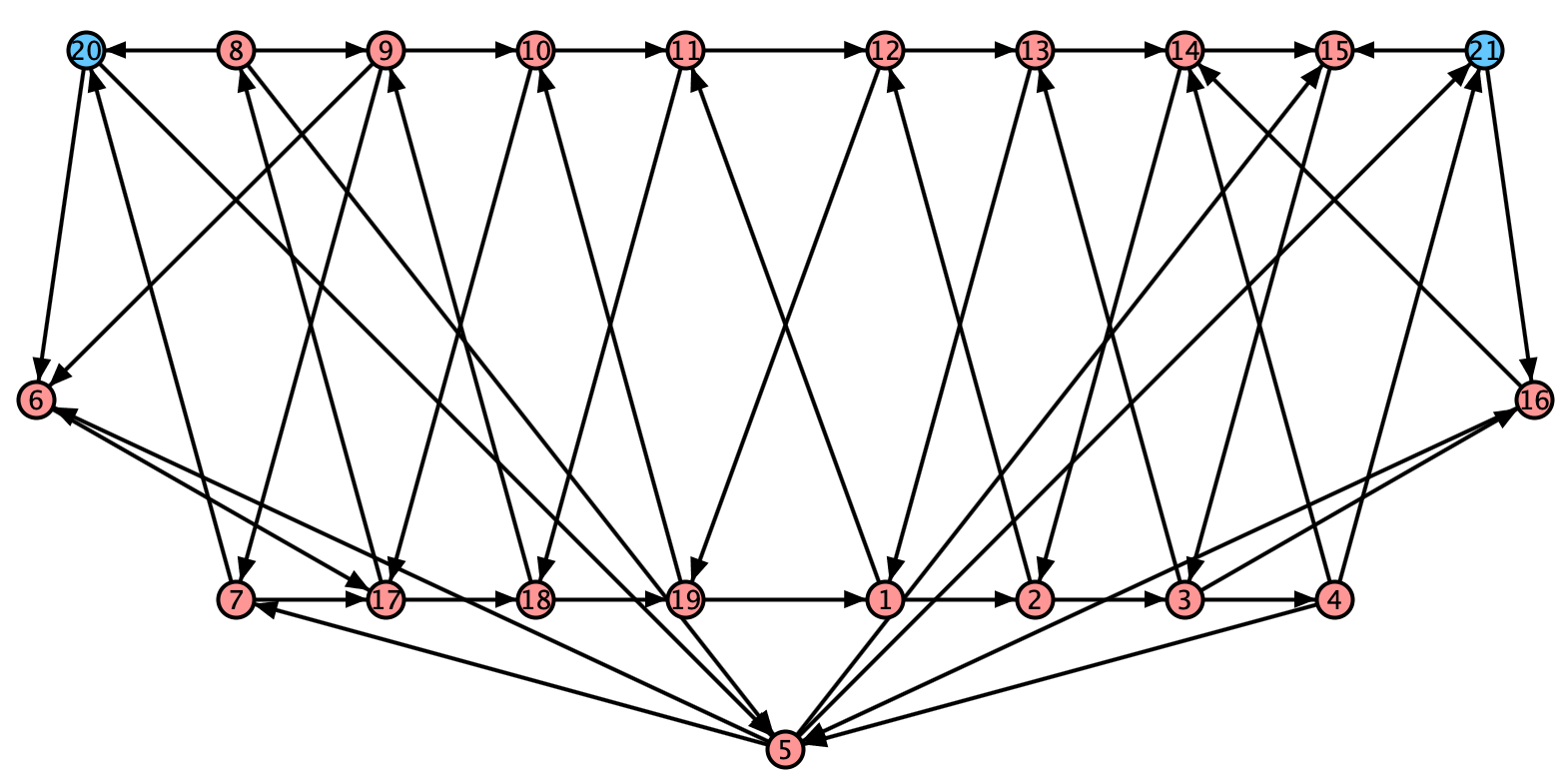}
\caption{ (Candidate) deformed quiver $Q_{A_{8}}$}
\label{fig:QuiverA8}
\end{figure}
Recall that the composition of mutations $\mu_{3}\mu_{2}\mu_{1}\mu_{15}\mu_{14}\mu_{6}$ maintains the form of the deformed quiver $Q_{A_{6}}$ except that the particular labellings of the nodes are permuted. Such mutation periodicity was already observed in the type $A_{4}$ case, in which the relevant sequence of mutations is $\mu_{2}\mu_{1}\mu_{11}\mu_{5}$. 
Comparing the cases, one can deduce the pattern of mutations for the type $A_{8}$, which is given by $\mu_{4}\mu_{3}\mu_{2}\mu_{1}\mu_{19}\mu_{18}\mu_{17}\mu_{7}$.  Then performing the iteration of matrix mutations above gives rise to the following exchange matrix:    
\begin{equation}\label{transformedexchmA8}
\includegraphics[scale=0.45,height=7cm,valign=c]{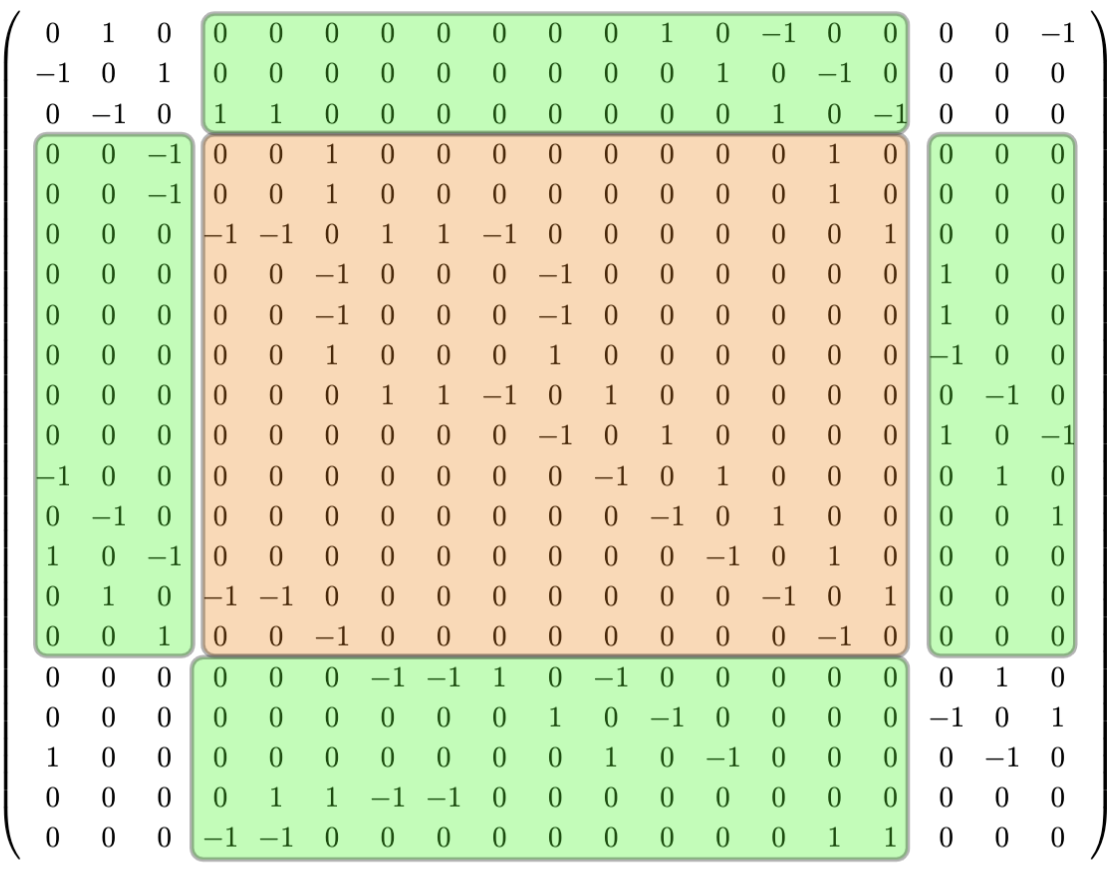}
  \end{equation}
Let us compare the mutated matrix \eqref{transformedexchmA8} with \eqref{exchmA8}. Then one can see that there have been changes in certain regions in the matrix, highlighted above. In the submatrix defined by the orange region, the transformation described by mutations is equivalent to cyclic permuting of the matrix, 
i.e.\ for the $13$-cycle $\rho=(4,5,6,7,8,9,10,11,12,13,14,15,16)$ cyclic permutation, the entries satisfy 
\begin{align*}
    &  b_{\rho(i)\rho(j)} =  b_{i+1,j+1} \quad  \text{for} \ 4\leq i,j\leq 16,
\end{align*}
with 
\begin{equation} \label{setting1}
    b_{i,\rho(16)} = b_{i,1} \ (\text{or} \ b_{\rho(16),j} = b_{1,j}) \ \text{for} \ i,j\in\qty{4,5,\dots,16}.
\end{equation}
 As for the green highlighted submatrices, the entries in the upper and lower rows are shifted to the right by 1. For the left and right end column vectors, the entries are moved downwards by 1, i.e letting $I_{1}=\qty{1,2,3}$ and $I_{2}=\qty{17,18,19}$ and then 
\begin{align*}
b_{l,\rho(j)} = b_{l,j+1} \quad \ \text{for} \ l\in I_{1} \cup I_{2}, \ j\in \qty{4,\dots,16}
\end{align*}
in agreement with \eqref{setting1}. Such a transformation is seen to be the cyclic permutation of the labels $\qty{4,5,\dots,16}$ of the deformed quiver $Q_{A_{8}}$.

Thus by using the same procedure that appeared in the type $A_{6}$ case,  one can construct the cluster map $\psi_{A_{8}} = \rho_{A_{8}}^{-1}\mu_{4}\mu_{3}\mu_{2}\mu_{1}\mu_{19}\mu_{18}\mu_{17}\mu_{7}$ which generates the set of cluster variables 
\begin{equation}\label{eq:systmA8}
\begin{aligned}
    \tau_{n+2}\sigma_{n} &= \sigma_{n+2}\tau_{n} + a_{1}p_{n} \\
    p_{0,n+1}p_{0,n} &= \sigma_{n+3}\sigma_{n+2}\tau_{n}\tau_{n+1} + p_{1,n} \sigma_{n+1}\tau_{n+2} \\ 
    p_{1,n+1}p_{1,n} &= \sigma_{n+4}\sigma_{n+3}\tau_{n}\tau_{n+1} + p_{2,n}p_{0,n+1}\\
    p_{2,n+1}p_{2,n} &= \sigma_{n+5}\sigma_{n+4}\tau_{n}\tau_{n+1} + q_{2,n}p_{1,n+1} \\
    q_{2,n+1}q_{2,n} &= \sigma_{n+6}\sigma_{n+5}\tau_{n}\tau_{n+1} + q_{1,n}p_{2,n+1} \\
    q_{1,n+1}q_{1,n} &= \sigma_{n+7}\sigma_{n+6}\tau_{n}\tau_{n+1} + q_{1,n}q_{1,n+1} \\
    q_{0,n+1}q_{0,n} &= \sigma_{n+8}\sigma_{n+7}\tau_{n}\tau_{n+1} + \sigma_{n+9}q_{1,n+1}\tau_{n-1} \\
    \sigma_{n+9}\tau_{n-1} &= \sigma_{n+7}\tau_{n+1} + a_{8}q_{0,n+1}
    \end{aligned}
\end{equation}
Similarly to the situation of the type $A_6$ case,  the exchange matrix \eqref{exchmA8} can be constructed through pullback of the original symplectic form  by the rational map which is analogous to the variable transformation 
\begin{equation}\label{DVTA8}
    \begin{split}
       &x_{1,n} = \frac{\sigma_{n}\tau_{n+1}}{\sigma_{n+1}\tau_{n}} \quad  x_{2,n} = \frac{p_{0,n}}{\sigma_{n+2}\tau_{n}} \quad  x_{3,n} =\frac{p_{1,n}}{\sigma_{n+3}\tau_{n}} \quad  x_{4,n} =\frac{p_{2,n}}{\sigma_{n+4}\tau_{n}}\\[1em]
    & x_{5,n} =\frac{q_{2,n}}{\sigma_{n+5}\tau_{n}} \quad x_{6,n} =  \frac{q_{1,n}}{\sigma_{n+6}\tau_{n}} \quad x_{7,n} =\frac{q_{0,n}}{\sigma_{n+7}\tau_{n}} \quad x_{8,n} =\frac{\sigma_{n+9}\tau_{n-1}}{\sigma_{n+8}\tau_{n}} 
    \end{split}
\end{equation}
Then by manipulation of equations in \eqref{eq:systmA8} and imposing \eqref{DVTA8}, we obtain the original two-parameter family of maps $\tilde{\varphi}_{A_{8}}$ defined on the initial variables $(x_1,x_2,\dots, x_8)$. Hence it has been shown that the candidate deformed quiver/exchange matrix that emerged from the constructive approach allows us to define the cluster map $\psi_{A_{8}}$ which is a Laurentified deformed map $\tilde{\varphi}_{A_{8}}$.  
    
\end{exmp}

As we have seen in the example, the deformed map $\tilde{\varphi}_{A_{8}}$ can be lifted, through Laurentification, to the cluster map $\psi_{A_{8}}$ that preserves the deformed exchange matrix \eqref{exchmA8}. We now extend this procedure to show that the deformation of type $A_{2N}$ cluster maps $\varphi_{A_{2N}}$ can be lifted to a cluster map $\psi_{A_{2N}}$.

Following the cases of type $A_{4}$, $A_{6}$ and $A_{8}$, one can extrapolate that the quiver $Q_{A_{2N}}$ is mutation periodic under the specific sequence of mutation i.e. $\mu_{i_{1}}\mu_{i_{2}}\cdots\mu_{i_{2N}} (Q_{A_{2N}}) = \rho ( Q_{A_{2N}})$. 
By considering the mutations formed in the same way as the ones considered in type $A_{4}$, $A_{6}$ and $A_{8}$, we have the following statement. 
\begin{prop}\label{propnA2N}
    For each deformed quiver $Q_{A_{2N}}$ with vertices \begin{align*}
    (q_{N-2}\dots&,q_{1},q_{0},\tau_{-1},\tau_{0},\tau_{1},\sigma_{0},\sigma_{1},\sigma_{2},\sigma_{3},\dots, \sigma_{2N+1},p_{0},p_{1},\dots, p_{N-2},a_{1},a_{2N}) \\
    &=(1,2,3,\dots, 4N+3,4N+4,4N+5)
\end{align*}
    we have invariance up to the permutation $ \rho = (N, N+1, \dots, 3N+4) $ under mutation:
\begin{equation}\label{matrixmuA2N}
    \mu_{N}\tilde{\mu}\mu_{N+3} (Q_{A_{2N}}) = \rho(Q_{A_{2N}}) 
    \end{equation}
\end{prop}
where $\tilde{\mu} = \mu_{N-1}\cdots\mu_{1}\mu_{4N+3}\cdots \mu_{3N+6}\mu_{3N+5}$.
\begin{proof}

To explicitly observe the periodicity phenomenon, it is more convenient to approach it from the quiver perspective rather than the exchange matrix.
\begin{figure}[H]

\begin{center}
\resizebox{1\textwidth}{!}{%
 \begin{tikzpicture}[every circle node/.style={draw,scale=0.6,thick},node distance=10mm]

\node [draw,circle,fill=blue!50,"\footnotesize{$4N+4$}"] (aa1) at (0,0) {};
  \node [draw,circle,fill=red!50,"\footnotesize{$N+4$}"] (aa2) [right=of aa1] {};
  \node [draw,circle,fill=red!50,"\footnotesize{$N+5$}"] (aa3) [right=of aa2] {};
  \node [draw,circle,fill=green!50,"\footnotesize{$N+6$}"] (aa4) [right=of aa3] {};
  \node [draw,circle,fill=green!50,"\footnotesize{$N+7$}"] (aa5) [right=of aa4] {};
  \node [draw,circle,fill=red!50,"\footnotesize{$N+2$}" left]  (cc1) [below left=of aa1] {};
  \node  (aa6) [right=of aa5] {};

  \node [draw,circle,fill=red!50,"\footnotesize{$N+3$}" below] (bb1) at (1.25,-2.5) {} ; 
  \node [draw,circle,fill=red!50,"\footnotesize{$3N+5$}" below] (bb2) [right=of bb1] {};
  \node [draw,circle,fill=green!50,"\footnotesize{$3N+6$}" below] (bb3) [right=of bb2] {};
   \node [draw,circle,fill=green!50,"\footnotesize{$3N+7$}" below] (bb4) [right=of bb3] {};
   \node (bb5) [right=of bb4] {};

    \node [draw,circle,fill=red!50,"\footnotesize{$N+1$}" below] (dd1) at (10,-4) {} ;
    
  \begin{scope}[>=Latex]
  
\draw[-> , thick]  (aa1) edge (cc1);  
\draw[-> , thick]  (aa1) edge (dd1);  
\draw[-> , thick]  (aa2) edge (aa1);  
\draw[-> , thick]  (bb1) edge (aa1);  

\draw[-> , thick]  (cc1) edge (bb2);  
\draw[-> , thick]  (dd1) edge (cc1);  
\draw[-> , thick]  (aa1) edge (cc1);  
\draw[-> , thick]  (aa3) edge (cc1);  

\draw[-> , thick]  (bb1) edge (aa1);  
\draw[-> , thick]  (bb1) edge (bb2);  
\draw[-> , thick]  (dd1) edge (bb1);  
\draw[-> , thick]  (aa3) edge (bb1);  

\draw[-> , thick]  (aa2) edge (aa1);  
\draw[-> , thick]  (aa2) edge (dd1);  
\draw[-> , thick]  (aa2) edge (aa3);  
\draw[-> , thick]  (bb2) edge (aa2); 

\draw[-> , thick]  (aa3) edge (aa4);  
\draw[-> , thick]  (bb3) edge (aa3);  

\draw[-> , thick]  (aa4) edge (bb2);  
\draw[-> , thick]  (aa4) edge (aa5);  
\draw[-> , thick]  (bb4) edge (aa4); 

\draw[-> , thick]  (aa5) edge (bb3);
\draw[-> , thick]  (aa5) edge (aa6);

\draw[-> , thick]  (bb2) edge (bb3); 
\draw[-> , thick]  (bb3) edge (bb4);
\draw[-> , thick]  (bb4) edge (bb5);

\draw[-> , thick]  (aa6) edge (bb4);
\draw[-> , thick]  (bb5) edge (aa5);

\end{scope}

 \node [draw,circle,fill=green!50,"\footnotesize{$3N+1$}"] (a3) at (15.5,0) {};
  \node [draw,circle,fill=red!50,"\footnotesize{$3N+2$}"] (a4) [right=of a3] {};
  \node [draw,circle,fill=red!50,"\footnotesize{$3N+3$}"] (a5) [right=of a4] {};
  \node [draw,circle,fill=blue!50,"\footnotesize{$4N+5$}"] (a6) [right=of a5] {};
  \node [draw,circle,fill=red!50,"\footnotesize{$3N+4$}" right] (c2) [below right=of a6] {};
  \node (a1) [left=of a3] {};
   \node (a0) [left=of a1] {};
   \node (A0) [draw,circle,fill=green!50,"\footnotesize{$2N+5$}"] [left=of a0] {};
   \node (A1) [draw,circle,fill=green!50,"\footnotesize{$2N+4$}"][left=of A0] {};
   \node (A2) [draw,circle,fill=green!50,"\footnotesize{$2N+3$}"][left=of A1] {};
   \node (A3)[left=of A2] {};

  \node [draw,circle,fill=green!50,"\footnotesize{$N-2$}" below](b2) at (15.5,-2.5) {};
  \node (b1) [left=of b2] {};
  \node (b0) [left=of b1] {};
  \node (B0)[draw,circle,fill=green!50,"\footnotesize{$2$}" below] [left=of b0] {};
  \node (B1)[draw,circle,fill=green!50,"\footnotesize{$1$}" below] [left=of B0] {};
  \node (B2)[draw,circle,fill=green!50,"\footnotesize{$4N+3$}" below] [left=of B1] {};
  \node (B3) [left=of B2] {};
  \node [draw,circle,fill=red!50,"\footnotesize{$N-1$}" below] (b3) [right=of b2] {};
   \node [draw,circle,fill=red!50,"\footnotesize{$N$}" below] (b4) [right=of b3] {};

\draw[decorate sep={1mm}{5mm},fill] (6.5,-1.25) -- (7.5,-1.25);
\draw[decorate sep={1mm}{4mm},fill] (13.25,-1.25) -- (14.25,-1.25);

  \begin{scope}[>=Latex]

    \draw[-> , thick]  (a1) edge  (a3);
   \draw[-> , thick]  (a3) edge  (a4);
    \draw[-> , thick]  (a4) edge  (a5);
    \draw[-> , thick]  (a6) edge  (a5);

     \draw[-> , thick]  (b1) edge  (b2);
   \draw[-> , thick]  (b2) edge  (b3);
   \draw[-> , thick]  (b3) edge  (b4);

 \draw[-> , thick]  (b2) edge  (a1);
   \draw[-> , thick]  (a4) edge  (b2);

 \draw[-> , thick]  (b3) edge  (a3);
   \draw[-> , thick]  (b3) edge  (c2);
    \draw[-> , thick]  (a5) edge  (b3);

     \draw[-> , thick]  (b4) edge  (dd1);
   \draw[-> , thick]  (b4) edge  (a4);
   \draw[-> , thick]  (b4) edge  (a4);
\draw[-> , thick]  (b4) edge  (a6);

\draw[-> , thick]  (c2) edge  (a4);
\draw[-> , thick]  (c2) edge[bend right=5]  (dd1);
\draw[-> , thick]  (b4) edge  (a4);
\draw[-> , thick]  (a6) edge  (c2);
    
   \draw[-> , thick]  (dd1) edge  (a5); 
   \draw[-> , thick]  (b4) edge  (a4);

\draw[-> , thick]  (dd1) edge  (a6); 

\draw[-> , thick]  (a3) edge  (b1); 

\draw[-> , thick]  (B0) edge  (b0); 
\draw[-> , thick]  (B1) edge  (B0);
\draw[-> , thick]  (B2) edge  (B1); 
\draw[-> , thick]  (B3) edge  (B2);

\draw[-> , thick]  (b0) edge  (A0); 
\draw[-> , thick]  (a0) edge  (B0);
\draw[-> , thick]  (B0) edge  (A1); 
\draw[-> , thick]  (A0) edge  (B1);
\draw[-> , thick]  (B1) edge  (A2); 
\draw[-> , thick]  (A1) edge  (B2);
\draw[-> , thick]  (B2) edge  (A3); 
\draw[-> , thick]  (A2) edge  (B3);

\draw[-> , thick]  (A3) edge  (A2); 
\draw[-> , thick]  (A2) edge  (A1); 
\draw[-> , thick]  (A1) edge  (A0); 
\draw[-> , thick]  (A0) edge  (a0);

    \end{scope}

\end{tikzpicture}
}
\end{center}
 \caption{Quiver, $Q_{A_{2N}}$, corresponding to \eqref{exchmA2N}.The red vertices indicate the nodes that already existed in $Q_{A_{4}}$ while the green vertices represent those newly included through local expansion. }
    \label{fig:QA2N}
\end{figure}
 Applying the mutations $\mu_{3N+6}\mu_{3N+5}\mu_{N+3}$ sequentially yields the following quiver.  
\begin{center}
\resizebox{0.85\textwidth}{!}{%
 \begin{tikzpicture}[every circle node/.style={draw,scale=0.6,thick},node distance=10mm]
  \node [draw,circle,fill=blue!50,"{\footnotesize{$4N+4$}}"] (a1) at (0,0) {};
  \node [draw,circle,fill=red!50,"{\footnotesize{$N+4$}}"] (a2) [right=of a1] {};
  \node [draw,circle,fill=red!50,"{\footnotesize{$N+5$}}"] (a3) [right=of a2] {};
  \node [draw,circle,fill=green!50,"{\footnotesize{$N+6$}}"] (a4) [right=of a3] {};
  \node [draw,circle,fill=green!50,"{\footnotesize{$N+7$}}"] (a5) [right=of a4] {};
  \node [draw,circle,fill=red!50,"{\footnotesize{$N+2$}}" left]  (c1) [below left=of a1] {};
     \node (a6)[right=of a5]{};
  
  \node [draw,circle,fill=red!50,"{\footnotesize{$N+3$}}" below] (b1) at (1.25,-2.5) {} ; 
  \node [draw,circle,fill=red!50,"{\footnotesize{$3N+5$}}" below] (b2) [right=of b1] {};
  \node [draw,circle,fill=green!50,"{\footnotesize{$3N+6$}}" below] (b3) [right=of b2] {};
   \node [draw,circle,fill=green!50,"{\footnotesize{$3N+7$}}" below] (b4) [right=of b3] {};
   \node (b5)[right=of b4]{};
    \node (b6) at (6.5,-3) {};
    \node (b7) at (6.5,-3.2) {};
    \node (b8) at (6.5,-3.4) {};
    \node (b9) at (6.5,-3.6) {};

    \node [draw,circle,fill=red!50,"{\footnotesize{$N+1$}}" below] (d1) at (5,-4) {} ;

\draw [-{Latex[length=3mm]}] (-3.5,-1.2) -- (-2.5,-1.2) node[midway,sloped,above] {$\mu_{N+3}$};

  \begin{scope}[>=Latex]

  \draw[-> , thick]  (a2) edge (a1); 
    \draw[-> , thick]  (a2) edge (a3); 
  \draw[-> , thick]  (a3) edge (a4); 
  \draw[-> , thick]  (a4) edge (a5); 

    \draw[-> , thick]  (a3) edge (c1); 
    \draw[-> , thick]  (c1) edge (b2);
    \draw[-> , thick]  (a1) edge (c1); 

    \draw[-> , thick]  (a1) edge (b1); 
    \draw[-> , thick]  (a2) edge (d1); 
    \draw[-> , thick]  (a3) edge[bend right=35] (a1); 
    
    \draw[- , thick]  (d1) edge (b6);
     \draw[- , thick]  (d1) edge (b7);
     \draw[-> , thick]  (b8) edge (d1);
     \draw[-> , thick]  (b9) edge (d1);

    \draw[-> , thick]  (b2) edge (b1);
    \draw[-> , thick]  (b2) edge (b3); 
    \draw[-> , thick]  (b3) edge (b4); 
    \draw[-> , thick]  (b1) edge (d1); 
    \draw[-> , thick]  (d1) edge (b2); 

    \draw[-> , thick]  (b1) edge (a3);
    \draw[-> , thick]  (b2) edge (a2);

    \draw[-> , thick]  (a3) edge (b2); 
    \draw[-> , thick]  (a4) edge (b2); 
    \draw[-> , thick]  (b3) edge (a3); 
    \draw[-> , thick]  (a5) edge (b3); 
    \draw[-> , thick]  (b4) edge (a4); 

    \draw[-> , thick] (d1) edge[bend left=35] (c1);
    \draw[-> , thick]  (a5) edge (a6); 
     \draw[-> , thick]   (b4) edge (b5) ;
     \draw[-> , thick]   (b5) edge (a5) ;
    \draw[-> , thick]   (a6) edge (b4); 
    \end{scope}
\draw [-{Latex[length=3mm]}] (6.5,-1.2) -- (7.5,-1.2) node[midway,sloped,above] {$\mu_{3N+5}$};

 \node [draw,circle,fill=blue!50,"{\footnotesize{$4N+4$}}"] (a1) at (10,0) {};
  \node [draw,circle,fill=red!50,"{\footnotesize{$N+4$}}"] (a2) [right=of a1] {};
  \node [draw,circle,fill=red!50,"{\footnotesize{$N+5$}}"] (a3) [right=of a2] {};
  \node [draw,circle,fill=green!50,"{\footnotesize{$N+6$}}"] (a4) [right=of a3] {};
  \node [draw,circle,fill=green!50,"{\footnotesize{$N+7$}}"] (a5) [right=of a4] {};
  \node [draw,circle,fill=red!50,"{\footnotesize{$N+2$}}" left]  (c1) [below left=of a1] {};
     \node (a6)[right=of a5]{};
  
  \node [draw,circle,fill=red!50,"{\footnotesize{$N+3$}}" below] (b1) at (11.25,-2.5) {} ; 
  \node [draw,circle,fill=red!50,"{\footnotesize{$3N+5$}}" below] (b2) [right=of b1] {};
  \node [draw,circle,fill=green!50,"{\footnotesize{$3N+6$}}" below] (b3) [right=of b2] {};
   \node [draw,circle,fill=green!50,"{\footnotesize{$3N+7$}}" below] (b4) [right=of b3] {};
   \node (b5)[right=of b4]{};
    \node (b6) at (16.5,-3) {};
    \node (b7) at (16.5,-3.2) {};
    \node (b8) at (16.5,-3.4) {};
    \node (b9) at (16.5,-3.6) {};

    \node [draw,circle,fill=red!50,"{\footnotesize{$N+1$}}" below] (d1) at (15,-4) {} ;

  \begin{scope}[>=Latex]
  \draw[-> , thick] (a2) edge (a1);
  \draw[-> , thick] (a3) edge [bend right=60] (a1);
  \draw[-> , thick]  (a4) edge [bend right=60] (a2); 
  \draw[-> , thick]  (a4) edge (a3); 
  \draw[-> , thick]  (a4) edge (a5); 
   \draw[-> , thick]  (b4) edge (a4); 
   \draw[-> , thick]  (a1) edge  (c1);
    \draw[-> , thick]  (a3) edge (c1); 
    \draw[-> , thick]  (c1) edge (b3); 
    \draw[-> , thick]  (c1) edge (a2); 
    \draw[-> , thick]  (b2) edge (c1);
    \draw[-> , thick]  (c1) edge (b1); 

     \draw[- , thick]  (d1) edge (b6);
     \draw[- , thick]  (d1) edge (b7);
     \draw[-> , thick]  (b8) edge (d1);
     \draw[-> , thick]  (b9) edge (d1);

    \draw[-> , thick]  (b3) edge (b2); 
    \draw[-> , thick]  (b3) edge (b4); 
    \draw[-> , thick]  (b1) edge (b2);
    
    \draw[-> , thick]  (d1) edge (b3); 
     \draw[-> , thick]  (b2) edge  (d1); 
    
    \draw[-> , thick]  (a4) edge (b1); 
     \draw[-> , thick]  (b2) edge (a4); 
    \draw[-> , thick]  (a1) edge (b1);
    \draw[-> , thick]  (a2) edge (b2); 
    \draw[-> , thick]  (b2) edge (a3); 
    \draw[-> , thick]  (a4) edge (b3); 
    \draw[-> , thick]  (a5) edge (b3); 

    \draw[-> , thick] (d1) edge[bend left=35] (c1);
    \draw[-> , thick]  (a5) edge (a6); 
    \draw[-> , thick]  (b4) edge (b5); 
    \draw[-> , thick]  (a6) edge (b4); 
    \draw[-> , thick]  (b5) edge (a5); 
    \draw[-> , thick]  (a5) edge (b2);

    \end{scope}

\end{tikzpicture}
}
\end{center}

\begin{center}
\resizebox{0.5\textwidth}{!}{%
 \begin{tikzpicture}[every circle node/.style={draw,scale=0.6,thick},node distance=10mm]
 
 \node [draw,circle,fill=blue!50,"{\footnotesize{$4N+4$}}"] (a1) at (0,0) {};
  \node [draw,circle,fill=red!50,"{\footnotesize{$N+4$}}"] (a2) [right=of a1] {};
  \node [draw,circle,fill=red!50,"{\footnotesize{$N+5$}}"] (a3) [right=of a2] {};
  \node [draw,circle,fill=green!50,"{\footnotesize{$N+6$}}"] (a4) [right=of a3] {};
  \node [draw,circle,fill=green!50,"{\footnotesize{$N+7$}}"] (a5) [right=of a4] {};
  \node [draw,circle,fill=red!50,"{\footnotesize{$N+2$}}" left]  (c1) [below left=of a1] {};
     \node (a6)[right=of a5]{};
  
  \node [draw,circle,fill=red!50,"{\footnotesize{$N+3$}}" below] (b1) at (1.25,-2.5) {} ; 
  \node [draw,circle,fill=red!50,"{\footnotesize{$3N+5$}}" below] (b2) [right=of b1] {};
  \node [draw,circle,fill=green!50,"{\footnotesize{$3N+6$}}" below] (b3) [right=of b2] {};
   \node [draw,circle,fill=green!50,"{\footnotesize{$3N+7$}}" below] (b4) [right=of b3] {};
   \node (b5)[right=of b4]{};
   \node (b6) at (6.5,-3) {};
    \node (b7) at (6.5,-3.2) {};
    \node (b8) at (6.5,-3.4) {};
    \node (b9) at (6.5,-3.6) {};

    \node [draw,circle,fill=red!50,"{\footnotesize{$N+1$}}" below] (d1) at (5,-4) {} ;

\draw [-{Latex[length=3mm]}] (-3.5,-1.2) -- (-2.5,-1.2) node[midway,sloped,above] {$\mu_{3N+6}$};

  \begin{scope}[>=Latex]
  \draw[-> , thick] (a2) edge (a1);
  \draw[-> , thick] (a3) edge [bend right=60] (a1);
  \draw[-> , thick]  (a4) edge [bend right=60] (a2); 
  \draw[-> , thick]  (a3) edge (a4); 
  \draw[-> , thick]  (a4) edge (a5); 
   \draw[-> , thick]  (a1) edge  (c1);
    \draw[-> , thick]  (a3) edge (c1); 
    \draw[-> , thick]  (b3) edge (c1); 
    \draw[-> , thick]  (c1) edge (a2); 
    \draw[-> , thick]  (c1) edge (b4);
    \draw[-> , thick]  (c1) edge (b1); 
    \draw[-> , thick]  (b1) edge (b2);
    \draw[-> , thick]  (b2) edge (b3); 
    \draw[-> , thick]  (b4) edge (b3); 
    \draw[-> , thick]  (b1) edge (b2);

    \draw[- , thick]  (d1) edge (b6);
     \draw[- , thick]  (d1) edge (b7);
     \draw[-> , thick]  (b8) edge (d1);
     \draw[-> , thick]  (b9) edge (d1);
    
    \draw[-> , thick]  (b3) edge (d1); 
     \draw[-> , thick]  (d1) edge (b4); 
    
    \draw[-> , thick]  (a4) edge (b1); 
    \draw[-> , thick]  (a1) edge (b1);
    \draw[-> , thick]  (a2) edge (b2); 
    \draw[-> , thick]  (b2) edge (a3); 
    \draw[-> , thick]  (b3) edge (a4); 
    \draw[-> , thick]  (b3) edge (a5); 
   \draw[-> , thick]  (a5) edge (b4); 

    \draw[-> , thick] (d1) edge[bend left=35] (c1);
    \draw[-> , thick]  (a5) edge (a6); 
    \draw[-> , thick]  (b4) edge (b5); 
    \draw[-> , thick]  (a6) edge (b4); 
    \draw[-> , thick]  (b5) edge (a5); 
    \draw[-> , thick]  (a5) edge (b2);

    \end{scope}

\end{tikzpicture}
}
\end{center}
Let us consider on the subquiver of the mutated quiver $\mu_{3N+6}\mu_{3N+5}\mu_{N+3}(Q_{A_{2N}})$ (1) as shown below. By applying the mutation at the node $3N+6$, we get the following quiver (2).  
\begin{center}
\resizebox{1\textwidth}{!}{%
 \begin{tikzpicture}[every circle node/.style={draw,scale=0.6,thick},node distance=10mm]

\node at (3.7,-5.5) {(1)};
\node at (13.7,-5.5) {(2)};

  \node (a1) at (0,0) {};
  
  \node [draw,circle,fill=green!50,"{\footnotesize{$N+6$}}"] (a4) [right=of a3] {};
  \node [draw,circle,fill=green!50,"{\footnotesize{$N+7$}}"] (a5) [right=of a4] {};
  \node [draw,circle,fill=red!50,"{\footnotesize{$N+2$}}" left]  (c1) [below left=of a1] {};
     \node [draw,circle,fill=green!50,"{\footnotesize{$N+8$}}"](a6)[right=of a5]{};
  
  \node [draw,circle,fill=red!50,"{\footnotesize{$N+3$}}" below] (b1) at (1.25,-2.5) {} ; 
  \node [draw,circle,fill=red!50,"{\footnotesize{$3N+5$}}" below] (b2) [right=of b1] {};
  \node [draw,circle,fill=green!50,"{\footnotesize{$3N+6$}}" below] (b3) [right=of b2] {};
   \node [draw,circle,fill=green!50,"{\footnotesize{$3N+7$}}" below] (b4) [right=of b3] {};
   \node[draw,circle,fill=green!50,"{\footnotesize{$3N+8$}}" below] (b5)[right=of b4]{};

    \node [draw,circle,fill=red!50,"{\footnotesize{$N+1$}}" below] (d1) at (5,-4) {} ;

  \begin{scope}[>=Latex]

  \draw[-> , thick]  (a4) edge (a5);

    \draw[-> , thick]  (b2) edge (c1); 

    \draw[-> , thick]  (c1) edge (b3);

    \draw[-> , thick]  (b1) edge (b2);
    \draw[-> , thick]  (b3) edge (b2); 
    \draw[-> , thick]  (b3) edge (b4); 
    \draw[-> , thick]  (b1) edge (b2); 
    \draw[-> , thick]  (b2) edge (d1); 
    \draw[-> , thick]  (d1) edge (b3); 
    \draw[-> , thick]  (a4) edge (b1);

    \draw[-> , thick]  (b2) edge (a4); 
    \draw[-> , thick]  (a4) edge (b3); 
    \draw[-> , thick]  (a5) edge (b3); 
    \draw[-> , thick]  (b4) edge (a4);

    \draw[-> , thick]  (a5) edge (a6); 
     \draw[-> , thick]   (b4) edge (b5) ;
     \draw[-> , thick]   (b5) edge (a5) ;
    \draw[-> , thick]   (a6) edge (b4); 
    \end{scope}
    
\draw [-{Latex[length=3mm]}] (6.5,-1.2) -- (7.5,-1.2) node[midway,sloped,above] {$\mu_{3N+6}$};

\node (a1) at (10,0) {};
  \node (a2) [right=of a1] {};
  \node (a3) [right=of a2] {};
  \node [draw,circle,fill=green!50,"{\footnotesize{$N+6$}}"] (a4) [right=of a3] {};
  \node [draw,circle,fill=green!50,"{\footnotesize{$N+7$}}"] (a5) [right=of a4] {};
  \node [draw,circle,fill=red!50,"{\footnotesize{$N+2$}}" left]  (c1) [below left=of a1] {};
     \node [draw,circle,fill=green!50,"{\footnotesize{$N+8$}}"](a6)[right=of a5]{};
  
  \node [draw,circle,fill=red!50,"{\footnotesize{$N+3$}}" below] (b1) at (11.25,-2.5) {} ; 
  \node [draw,circle,fill=red!50,"{\footnotesize{$3N+5$}}" below] (b2) [right=of b1] {};
  \node [draw,circle,fill=green!50,"{\footnotesize{$3N+6$}}" below] (b3) [right=of b2] {};
   \node [draw,circle,fill=green!50,"{\footnotesize{$3N+7$}}" below] (b4) [right=of b3] {};
   \node[draw,circle,fill=green!50,"{\footnotesize{$3N+8$}}" below] (b5)[right=of b4]{};

    \node [draw,circle,fill=red!50,"{\footnotesize{$N+1$}}" below] (d1) at (15,-4) {} ;

  \begin{scope}[>=Latex]

  \draw[-> , thick]  (a4) edge (a5);

    \draw[-> , thick]  (b3) edge (c1); 
    
    \draw[-> , thick]  (c1) edge (b4);
   
    \draw[-> , thick]  (b1) edge (b2);
    \draw[-> , thick]  (b2) edge (b3); 
    \draw[-> , thick]  (b4) edge (b3); 
    \draw[-> , thick]  (b1) edge (b2);
    
    \draw[-> , thick]  (b3) edge (d1); 
     \draw[-> , thick]  (d1) edge  (b4); 
    
    \draw[-> , thick]  (a4) edge (b1);

    \draw[-> , thick]  (b3) edge (a4); 
    \draw[-> , thick]  (b3) edge (a5); 
   \draw[-> , thick]  (a5) edge (b4);

    \draw[-> , thick]  (a5) edge (a6); 
    \draw[-> , thick]  (b4) edge (b5); 
    \draw[-> , thick]  (a6) edge (b4); 
    \draw[-> , thick]  (b5) edge (a5); 
    \draw[-> , thick]  (a5) edge (b2);

    \end{scope}

\end{tikzpicture}
}
\end{center}
As we can see, the nodes (N+1, N+2, N+7, 3N+5, 3N+6, 3N+7,3N+8) in the quiver (2) are connected in the same way as the one with the nodes (N+1, N+2, N+3, N+6, 3N+5, 3N+6, 3N+7) in the quiver (1). Consequently, applying the quiver mutation at node 3N+7 results in the same transformation. Since the connection patterns between the nodes on the right-hand side of the mutated quiver remain consistent across certain portions, the transformation persists throughout the mutations. By sequentially applying the mutations $\mu_{N-2} \cdots \mu_{2}\mu_{1}\mu_{4N+3}\dots \mu_{3N+8}\mu_{3N+7}$, we see the following quivers. 
\begin{center}
\resizebox{1\textwidth}{!}{%
 \begin{tikzpicture}[every circle node/.style={draw,scale=0.6,thick},node distance=10mm]
 
  \draw [-{Latex[length=3mm]}] (-2.5,-1.2) -- (-1.5,-1.2) node[midway,sloped,above] {$\mu_{3N+6}$};
  
  \node [draw,circle,fill=green!50,"{\footnotesize{$N+7$}}"] (a3) at (1.265,0) {};
  \node [draw,circle,fill=green!50,"{\footnotesize{$N+8$}}"] (a4) [right=of a3] {};
  \node [draw,circle,fill=green!50,"{\footnotesize{$N+9$}}"] (a5) [right=of a4] {};
  \node [draw,circle,fill=green!50,"{\footnotesize{$N+10$}}"] (a6) [right=of a5] {};
  \node [draw,circle,fill=red!50,"{\footnotesize{$N+2$}}"]  (c1) at (-1,-0.5) {};
  \node (a1) [left=of a3]{};
  \node (a7) [right=of a6]{};
  
  \node [draw,circle,fill=green!50,"{\footnotesize{$3N+6$}}" below] (b2) at (1.265,-2.5) {};
  \node [draw,circle,fill=green!50,"{\footnotesize{$3N+7$}}" below] (b3) [right=of b2] {};
   \node [draw,circle,fill=green!50,"{\footnotesize{$3N+8$}}" below] (b4) [right=of b3] {};
\node [draw,circle,fill=green!50,"{\footnotesize{$3N+9$}}" below] (b5) [right=of b4] {};
\node (b1) [left=of b2]{};
\node (b6) [right=of b5]{};

    \node [draw,circle,fill=red!50,"{\footnotesize{$N+1$}}" below] (d1) at (5,-4) {} ;

  \begin{scope}[>=Latex]

    \draw[-> ,thick]  (b2) edge (c1);  
    \draw[-> , thick]  (c1) edge (b3); 
    \draw[-> , thick]  (a3) edge  (a4);
    \draw[-> , thick]  (a4) edge  (a5);
     \draw[-> , thick]  (a5) edge  (a6);

       \draw[-> , thick]  (b5) edge (b6); 
    \draw[-> , thick]  (b6) edge (a6); 
    \draw[-> , thick]  (a7) edge (b5); 
    \draw[-> , thick]  (a6) edge (a7); 

      \draw[-> , thick]  (a1) edge  (a3);
       \draw[-> , thick]  (b1) edge  (b2);
       \draw[- , thick]  (a3) edge  (0.2,-0.8); 
       \draw[-> , thick]  (a4) edge  (b1);
    
    \draw[-> , thick]  (b3) edge (b2); 
    \draw[-> , thick]  (b3) edge (b4); 
    \draw[-> , thick]  (b4) edge (b5); 
    
    \draw[-> , thick]  (b2) edge (d1); 
    \draw[-> , thick]  (d1) edge (b3);

    \draw[-> , thick]  (b2) edge (a3); 
    \draw[-> , thick]  (b2) edge (a4); 
    \draw[-> , thick]  (a4) edge (b3); 
    \draw[-> , thick]  (a5) edge (b3); 
    \draw[-> , thick]  (b4) edge (a4); 

    \draw[-> , thick] (d1) edge[bend left=35] (c1);
    \draw[-> , thick]  (b5) edge (a5); 
    \draw[-> , thick]  (a6) edge (b4); 
    \end{scope}
\draw [-{Latex[length=3mm]}] (6.5,-1.2) -- (7.5,-1.2) node[midway,sloped,above] {$\mu_{3N+7}$};
\node [draw,circle,fill=green!50,"{\footnotesize{$N+7$}}"] (a3) at (10.265,0) {};
  \node [draw,circle,fill=green!50,"{\footnotesize{$N+8$}}"] (a4) [right=of a3] {};
  \node [draw,circle,fill=green!50,"{\footnotesize{$N+9$}}"] (a5) [right=of a4] {};
  \node [draw,circle,fill=green!50,"{\footnotesize{$N+10$}}"] (a6) [right=of a5] {};
  \node [draw,circle,fill=red!50,"{\footnotesize{$N+2$}}"]  (c1) at (8,-0.5) {};
  \node (a1) [left=of a3]{};
  \node (a7) [right=of a6]{};
  
  \node [draw,circle,fill=green!50,"{\footnotesize{$3N+6$}}" below] (b2) at (10.265,-2.5) {};
  \node [draw,circle,fill=green!50,"{\footnotesize{$3N+7$}}" below] (b3) [right=of b2] {};
   \node [draw,circle,fill=green!50,"{\footnotesize{$3N+8$}}" below] (b4) [right=of b3] {};
\node [draw,circle,fill=green!50,"{\footnotesize{$3N+9$}}" below] (b5) [right=of b4] {};
\node (b1) [left=of b2]{};
\node (b6) [right=of b5]{};

    \node [draw,circle,fill=red!50,"{\footnotesize{$N+1$}}" below] (d1) at (14,-4) {} ;

  \begin{scope}[>=Latex]

    \draw[-> , thick]  (b3) edge (c1);  
    \draw[-> , thick]  (c1) edge (b4); 
    \draw[-> , thick]  (a3) edge  (a4);
    \draw[-> , thick]  (a4) edge  (a5);
     \draw[-> , thick]  (a5) edge  (a6);

        \draw[-> , thick]  (a1) edge  (a3);
       \draw[-> , thick]  (b1) edge  (b2);
       \draw[- , thick]  (a3) edge  (9.2,-0.8); 
       \draw[-> , thick]  (a4) edge  (b1);

       \draw[-> , thick]  (b5) edge (b6); 
    \draw[-> , thick]  (b6) edge (a6); 
    \draw[-> , thick]  (a7) edge (b5); 
    \draw[-> , thick]  (a6) edge (a7); 
    
    \draw[-> , thick]  (b2) edge (b3); 
    \draw[-> , thick]  (b4) edge (b3); 
    \draw[-> , thick]  (b4) edge (b5); 
    
    \draw[-> , thick]  (b3) edge (d1); 
    \draw[-> , thick]  (d1) edge (b4);

    \draw[-> , thick]  (b2) edge (a3); 
    
    \draw[-> , thick]  (b3) edge (a4); 
    \draw[-> , thick]  (b3) edge (a5); 
    \draw[-> , thick]  (a5) edge (b4); 
        \draw[-> , thick]  (a5) edge (b2); 

    \draw[-> , thick] (d1) edge[bend left=35] (c1);
    \draw[-> , thick]  (b5) edge (a5); 
    \draw[-> , thick]  (a6) edge (b4); 
    \end{scope}


 \draw [-{Latex[length=3mm]}] (-2.5,-7.2) -- (-1.5,-7.2) node[midway,sloped,above] {$\mu_{3N+8}$};

 \draw[decorate sep={1mm}{5mm},fill] (2,-7.2) -- (4,-7.2);

 \draw [-{Latex[length=3mm]}] (6.5,-7.2) -- (7.5,-7.2) node[midway,sloped,above] {$\mu_{N-2}$};

  \node [draw,circle,fill=green!50,"{\footnotesize{$3N+1$}}"] (a3) at (10.265,-6) {};
  \node [draw,circle,fill=red!50,"{\footnotesize{$3N+2$}}"] (a4) [right=of a3] {};
  \node [draw,circle,fill=red!50,"{\footnotesize{$3N+3$}}"] (a5) [right=of a4] {};
  \node [draw,circle,fill=blue!50,"{\footnotesize{$4N+5$}}"] (a6) [right=of a5] {};
  \node [draw,circle,fill=red!50,"{\footnotesize{$3N+4$}}" right] (c2) [below right=of a6] {};
  \node [draw,circle,fill=red!50,"$N+2$"]  (c1) at (8,-6.5) {};
  \node (a1) [left=of a3] {};

  \node [draw,circle,fill=green!50,"{\footnotesize{$N-2$}}" below] (b2) at (10.265,-8.5) {};
  \node (b1) [left=of b2] {};
  \node [draw,circle,fill=red!50,"{\footnotesize{$N-1$}}" below] (b3) [right=of b2] {};
   \node [draw,circle,fill=red!50,"{\footnotesize{$N$}}" below] (b4) [right=of b3] {};

    \node [draw,circle,fill=red!50,"$N+1$" below] (d1) at (9.6,-10) {} ;

  \begin{scope}[>=Latex]

    \draw[-> , thick]  (b2) edge (c1);  
    \draw[-> , thick]  (c1) edge (b3); 
    \draw[-> , thick]  (a3) edge  (a4);
    \draw[-> , thick]  (a4) edge  (a5);
     \draw[-> , thick]  (a6) edge  (a5);

 \draw[-> , thick]  (a1) edge  (a3);
       \draw[-> , thick]  (b1) edge  (b2);
       \draw[- , thick]  (a3) edge  (9.2,-6.8); 
       \draw[-> , thick]  (a4) edge  (b1);

     \draw[-> , thick]  (a6) edge  (c2);
    \draw[-> , thick]  (b3) edge (c2); 
    \draw[-> , thick]  (c2) edge (a4); 
    \draw[-> , thick]  (c2) edge[bend left=35] (d1); 
    
    \draw[-> , thick]  (b4) edge (d1); 
    \draw[-> , thick]  (d1) edge (a5); 
    \draw[-> , thick]  (d1) edge (a6); 

    \draw[-> , thick]  (b4) edge (a6); 
    
    \draw[-> , thick]  (b3) edge (b2); 
    \draw[-> , thick]  (b3) edge (b4);

    \draw[-> , thick]  (b2) edge (d1); 
    \draw[-> , thick]  (d1) edge (b3);

    \draw[-> , thick]  (b2) edge (a3); 
    \draw[-> , thick]  (b2) edge (a4); 
    \draw[-> , thick]  (a4) edge (b3); 
    \draw[-> , thick]  (a5) edge (b3); 
    \draw[-> , thick]  (b4) edge (a4); 

    \draw[-> , thick] (d1) edge[bend left=35] (c1);

    \end{scope}


 \draw [-{Latex[length=3mm]}] (-3.5,-13.2) -- (-2.5,-13.2) node[midway,sloped,above] {$\mu_{N-1}$};

 \node [draw,circle,fill=green!50,"{\footnotesize{$3N+1$}}"] (a3) at (1.265,-12) {};
  \node [draw,circle,fill=red!50,"{\footnotesize{$3N+2$}}"] (a4) [right=of a3] {};
  \node [draw,circle,fill=red!50,"{\footnotesize{$3N+3$}}"] (a5) [right=of a4] {};
  \node [draw,circle,fill=blue!50,"{\footnotesize{$4N+5$}}"] (a6) [right=of a5] {};
  \node [draw,circle,fill=red!50,"{\footnotesize{$3N+4$}}" right] (c2) [below right=of a6] {};
  \node [draw,circle,fill=red!50,"$N+2$"]  (c1) at (-1,-12.5) {};
  \node (a1) [left=of a3] {};

  \node [draw,circle,fill=green!50,"{\footnotesize{$N-2$}}" below] (b2) at (1.265,-14.5) {};
  \node (b1) [left=of b2] {};
  \node [draw,circle,fill=red!50,"{\footnotesize{$N-1$}}" below] (b3) [right=of b2] {};
   \node [draw,circle,fill=red!50,"{\footnotesize{$N$}}" below] (b4) [right=of b3] {};

    \node [draw,circle,fill=red!50,"$N+1$" below] (d1) at (0.6,-16) {} ;

  \begin{scope}[>=Latex]

    \draw[-> , thick]  (b3) edge (c1); 
    \draw[-> , thick]  (c1) edge (b4); 
    \draw[-> , thick]  (a3) edge  (a4);
    \draw[-> , thick]  (a4) edge  (a5);
     \draw[-> , thick]  (a6) edge  (a5);

 \draw[-> , thick]  (a1) edge  (a3);
       \draw[-> , thick]  (b1) edge  (b2);
       \draw[- , thick]  (a3) edge  (0.2,-12.8); 
       \draw[-> , thick]  (a4) edge  (b1);

     \draw[-> , thick]  (a6) edge  (c2);
    \draw[-> , thick]  (c2) edge (b3); 
    \draw[-> , thick]  (a5) edge (c2); 
    \draw[-> , thick]  (c1) edge (c2);

    \draw[-> , thick]  (d1) edge (a5); 
    \draw[-> , thick]  (d1) edge (a6); 

    \draw[-> , thick]  (b4) edge (a6); 
    
    \draw[-> , thick]  (b2) edge (b3); 
    \draw[-> , thick]  (b4) edge (b3);

    \draw[-> , thick]  (b3) edge (d1);

    \draw[-> , thick]  (b2) edge (a3); 

    \draw[-> , thick]  (b3) edge (a4); 
    \draw[-> , thick]  (b3) edge (a5);  
    \draw[-> , thick]  (a5) edge (b2);  
    \draw[-> , thick]  (a5) edge (b4);

    \draw[-> , thick] (d1) edge[bend left=35] (c1);

    \end{scope}
\draw [-{Latex[length=3mm]}] (7.5,-13.2) -- (8.5,-13.2) node[midway,sloped,above] {$\mu_{N}$};

\node [draw,circle,fill=green!50,"{\footnotesize{$3N+1$}}"] (a3) at (12.265,-12) {};
  \node [draw,circle,fill=red!50,"{\footnotesize{$3N+2$}}"] (a4) [right=of a3] {};
  \node [draw,circle,fill=red!50,"{\footnotesize{$3N+3$}}"] (a5) [right=of a4] {};
  \node [draw,circle,fill=blue!50,"{\footnotesize{$4N+5$}}"] (a6) [right=of a5] {};
  \node [draw,circle,fill=red!50,"{\footnotesize{$3N+4$}}" right] (c2) [below right=of a6] {};
  \node [draw,circle,fill=red!50,"$N+2$"]  (c1) at (10,-12.5) {};
  \node (a1) [left=of a3] {};

  \node [draw,circle,fill=green!50,"{\footnotesize{$N-2$}}" below] (b2) at (12.265,-14.5) {};
  \node (b1) [left=of b2] {};
  \node [draw,circle,fill=red!50,"{\footnotesize{$N-1$}}" below] (b3) [right=of b2] {};
   \node [draw,circle,fill=red!50,"{\footnotesize{$N$}}" below] (b4) [right=of b3] {};

    \node [draw,circle,fill=red!50,"$N+1$" below] (d1) at (11.6,-16) {} ;

 \begin{scope}[>=Latex]

    \draw[-> , thick]  (b4) edge (c1);  
    \draw[-> , thick]  (c1) edge (a6); 
    \draw[-> , thick]  (c1) edge (c2); 
    
    \draw[-> , thick]  (a3) edge  (a4);
    \draw[-> , thick]  (a4) edge  (a5);
     
    \draw[-> , thick]  (a5) edge  (c2);
     \draw[-> , thick]  (a6) edge  (c2);
    \draw[-> , thick]  (c2) edge (b3); 
    
 \draw[-> , thick]  (a1) edge  (a3);
       \draw[-> , thick]  (b1) edge  (b2);
       \draw[- , thick]  (a3) edge  (11.2,-12.8); 
       \draw[-> , thick]  (a4) edge  (b1);
    
    \draw[-> , thick]  (d1) edge (a5); 
    \draw[-> , thick]  (d1) edge (a6); 

    \draw[-> , thick]  (a6) edge (b4); 
    
    \draw[-> , thick]  (b2) edge (b3); 
    \draw[-> , thick]  (b3) edge (b4);

    \draw[-> , thick]  (b3) edge (d1); 
   
    \draw[-> , thick]  (b2) edge (a3); 
    
    \draw[-> , thick]  (b3) edge (a4); 
    \draw[-> , thick]  (a5) edge (b2); 
     \draw[-> , thick]  (b4) edge (a5);

    \draw[-> , thick] (d1) edge[bend left=35] (c1);
    
    \end{scope}

\end{tikzpicture}
}
\end{center}
As a result of mutation $\mu_{N}\tilde{\mu} = \mu_{N-1}\cdots\mu_{1}\mu_{4N+3}\cdots \mu_{3N+6}\mu_{3N+5} \mu_{N+3}$, we find 
\begin{center}
\resizebox{1\textwidth}{!}{%
 \begin{tikzpicture}[every circle node/.style={draw,scale=0.6,thick},node distance=10mm]

\node [draw,circle,fill=blue!50,"\footnotesize{$4N+4$}"] (aa1) at (0,0) {};
  \node [draw,circle,fill=red!50,"\footnotesize{$N+4$}"] (aa2) [right=of aa1] {};
  \node [draw,circle,fill=green!50,"\footnotesize{$N+5$}"] (aa3) [right=of aa2] {};
  \node [draw,circle,fill=green!50,"\footnotesize{$N+6$}"] (aa4) [right=of aa3] {};
  \node [draw,circle,fill=green!50,"\footnotesize{$N+7$}"] (aa5) [right=of aa4] {};
  \node [draw,circle,fill=red!50,"\footnotesize{$N+2$}" left]  (cc1) [below left=of aa1] {};
  \node  (aa6) [right=of aa5] {};

  \node [draw,circle,fill=red!50,"\footnotesize{$N+3$}" below] (bb1) at (1.25,-2.5) {} ; 
  \node [draw,circle,fill=red!50,"\footnotesize{$3N+5$}" below] (bb2) [right=of bb1] {};
  \node [draw,circle,fill=green!50,"\footnotesize{$3N+6$}" below] (bb3) [right=of bb2] {};
   \node [draw,circle,fill=green!50,"\footnotesize{$3N+7$}" below] (bb4) [right=of bb3] {};
   \node (bb5) [right=of bb4] {};

    \node [draw,circle,fill=red!50,"\footnotesize{$N+1$}" below] (dd1) at (10,-4) {} ;

  \begin{scope}[>=Latex]
  \draw[-> , thick] (aa2) edge (aa1);
  \draw[-> , thick] (aa3) edge [bend right=60] (aa1);
  \draw[-> , thick]  (aa4) edge [bend right=60] (aa2); 
  \draw[-> , thick]  (aa3) edge (aa4); 
  \draw[-> , thick]  (aa4) edge (aa5); 
   \draw[-> , thick]  (aa1) edge  (cc1);
    \draw[-> , thick]  (aa3) edge (cc1); 

    \draw[-> , thick]  (cc1) edge (aa2); 
    
    \draw[-> , thick]  (cc1) edge (bb1); 
    \draw[-> , thick]  (bb1) edge (bb2);
    \draw[-> , thick]  (bb2) edge (bb3); 
    \draw[-> , thick]  (bb4) edge (bb3); 
    \draw[-> , thick]  (bb1) edge (bb2);
    
    \draw[-> , thick]  (aa6) edge (bb3); 
     \draw[- , thick]  (bb4) edge (6,-1.5);

    \draw[-> , thick]  (aa4) edge (bb1); 
    \draw[-> , thick]  (aa1) edge (bb1);
    \draw[-> , thick]  (aa2) edge (bb2); 
    \draw[-> , thick]  (bb2) edge (aa3); 
    \draw[-> , thick]  (bb3) edge (aa4);  
   \draw[-> , thick]  (bb4) edge (aa5); 
    \draw[-> , thick]  (aa5) edge (bb2); 

    \draw[-> , thick] (dd1) edge (cc1);
    \draw[-> , thick]  (aa5) edge (aa6); 
    \draw[-> , thick]  (bb4) edge (bb5);

\end{scope}

 \node [draw,circle,fill=green!50,"\footnotesize{$3N+1$}"] (a3) at (15.5,0) {};
  \node [draw,circle,fill=red!50,"\footnotesize{$3N+2$}"] (a4) [right=of a3] {};
  \node [draw,circle,fill=red!50,"\footnotesize{$3N+3$}"] (a5) [right=of a4] {};
  \node [draw,circle,fill=blue!50,"\footnotesize{$4N+5$}"] (a6) [right=of a5] {};
  \node [draw,circle,fill=red!50,"\footnotesize{$3N+4$}" right] (c2) [below right=of a6] {};
  \node (a1) [left=of a3] {};
   \node (a0) [left=of a1] {};
   \node (A0) [draw,circle,fill=green!50,"\footnotesize{$2N+5$}"] [left=of a0] {};
   \node (A1) [draw,circle,fill=green!50,"\footnotesize{$2N+4$}"][left=of A0] {};
   \node (A2) [draw,circle,fill=green!50,"\footnotesize{$2N+3$}"][left=of A1] {};
   \node (A3)[left=of A2] {};

  \node [draw,circle,fill=green!50,"\footnotesize{$N-2$}" below](b2) at (15.5,-2.5) {};
  \node (b1) [left=of b2] {};
  \node (b0) [left=of b1] {};
  \node (B0)[draw,circle,fill=green!50,"\footnotesize{$2$}" below] [left=of b0] {};
  \node (B1)[draw,circle,fill=green!50,"\footnotesize{$1$}" below] [left=of B0] {};
  \node (B2)[draw,circle,fill=green!50,"\footnotesize{$4N+3$}" below] [left=of B1] {};
  \node (B3) [left=of B2] {};
  \node [draw,circle,fill=red!50,"\footnotesize{$N-1$}" below] (b3) [right=of b2] {};
   \node [draw,circle,fill=red!50,"\footnotesize{$N$}" below] (b4) [right=of b3] {};

\draw[decorate sep={1mm}{5mm},fill] (6.5,-1.25) -- (7.5,-1.25);
\draw[decorate sep={1mm}{4mm},fill] (13.25,-1.25) -- (14.25,-1.25);

  \begin{scope}[>=Latex]

    \draw[-> , thick]  (b4) edge (cc1);  
    \draw[-> , thick]  (cc1) edge (a6); 
    \draw[-> , thick]  (cc1) edge (c2); 
    
    \draw[-> , thick]  (a3) edge  (a4);
    \draw[-> , thick]  (a4) edge  (a5);
     
    \draw[-> , thick]  (a5) edge  (c2);
     \draw[-> , thick]  (a6) edge  (c2);
    \draw[-> , thick]  (c2) edge (b3); 
    
 \draw[-> , thick]  (a1) edge  (a3);
       \draw[-> , thick]  (b1) edge  (b2);
       \draw[- , thick]  (a3) edge  (8.2,-0.8); 
       \draw[-> , thick]  (a4) edge  (b1);
    
    \draw[-> , thick]  (dd1) edge (a5); 
    \draw[-> , thick]  (dd1) edge (a6); 

    \draw[-> , thick]  (a6) edge (b4); 
    
    \draw[-> , thick]  (b2) edge (b3); 
    \draw[-> , thick]  (b3) edge (b4);

    \draw[-> , thick]  (b3) edge (dd1); 
   
    \draw[-> , thick]  (b2) edge (a3); 
    
    \draw[-> , thick]  (b3) edge (a4); 
    \draw[-> , thick]  (a5) edge (b2); 
     \draw[-> , thick]  (b4) edge (a5);
     
\draw[-> , thick]  (B0) edge  (b0); 
\draw[-> , thick]  (B1) edge  (B0);
\draw[-> , thick]  (B2) edge  (B1); 
\draw[-> , thick]  (B3) edge  (B2);

\draw[-> , thick]  (b0) edge  (A0); 
\draw[-> , thick]  (a0) edge  (B0);
\draw[-> , thick]  (B0) edge  (A1); 
\draw[-> , thick]  (A0) edge  (B1);
\draw[-> , thick]  (B1) edge  (A2); 
\draw[-> , thick]  (A1) edge  (B2);
\draw[-> , thick]  (B2) edge  (A3); 
\draw[-> , thick]  (A2) edge  (B3);

\draw[-> , thick]  (A3) edge  (A2); 
\draw[-> , thick]  (A2) edge  (A1); 
\draw[-> , thick]  (A1) edge  (A0); 
\draw[-> , thick]  (A0) edge  (a0);

    \end{scope}

\end{tikzpicture}
}
\end{center}

Furthermore if we apply the cyclic permutation $\rho = (N,N+1,\dots,3N+4) \dots $ on labelling and shift the nodes, then we attain the following resulting quiver 

\begin{center}
\resizebox{1\textwidth}{!}{%
 \begin{tikzpicture}[every circle node/.style={draw,scale=0.6,thick},node distance=10mm]

\node [draw,circle,fill=blue!50,"\footnotesize{$4N+4$}"] (aa1) at (0,0) {};
  \node [draw,circle,fill=red!50,"\footnotesize{$N+5$}"] (aa2) [right=of aa1] {};
  \node [draw,circle,fill=green!50,"\footnotesize{$N+6$}"] (aa3) [right=of aa2] {};
  \node [draw,circle,fill=green!50,"\footnotesize{$N+7$}"] (aa4) [right=of aa3] {};
  \node [draw,circle,fill=green!50,"\footnotesize{$N+8$}"] (aa5) [right=of aa4] {};
  \node [draw,circle,fill=red!50,"\footnotesize{$N+3$}" left]  (cc1) [below left=of aa1] {};
  \node  (aa6) [right=of aa5] {};

  \node [draw,circle,fill=red!50,"\footnotesize{$N+4$}" below] (bb1) at (1.25,-2.5) {} ; 
  \node [draw,circle,fill=red!50,"\footnotesize{$3N+5$}" below] (bb2) [right=of bb1] {};
  \node [draw,circle,fill=green!50,"\footnotesize{$3N+6$}" below] (bb3) [right=of bb2] {};
   \node [draw,circle,fill=green!50,"\footnotesize{$3N+7$}" below] (bb4) [right=of bb3] {};
   \node (bb5) [right=of bb4] {};

    \node [draw,circle,fill=red!50,"\footnotesize{$N+2$}" below] (dd1) at (10,-4) {} ;
    
  \begin{scope}[>=Latex]
  
\draw[-> , thick]  (aa1) edge (cc1);  
\draw[-> , thick]  (aa1) edge (dd1);  
\draw[-> , thick]  (aa2) edge (aa1);  
\draw[-> , thick]  (bb1) edge (aa1);  

\draw[-> , thick]  (cc1) edge (bb2);  
\draw[-> , thick]  (dd1) edge (cc1);  
\draw[-> , thick]  (aa1) edge (cc1);  
\draw[-> , thick]  (aa3) edge (cc1);  

\draw[-> , thick]  (bb1) edge (aa1);  
\draw[-> , thick]  (bb1) edge (bb2);  
\draw[-> , thick]  (dd1) edge (bb1);  
\draw[-> , thick]  (aa3) edge (bb1);  

\draw[-> , thick]  (aa2) edge (aa1);  
\draw[-> , thick]  (aa2) edge (dd1);  
\draw[-> , thick]  (aa2) edge (aa3);  
\draw[-> , thick]  (bb2) edge (aa2); 

\draw[-> , thick]  (aa3) edge (aa4);  
\draw[-> , thick]  (bb3) edge (aa3);  

\draw[-> , thick]  (aa4) edge (bb2);  
\draw[-> , thick]  (aa4) edge (aa5);  
\draw[-> , thick]  (bb4) edge (aa4); 

\draw[-> , thick]  (aa5) edge (bb3);
\draw[-> , thick]  (aa5) edge (aa6);

\draw[-> , thick]  (bb2) edge (bb3); 
\draw[-> , thick]  (bb3) edge (bb4);
\draw[-> , thick]  (bb4) edge (bb5);

\draw[-> , thick]  (aa6) edge (bb4);
\draw[-> , thick]  (bb5) edge (aa5);

\end{scope}

 \node [draw,circle,fill=red!50,"\footnotesize{$3N+2$}"] (a3) at (15.5,0) {};
  \node [draw,circle,fill=red!50,"\footnotesize{$3N+3$}"] (a4) [right=of a3] {};
  \node [draw,circle,fill=red!50,"\footnotesize{$3N+4$}"] (a5) [right=of a4] {};
  \node [draw,circle,fill=blue!50,"\footnotesize{$4N+5$}"] (a6) [right=of a5] {};
  \node [draw,circle,fill=red!50,"\footnotesize{$N$}" right] (c2) [below right=of a6] {};
  \node (a1) [left=of a3] {};
   \node (a0) [left=of a1] {};
   \node (A0) [draw,circle,fill=green!50,"\footnotesize{$2N+5$}"] [left=of a0] {};
   \node (A1) [draw,circle,fill=green!50,"\footnotesize{$2N+4$}"][left=of A0] {};
   \node (A2) [draw,circle,fill=green!50,"\footnotesize{$2N+3$}"][left=of A1] {};
   \node (A3)[left=of A2] {};

  \node [draw,circle,fill=green!50,"\footnotesize{$N-2$}" below](b2) at (15.5,-2.5) {};
  \node (b1) [left=of b2] {};
  \node (b0) [left=of b1] {};
  \node (B0)[draw,circle,fill=green!50,"\footnotesize{$2$}" below] [left=of b0] {};
  \node (B1)[draw,circle,fill=green!50,"\footnotesize{$1$}" below] [left=of B0] {};
  \node (B2)[draw,circle,fill=green!50,"\footnotesize{$4N+3$}" below] [left=of B1] {};
  \node (B3) [left=of B2] {};
  \node [draw,circle,fill=red!50,"\footnotesize{$N-1$}" below] (b3) [right=of b2] {};
   \node [draw,circle,fill=red!50,"\footnotesize{$N+1$}" below] (b4) [right=of b3] {};

\draw[decorate sep={1mm}{5mm},fill] (6.5,-1.25) -- (7.5,-1.25);
\draw[decorate sep={1mm}{4mm},fill] (13.25,-1.25) -- (14.25,-1.25);

  \begin{scope}[>=Latex]

    \draw[-> , thick]  (a1) edge  (a3);
   \draw[-> , thick]  (a3) edge  (a4);
    \draw[-> , thick]  (a4) edge  (a5);
    \draw[-> , thick]  (a6) edge  (a5);

     \draw[-> , thick]  (b1) edge  (b2);
   \draw[-> , thick]  (b2) edge  (b3);
   \draw[-> , thick]  (b3) edge  (b4);

 \draw[-> , thick]  (b2) edge  (a1);
   \draw[-> , thick]  (a4) edge  (b2);

 \draw[-> , thick]  (b3) edge  (a3);
   \draw[-> , thick]  (b3) edge  (c2);
    \draw[-> , thick]  (a5) edge  (b3);

     \draw[-> , thick]  (b4) edge  (dd1);
   \draw[-> , thick]  (b4) edge  (a4);
   \draw[-> , thick]  (b4) edge  (a4);
\draw[-> , thick]  (b4) edge  (a6);

\draw[-> , thick]  (c2) edge  (a4);
\draw[-> , thick]  (c2) edge[bend right=5]  (dd1);
\draw[-> , thick]  (b4) edge  (a4);
\draw[-> , thick]  (a6) edge  (c2);
    
   \draw[-> , thick]  (dd1) edge  (a5); 
   \draw[-> , thick]  (b4) edge  (a4);

\draw[-> , thick]  (dd1) edge  (a6); 

\draw[-> , thick]  (a3) edge  (b1); 

\draw[-> , thick]  (B0) edge  (b0); 
\draw[-> , thick]  (B1) edge  (B0);
\draw[-> , thick]  (B2) edge  (B1); 
\draw[-> , thick]  (B3) edge  (B2);

\draw[-> , thick]  (b0) edge  (A0); 
\draw[-> , thick]  (a0) edge  (B0);
\draw[-> , thick]  (B0) edge  (A1); 
\draw[-> , thick]  (A0) edge  (B1);
\draw[-> , thick]  (B1) edge  (A2); 
\draw[-> , thick]  (A1) edge  (B2);
\draw[-> , thick]  (B2) edge  (A3); 
\draw[-> , thick]  (A2) edge  (B3);

\draw[-> , thick]  (A3) edge  (A2); 
\draw[-> , thick]  (A2) edge  (A1); 
\draw[-> , thick]  (A1) edge  (A0); 
\draw[-> , thick]  (A0) edge  (a0);
    \end{scope}

\end{tikzpicture}
}
\end{center}

The cyclic permutation of $2N+5$ labellings of the quiver above returns to its original form quiver, as shown in Figure \ref{fig:QA2N}. 

\end{proof}

The statement above enables us to define the cluster map which is associated with the exchange matrix \eqref{exchmA2N} (or $Q_{A_{2N}}$ in fig \ref{fig:QA2N}), that is, we define $\psi_{A_{2N}} = \rho_{A_{2N}}^{-1}\mu_{N}\mu_{N-1}\tilde{\mu}\mu_{3N+5}\mu_{N+3}$ and see that
\begin{equation}
    \psi_{A_{2N}}(B_{A_{2N}}) = \rho_{A_{2N}}^{-1}\mu_{N}\mu_{N-1}\tilde{\mu}\mu_{3N+5}\mu_{N+3} (B_{A_{2N}}) = B_{A_{2N}}.
\end{equation}
Following a sequence of the above quiver mutations, we see that  the $n$-th iterate of the cluster map $\psi_{A_{2N}}$, starting from the initial seed
$$(q_{N-2, 0},\dots,q_{1,0},q_{0,0},\tau_{-1},\tau_{0},\tau_{1},\sigma_{0},\sigma_{1},\sigma_{2},\dots, \sigma_{2N},\sigma_{2N+1},p_{0,0},p_{1,0},\dots, p_{N-2,0})$$
generates cluster variables defined by the system of recursion relations below:
\begin{equation}\label{A2nclusterm}
\begin{split}
   \tau_{n+2}\sigma_{n} &= \sigma_{n+2} \tau_{n} + a_{1}p_{0,n} \\
   p_{0,n+1}p_{0,n} &= \sigma_{n+3}\sigma_{n+2}\tau_{n}\tau_{n+1} + p_{1,n} \sigma_{n+1}\tau_{n+2} \\ 
   p_{1,n+1} p_{1,n} &= \sigma_{n+4}\sigma_{n+3}\tau_{n}\tau_{n+1} + p_{2,n}p_{0,n+1}\\
   &\quad \vdots\\
   p_{N-2,n+1}p_{N-2,n} &= \sigma_{n+N+1}\sigma_{ n+N }\tau_{n}\tau_{n+1} + q_{N-2,n}p_{N-3,n+1}\\
   q_{N-2,n+1}q_{N-2,n} &= \sigma_{n+ N+2}\sigma_{n+ N+1 }\tau_{n}\tau_{n+1} + q_{N-3,n}p_{N-2,n+1}\\
   q_{N-3,n+1}q_{N-3,n} &= \sigma_{n+ N+3}\sigma_{ n+N+2 }\tau_{n}\tau_{n+1} + q_{N-4,n}q_{N-2,n+1}\\
   & \quad \vdots\\ 
   q_{0,n+1}q_{0,n} &= \sigma_{n+ 2N}\sigma_{ n+2N-1 }\tau_{n}\tau_{n+1} + \sigma_{n+2N+1} q_{1,n+1}\tau_{n-1}\\
   \sigma_{n+2N+2}\tau_{n-1} &= \sigma_{n+2N}\tau_{n+1} + a_{2N}q_{0,n+1}
   \end{split}
   \end{equation}
We then impose the variable transformations in \eqref{vartransA2N} and obtain the exchange relations
\begin{equation}\label{deformedmA2n}
\begin{split}
   x_{1,n} x_{1,n+1} &= 1 + a_1x_{2,n} \\
   x_{2,n} x_{2,n+1} &= 1 + x_{3,n}x_{1,n+1} \\
  &\vdots \\ 
 x_{2N-1,n} x_{2N-1,n+1}&= 1 + x_{2N,n}x_{2N-2,n+1} \\
 x_{2N,n} x_{2N,n+1} &= 1 + a_N x_{2N-1,n+1} , \\
 \end{split}
\end{equation}
which are induced by the deformed type $A_{2N}$ map $\tilde{\varphi}_{A_{2N}} = \mu_{2N}\mu_{2N-1}\cdots\mu_{1}$.

In conclusion,

\begin{thm} By applying Laurentification, deformed map $\tilde{\varphi}_{A_{2N}}$, derived from the cluster map $\varphi_{A_{2N}}$, can be lifted to the cluster map $\psi_{A_{2N}}$. 
 \qed 
\end{thm}

\section{Tropical dynamics and degree growth for Laurentified deformed map }\label{s:degreeA2N}
\setcounter{equation}{0}

In the previous section, we observed that applying Laurentification retrieves a higher dimensional cluster map with the Laurent property from the type $A_{2N}$ deformed map. However, we did not address the integrability of the higher case of deformed maps. The challenge lies in the increasing complexity of calculating the conditions for the parameters as we apply the procedure to each modified first integral in even cases. Thus, proving the integrability by constructing invariant functions becomes difficult. Instead, we perform an algebraic entropy test as an alternative, providing strong evidence that deformed maps possess integrability. Here we begin by calculating the degree growth (ref. \cite{2013}) of lifted maps, arising from the deformation of type $A_{4}$ cluster map and then extend our analysis to the higher rank case, type $A_{2N}$.

\subsection{Tropical dynamics and degree growth for type $A_{4}$ deformed map} 

Recalling Example~\ref{LaurentA4}, given the initial cluster
$$\hat{\vb{x}} = (q_{0}, \tau_{-1},\tau_{0},\tau_{1}, \sigma_{0},\sigma_{1},\sigma_{2},\sigma_{3},\sigma_{4},\sigma_{5}, p_{0}, a_{1},a_{4}) = (\tilde{x}_{j})_{1\leq j \leq 13},$$
we have a deformed map lifted to a  higher dimensional cluster map $\psi_{A_{4}} = \rho^{-1}\mu_2\mu_1\mu_{11}\mu_5$ via the rational map $\pi$, 
\begin{equation}\label{vartransA41}
    \begin{split}
        x_{1,n} = \frac{\sigma_{n}\tau_{n+1}}{\sigma_{n+1}\tau_{n}} \quad & x_{2,n} = \frac{p_{n}}{\sigma_{n+2}\tau_{n}} \quad  x_{3,n} =\frac{q_{n}}{\sigma_{n+3}\tau_{n}} \quad x_{4,n} =\frac{\sigma_{n+5}\tau_{n-1}}{\sigma_{n+4}\tau_{n}}
    \end{split}
\end{equation}
which is equivalent to the system
\begin{equation}\label{eq:clustermA4}
    \begin{aligned}
 \tau_{n+2}\sigma_{n} & = \sigma_{n+2}\tau_{n} + a_{1}p_{n} \\
    p_{n+1}p_{n} & = \sigma_{n+3}\sigma_{n+2}\tau_{n}\tau_{n+1} + q_{n} \sigma_{n+1}\tau_{n+2} \\ 
    q_{n+1}q_{n} & = \sigma_{n+4}\sigma_{n+3}\tau_{n}\tau_{n+1} + p_{n+1}\sigma_{n+5}\tau_{n-1}\\
    \sigma_{n+6}\tau_{n-1} & = \sigma_{n+4}\tau_{n+1} + a_{1}q_{n+1}        
    \end{aligned}
\end{equation}
Since the lifted map has the Laurent property, the system of equations in \eqref{eq:clustermA4} produces a sequence of tau functions which can be expressed in the form
\begin{equation}\label{tauvarA4}
\tau_{n} = \frac{N^{(1)}_{n}(\hat{\vb{x}})}{\hat{\vb{x}}^{\vb{d}_{n}}}, \quad \sigma_{n} = \frac{N^{(2)}_{n}(\hat{\vb{x}})}{\hat{\vb{x}}^{\vb{e}_{n}}}, \quad  p_{n} = \frac{N^{(3)}_{n}(\hat{\vb{x}})}{\hat{\vb{x}}^{\vb{f}_{n}}}, \quad q_{n} = \frac{N^{(4)}_{n}(\hat{\vb{x}})}{\hat{\vb{x}}^{\vb{g}_{n}}}
\end{equation}
where $N^{(j)}(\hat{\vb{x}}) \in \mathbb{Z}[\hat{\vb{x}}]$ are polynomials in the initial cluster variables $\hat{\vb{x}}$, and $\vb{d}_{n},\vb{e}_{n},\vb{f}_{n},\vb{g}_{n}  \in \mathbb{Z}^{11} $ are sequences of d-vectors specified by the initial data 
\begin{equation}\label{initd}
\qty( \vb{g}_{0} \ \vb{d}_{-1} \ \vb{d}_{0} \ \vb{d}_{1} \ \vb{e}_{0} \ \vb{e}_{1}  \ \vb{e}_{2}  \ \vb{e}_{3}  \ \vb{e}_{4}  \ \vb{e}_{5}  \ \vb{g}_{0}) = - \mathbf{I} . 
\end{equation}
Note that the monomial denominators in \eqref{tauvarA4} only depend on the 11 mutable cluster variables, hence the d-vectors that appear as exponents of $\hat{\vb{x}}$ have 11 components (since the coefficients $a_1,a_4$, which are frozen variables, do not appear in the denominators).  
With the same argument as described in previous sections, we can find the $(\max,+)$ tropical relations of  the d-vectors arising from \eqref{eq:clustermA4}, shown below:
\begin{equation}\label{bilinearA4deg}
\begin{array}{rcl}
\vb{d}_{n+2} + \vb{e}_{n} & = & \max(\vb{e}_{n+2} + \vb{d}_{n}, \vb{f}_{n} ), \\
 \vb{f}_{n+1} + \vb{f}_{n} & = & \max(\vb{e}_{n+2} + \vb{e}_{n+3} +\vb{d}_{n} +\vb{d}_{n+1}, \vb{g}_{n} + \vb{e}_{n+1} + \vb{d}_{n+2} ), \\
  \vb{g}_{n+1} + \vb{g}_{n} & = & \max(\vb{e}_{n+3} + \vb{e}_{n+4} +\vb{d}_{n} +\vb{d}_{n+1}, \vb{f}_{n+1} + \vb{e}_{n+5}+\vb{d}_{n-1}  ), \\
 \vb{e}_{n+6} + \vb{d}_{n-1} & = & \max(\vb{e}_{n+4} + \vb{d}_{n+1}, \vb{g}_{n+1} ), \\
\end{array}  
\end{equation}

Let us consider the tropical analogue of the variables \eqref{vartransA41} given by the rational map $\pi$, 
\begin{equation}\label{tropA4var}
\begin{array}{rcl}
\vb{X}_{1,n} = \vb{e}_{n} + \vb{d}_{n+1} - \vb{e}_{n+1} - \vb{d}_{n},& \quad &\vb{X}_{2,n} = \vb{g}_{n} - \vb{e}_{n+2} - \vb{d}_{n},\\
  \vb{X}_{3,n} = \vb{f}_{n} - \vb{e}_{n+3} - \vb{d}_{n}, & \quad &\vb{X}_{4,n} = \vb{e}_{n+5} +  \vb{d}_{n-1} - \vb{e}_{n+4} - \vb{d}_{n} \\ 
\end{array}
\end{equation}
Then one can show the following.
\begin{lm} Given d-vectors $\vb{d}_{n}, \vb{e}_{n},\vb{f}_{n}, \vb{e}_{n}$ satisfying the system of equations \eqref{bilinearA4deg},  the quantities $\vb{X}_{j,n}$ for $1 \leq j \leq 4$ in \eqref{tropA4var} are induced by a tropical map $\varphi^{trop}_{A_{4}}$, which is specified by
\begin{equation}\label{troporignA4}
\begin{array}{rcl}
\vb{X}_{1,n+1} + \vb{X}_{1,n} & = & \qty[\vb{X}_{2,n}]_{+}, \\
\vb{X}_{2,n+1} + \vb{X}_{2,n} & = & \qty[\vb{X}_{1,n+1} +\vb{X}_{3,n} ]_{+} ,\\
\vb{X}_{3,n+1} + \vb{X}_{3,n} & = & \qty[\vb{X}_{2,n+1} +\vb{X}_{4,n} ]_{+} ,\\
\vb{X}_{4,n+1} + \vb{X}_{4,n} & = & \qty[\vb{X}_{3,n+1} ]_{+}, \\
\end{array}
\end{equation}
where $\qty[X_{j,n}]_{+} = \max(X_{j,n},0)$. Given arbitrary initial data $(X_{1,0},X_{2,0},X_{3,0},X_{4,0})$, the orbit of $\varphi^{trop}_{A_{4}}$ is periodic with period  7. 
 
\end{lm}
\begin{proof}  
Firstly, the relations can be found by rearranging  \eqref{bilinearA4deg} and rewriting the equations in terms of $\vb{X}_{i,n}$ by substituting \eqref{tropA4var}. One can check that the map $\varphi^{trop}$ is the tropical analogue of the original cluster map $\varphi_{A_{4}}$ associated with type $A_{4}$. This implies that each quantity $\vb{X}_{i,n}$ corresponds to the  d-vector of the cluster variable $x_{i,n}$ defined in the cluster algebra of type $A_{4}$. Furthermore, as already noted, the cluster map is periodic with period 7, and hence the orbits of $\varphi^{trop}_{A_4}$ is also periodic with period 7.
\end{proof}
We can utilise the periodicity of terms \eqref{tropA4var} to obtain the following result. 
\begin{thm} 
The  d-vectors $\vb{d}_{n}$, $\vb{e}_{n}$, $\vb{q}_{n}$, $\vb{p}_{n}$, which solve the system of equations \eqref{bilinearA4deg} are elements in the kernel of the linear difference operator 
\begin{equation}
\cL=(\cS^{7} - 1)(\cS^5 - 1)(\cS-1)
\end{equation}
where, as before, $\cS$ is the shift operator that sends 
$n \to n+1$.  For the tau functions \eqref{tauvarA4} generated by this, the leading order of degree growth of their denominators is given by 
\begin{equation}\label{d-vectorsA4}
\begin{split}
&\vb{e}_{n} = \frac{n^2}{35}(2,1,1,1,1,1,1,1,1,1,2)^T + O(n)\\
&\vb{d}_{n} = \frac{n^2}{35}(2,1,1,1,1,1,1,1,1,1,2)^T + O(n)\\
&\vb{q}_{n} = \frac{n^2}{35}(4,2,2,2,2,2,2,2,2,2,4)^T + O(n)\\
&\vb{p}_{n} = \frac{n^2}{35}(4,2,2,2,2,2,2,2,2,2,4)^T + O(n)\\
\end{split}
\end{equation}
\end{thm}

\begin{proof} By \eqref{tropA4var}, we have 
\begin{align*}
&\vb{X}_{1,n} =  (\cS - 1)\vb{d}_{n} - (\cS - 1)\vb{e}_{n} \\
& \vb{X}_{4,n+1} =  - (\cS - 1)\vb{d}_{n} +(\cS^{6} - \cS^{5})\vb{e}_{n} \\
\end{align*}
From the periodicity of the terms, the sum $\vb{X}_{1,n} +  \vb{X}_{4,n+1}$ produces the linear difference equation
\begin{align*}
\cL \vb{e}_{n} &= (\cS^{7} - 1)(\cS^{6} - \cS^{5} - \cS +1 )\vb{e}_{n} \\
& = (\cS^{7} - 1)(\cS^{5} -1 )(\cS-1 )\vb{e}_{n} = 0
\end{align*}
Therefore, given that the $\vb{e}_{n}$ satisfy $\cL \vb{e}_{n} = 0 $, from the formula of $\vb{X}_{1,n}$ above we find 
\begin{align*}
(\cS^{7} - 1)(\cS^{6} - \cS^{5} - \cS +1 )\vb{d}_{n} = 0 . 
\end{align*}
Similarly, we can find the relations for the d-vectors $\vb{g}_{n}$ and $\vb{f}_{n}$ from the relations $\vb{X}_{2,n}$ and $\vb{X}_{3,n}$ in \eqref{tropA4var}  respectively. Then we have
\begin{align*}
&\cL \vb{g}_{n} = \cL\vb{X}_{2,n} + \cL\vb{e}_{n+2} + \cL\vb{d}_{n} = 0 \\ 
&\cL \vb{f}_{n} = \cL \vb{X}_{3,n} + \cL \vb{e}_{n+3} + \cL \vb{d}_{n} = 0 
\end{align*}
Altogether, solving the linear difference equation above gives the expression for d-vectors whose leading order term is $n^2$ with some constant coefficient, $\vb{e}_{n} = \vb{a}n^2 + O(n)$.  We now consider the sequences of d-vectors emerging from the iteration of \eqref{bilinearA4deg}, given by
\begin{align*}
    \begin{array}{rcl}
    \vb{e}_{6} & = & (1,1,0,0,1,0,0,0,0,0,1)^{T} \\
       \vb{e}_{7}& = & (1,1,1,0,1,1,0,0,0,0,1)^{T} \\ 
      \vb{e}_{8} & = &(1,1,1,1,2,1,1,0,0,0,2)^{T}  \\ 
      \vb{e}_{9} & = & (2,1,1,1,2,2,1,1,0,0,3)^{T}  \\ 
      \vb{e}_{10} & = & (3,2,1,1,2,2,2,1,1,0,3)^{T}   \\
      \vb{e}_{11}& = & (4,3,2,1,3,2,2,2,1,1,4)^{T}   \\
       \vb{e}_{12}& = &  (5,3,3,2,4,3,2,2,2,1,5)^{T}   \\
        \vb{e}_{13}& = &  (6,4,3,3,5,4,3,2,2,2,7)^{T}   \\
\end{array}
\end{align*}
Note that from the sequences above,  we can see the components of $\vb{e}_{n}$ taking the form
\begin{equation}\label{e:d-vectorA4}
\vb{e}_{n} = (e^{(1)}_{n},e^{(2)}_{n+2},e^{(2)}_{n+1},e^{(2)}_{n}, e^{(3)}_{n+5},  e^{(3)}_{n+4},  e^{(3)}_{n+3},  e^{(3)}_{n+2},  e^{(3)}_{n+1}, e^{(3)}_{n}, e^{(4)}_{n})
\end{equation}
from which we have the following relation:
\begin{align*}
 (\cS^{7} - 1)(\cS^{5} -1 )\vb{e}_{n} = (4,2,2,2,2,2,2,2,2,2,4)^{T} = 70 \vb{a}
\end{align*}
This enables us to fix the constant coefficient $\vb{a}$, which leads to the d-vector
\begin{align*}
\vb{e}_{n} =\vb{a}n^2 +O(n) =  \frac{n^2}{35}(2,1,1,1,1,1,1,1,1,1,2)^{T} + O(n)
\end{align*}
Similarly to \eqref{e:d-vectorA4}, the sequence of d-vectors $\vb{d}_{n}$ takes the particular form 
\begin{align*}
\vb{d}_{n} = (d^{(1)}_{n},d^{(2)}_{n+2},d^{(2)}_{n+1},d^{(2)}_{n}, d^{(3)}_{n+5},  d^{(3)}_{n+4},  d^{(3)}_{n+3},  d^{(3)}_{n+2},  d^{(3)}_{n+1}, d^{(3)}_{n}, d^{(3)}_{n})
\end{align*}
whose components are given by integer sequences beginning with
\begin{align*}
    &d^{(1)}_{n} : 0, 0,0,0,0,1,1,2,3,3,4,5,7,8,9,\ldots ,\\
     &  d^{(2)}_{n} : 0, 0,-1,0,0,0,1,1,2,2,2,3,4,5,5,\ldots , \\ 
      & d^{(3)}_{n}  :0,0,0,0,0,0,0,1,1,1,1,2,3,3,4,4,5,6,7,\ldots . \\ 
\end{align*}
Then we find leading order quadratic growth of the d-vectors: \eqref{tropA4var}
\begin{align*}
\vb{d}_{n} = \frac{n^2}{35}(2,1,1,1,1,1,1,1,1,1,2)^{T} + O(n)
\end{align*}
Following on from the result above, we can determine the coefficient of leading order terms for the d-vectors $\vb{g}_{n}$ and $\vb{f}_{n}$ from the formul\ae\ of $\vb{X}_{2,n},\vb{X}_{3,n}$ in  \eqref{tropA4var} together with periodicity as
\begin{align*}
&\vb{g}_{n}  \sim \vb{e}_{n+2} + \vb{d}_{n}  \sim 2\vb{a}n^2  \\
&\vb{f}_{n} \sim  \vb{e}_{n+3} + \vb{d}_{n}  \sim 2\vb{a}n^2 \\
 \end{align*}
with the same constant $\vb{a}$, giving us \eqref{d-vectorsA4}. 

\end{proof}

\begin{rem} For the homogenous degree of the left and right hand sides of relation \eqref{eq:clustermA4} to be consistent,  the homogeneous degree for $q_{n}$ and $p_{n}$ must be twice the degree of $\tau_{n}$ and $\sigma_{n}$. Hence the  degrees of $\vb{g}_{n}$ and $\vb{f}_{n}$ should grow twice as fast as $\vb{d}_{n}$ and $\vb{e}_{n}$, which matches the result above. 
\end{rem}

\subsection{Algebraic entropy of the Laurentified cluster map $\psi_{A_{2N}}$} \label{ss:algentpA2n}

Following the process of the previous section, we set the initial d-vectors to be
 \begin{equation}
 \qty(\vb{g}^{(N-2)}_{0},\dots,\vb{g}^{(0)}_{0},  \vb{d}_{-1},\vb{d}_{0}, \vb{d}_{1}, \vb{e}_{0}, \vb{e}_{1}, \dots, \vb{e}_{2N+1}, \vb{f}_{0}^{(0)}, \dots \vb{f}^{(N-2)}_{0} ) = - I
 \end{equation}
 for the initial cluster 
\begin{align*}
  \tilde{\vb{x}} = (\tilde{x}_{j}) = \qty(q_{N-2, 0},\dots,q_{0,0},\tau_{-1},\tau_{0},\tau_{1},\sigma_{0},\sigma_{1},\dots, \sigma_{2N},\sigma_{2N+1},p_{0,0},\dots, p_{N-2,0})
\end{align*}
Recall that we have shown that the map $\psi_{A_{2N}}$ given by \eqref{A2nclusterm} corresponds to the mutations in a cluster algebra. By the Laurent property, we can write the tau functions generated by  $\psi_{A_{2N}}$ in the form  
\begin{equation}
  \tau_{n} = \frac{N^{(1)}_{n}(\tilde{\vb{x}})}{\tilde{\vb{x}}^{\vb{d}_{n}}}, \quad \sigma_{n} =  \frac{N^{(2)}_{n}(\tilde{\vb{x}})}{\tilde{\vb{x}}^{\vb{e}_{n}}} , \quad  p_{j,n} =   \frac{\hat{N}^{(j)}_{n}(\tilde{\vb{x}})}{\tilde{\vb{x}}^{\vb{f}^{(j)}_{n}}}, \quad  q_{j,n}=  \frac{\widetilde{N}^{( j)}_{n}(\tilde{\vb{x}})}{\tilde{\vb{x}}^{\vb{g}^{(j)}_{n}}}, 
\end{equation}
Using the same argument as in the  previous section, we find that the corresponding d-vectors satisfy the tropical version of the system of exchange relations \eqref{A2nclusterm}, given by  
\begin{equation}\label{bilinearA2N}
\begin{array}{rcl}
\vb{d}_{n+2} + \vb{e}_{n} & = & \max(\vb{e}_{n+2} + \vb{d}_{n}, \vb{f}^{(0)}_{n}), \\ 
\vb{f} ^{(0)}_{n+1} + \vb{f}^{(0)}_{n} & = & \max ( \vb{e}_{n+3} + \vb{e}_{n+2} + \vb{d}_{n} + \vb{d}_{n+1}, \vb{f}^{(1)}_{n} + \vb{e}_{n+1} + \vb{d}_{n+2}),\\ 
\vb{f}^{(1)}_{n+1} + \vb{f}^{(1)}_{n} & = & \max ( \vb{e}_{n+4} + \vb{e}_{n+3} + \vb{d}_{n} + \vb{d}_{n+1}, \vb{f}^{(2)}_{n} + \vb{f}^{(0)}_{n+1}  ),\\  
&\vdots & \\ 
\vb{f}^{(N-2)}_{n+1} + \vb{f}^{(N-2)}_{n} & = & \max ( \vb{e}_{n+N+1} + \vb{e}_{n+N} + \vb{d}_{n} + \vb{d}_{n+1}, \vb{g}^{(N-2)}_{n} + \vb{f}^{(N-3)}_{n+1} ), \\ 
\vb{g}^{(N-2)}_{n+1} + \vb{g}^{(N-2)}_{n} & = & \max ( \vb{e}_{n+N+2} + \vb{e}_{n+N+1} + \vb{d}_{n} + \vb{d}_{n+1}, \vb{g}^{(N-3)}_{n} + \vb{f}^{(N-2)}_{n+1} ), \\ 
\vb{g}^{(N-3)}_{n+1} + \vb{g}^{(N-3)}_{n} & = & \max ( \vb{e}_{n+N+3} + \vb{e}_{n+N+2} + \vb{d}_{n} + \vb{d}_{n+1}, \vb{g}^{(N-4)}_{n} + \vb{g}^{(N-2)}_{n+1} ), \\ 
& \vdots & \\ 
\vb{g}^{(0)}_{n+1} + \vb{g}^{(0)}_{n} & = & \max ( \vb{e}_{n+2N} + \vb{e}_{n+2N-1} + \vb{d}_{n} + \vb{d}_{n+1}, \vb{e}_{n+2N+1} + \vb{d}_{n-1} + \vb{g}^{(1)}_{n+1}  ), \\ 
\vb{e}_{n+2N+2} + \vb{d}_{n-1} & = & \max(\vb{e}_{n+2N} + \vb{d}_{n+1} , \vb{g}^{(0)}_{n+1})
\end{array}  
\end{equation}
For the next step, we introduce the quantities $\vb{X}_{i}$ which are analogous to the tropical version of the variable transformation \eqref{vartransA2N}:
\begin{equation}\label{dgvarA2n}
    \begin{array}{rcl}
X_{1,n} & = & \vb{e}_{n} + \vb{d}_{n+1} - \vb{e}_{n+1} - \vb{d}_{n}, \\ 
 X_{2,n} & = & \vb{f}^{(0)}_{n} - \vb{e}_{n+2} - \vb{d}_{n},\\ 
  & \vdots &  \\ 
X_{N,n} & = & \vb{f}^{(N-2)}_{n} - \vb{e}_{n+N} - \vb{d}_{n}, \\ 
X_{N+1,n} & = & \vb{g}^{(N-2)}_{n} - \vb{e}_{n+N+1} - \vb{d}_{n}, \\
& \vdots & \\ 
X_{2N-1,n} & = & \vb{g}^{(0)}_{n} - \vb{e}_{n+2N-1} - \vb{d}_{n},\\
X_{2N,n} & = & \vb{e}_{n+2N+1} + \vb{d}_{n-1} - \vb{e}_{n+2N} - \vb{d}_{n}
\end{array} 
\end{equation}
leading to the following lemma.

\begin{lm} The combination of d-vectors defined by \eqref{dgvarA2n} satisfy the tropical analogue of the deformed type $A_{2N}$ map $\tilde{\varphi}_{A_{2N}}$ \eqref{deformedmA2n}, given by the following system of equations, 
\begin{equation}\label{maxpA2n}
    \begin{array}{rcl}
     X_{1,n+1} + X_{1,n} & = &[X_{2,n}]_{+} \\ 
      X_{2,n+1} + X_{2,n} & = &[X_{1,n+1} + X_{3,n}]_{+} \\ 
      & \vdots & \\ 
      X_{2N-1,n+1} + X_{2N-1,n} & = & [X_{2N-2,n+1} + X_{2N,n}]_{+} \\ 
      X_{2N,n+1} + X_{2N,n} & = &[X_{2N-1,n}]_{+} \\
\end{array}
\end{equation}
which we specify by the tropical map $\varphi^{trop}_{A_{2N}}$. Given arbitrary initial data $(X_{j,0})_{1\leq j \leq 2N}$, the orbit of the $\varphi^{trop}_{A_{2N}}$ is periodic with period  $2N+3$.
\end{lm}
\begin{proof}
The $(\max,+)$ relations \eqref{maxpA2n} can be directly derived from the structure of the quantities  \eqref{dgvarA2n} and the tropical analogue (i.e.\ ultradiscretization) of the exchange relations \eqref{deformedmA2n} which gives \eqref{bilinearA2N}. Notice that \eqref{maxpA2n} is indeed a $(\max,+)$ relation of undeformed cluster map $\varphi_{A_{2N}}$. As we are aware that the map $\varphi_{A_{2N}}$ possesses periodicity with period $2N+3$ due to Zamolodchikov periodicity, hence the components of $\varphi^{trop}_{A_{2N}}$ are periodic with period $2N+3$, i.e.\ $(\varphi^{\text{trop}}_{A_{2N}})^{2N+3}\vb{X}_{n} = \vb{X_{n}}$. 
\end{proof}
By using the periodicity, one can calculate the degree growth of d-vectors. 
\begin{thm} The d-vectors $\vb{d}_{n},\vb{e}_{n},\vb{g}^{(i)}_{n},\vb{f}^{(i)}_{n}$ satisfying the $(\max,+)$ relations \eqref{bilinearA2N} are solutions of following linear difference equation
 \begin{equation}
(\cS^{2N+3} -1)(\cS^{2N+1} - 1)(\cS - 1)\vb{r}_{n} = 0
 \end{equation}
 for $\vb{r}_{n} = \vb{d}_{n},\vb{e}_{n},\vb{g}^{(i)}_{n},\vb{f}^{(i)}_{n}$. For the associated tau functions, the leading order of degree growth of their denominator vectors is given by 
\begin{equation}\label{tropdegA2N}
\begin{split}
\vb{d}_{n} = \frac{1}{(8N^2 + 16N + 6)}\vb{a}  
n^2 + O(n),& \quad  \vb{e}_{n} =  \frac{1}{(8N^2 + 16N + 6)} \vb{a} n^2 + O(n), \\
\vb{g}^{(i)}_{n}=  \frac{1}{(4N^2 + 8N + 3)} \vb{a} n^2 + O(n), & \quad \vb{f}^{(i)}_{n} =  \frac{1}{(4N^2 + 8N + 3)} \vb{a} n^2 + O(n)
\end{split}
\end{equation}
where $\vb{a}= (a_{i})_{1\leq i \leq 4N+3}$ whose entries $a_{i} = 4$ for $ i = \qty{1,\dots,N-1} \cup   \qty{3N+5,\dots,4N+3}$ and  $a_{i} = 2$ for $i= N-1,\dots 3N+4$.
\end{thm}

\begin{proof} 
\begin{equation}\label{tropXA2n}
\begin{array}{rcl}
    X_{1,n} & = & (\cS-1) \vb{d}_{n} - (\cS-1  )\vb{e}_{n} , \\ 
     X_{2N,n+1} & = & -(\cS -1) \vb{d}_{n} + (\cS^{2N+2} - \cS^{2N+1})\vb{e}_{n}
      \end{array}
\end{equation}
Addition of these two terms, followed by applying the periodicity property, i.e.\ $(\cS^{2N+3} - 1)(X_{i,n})_{i=1,2} = 0$, produces the linear difference equation
\begin{equation}\label{drecA2N}
\begin{split}
&(\cS^{2N+3} -1)(\cS^{2N+1} - 1)(\cS - 1)\vb{e}_{n} = 0 \\
 &\implies  (\cS - 1)^{3} (\sum_{i=0}^{2N+2} \cS^{i}) (\sum_{i=0}^{2N} \cS^{i})\vb{e}_{n} = 0
 \end{split}
\end{equation}
It is clear that the corresponding characteristic polynomial has root $\lambda = 1$ with multiplicity 3, and the other (simple) roots are $\lambda = e^{2k\pi \ri/ (2N+3)}, e^{2l\pi\ri / (2N+1)} $ for $1\leq k \leq 2N+2$ and $1\leq l \leq 2N$ (and $\ri=\sqrt{-1}$). Thus for constant $\vb{a}$, the leading order of $\vb{e}_{n}$ takes the form
\begin{equation}\label{etrop}
    \vb{e}_{n} = \vb{a} n^2 + O(n)
\end{equation}
as $n \to \infty$. Since the linear expression holds for $\vb{e}_{n}$, the constant vector $\vb{a}$ satisfies the linear relations, 
\begin{equation}
(\cS^{2N+3} - 1)(\cS^{2N+1} - 1) \vb{e}_{n} = (8N^2 + 16N + 6)\vb{a} 
\end{equation}
By induction, one can deduce that the vector $\vb{a}= (a_{i})_{1\leq i \leq 4N+3}$ has entries $a_{i} = 4$ for $ i = \qty{1,\dots,N-1} \cup   \qty{3N+5,\dots,4N+3}$ and  $a_{i} = 2$ for $i= N-1,\dots 3N+4$. 

By using the first relation in \eqref{tropXA2n} together with periodicity of $\vb{X}_{1,n}$ , the d-vector $\vb{e}_{n}$ satisfies the same linear relation as $\vb{d}_{n}$, namely
\begin{align*}
( \cS^{2N+3} -1)(\cS^{2N+1} - 1)\vb{X}_{1,n} &=\underbrace{( \cS^{2N+3} -1)(\cS^{2N+1} - 1) (\cS-1) \vb{d}_{n}}_{=0}\\
& - (\cS^{2N+3} -1)(\cS^{2N+1} - 1)(\cS-1  )\vb{e}_{n} \\
& =  0
\end{align*}
Then 
\begin{equation}\label{dtrop}
    \vb{d}_{n} = \vb{a} n^2 + O(n)
\end{equation}
We can impose \eqref{etrop} and \eqref{dtrop} onto $\vb{X}_{i,n}$ for $2 \leq i \leq 2N-1$ and obtain the expression of the other d-vectors, whose leading order term is polynomial $n^2$. Then by induction again, we have the required result, \eqref{tropdegA2N}. 
\end{proof}
Since the growths of degree of the variables is quadratic, the entropy is zero. Hence the result leads us to conjecture that deformed type $A_{2N}$ cluster map is Liouville integrable.

\section{Conclusions}
\setcounter{equation}{0}

We have constructed a 2-parameter family of deformations of the periodic cluster maps of type $A_{2N}$, which admit a lift to a cluster algebra of rank $4N+3$, with an additional two frozen variables. It is an interesting question whether the $A_{2N}$ cluster maps admit deformations with more frozen variables. Indeed in the case of $A_6$, we obtained a 3-parameter family of deformed cluster maps and showed that every member is Liouville integrable; but we were only able to find a cluster structure for a 2-parameter subset of this family.   

We showed that the map $\psi_{A_{2N}}$, which is the 
Laurentification of the deformed type $A_{2N}$ cluster map, 
has quadratic degree growth, which implies that its algebraic entropy is zero. This result leads us to conjecture that the 2-parameter deformed $A_{2N}$ map is Liouville integrable.  However, the question of finding a general proof of integrability of the deformed map  
$\varphi_{A_{2N}}$ for $N>3$ remains open. The best way to resolve this question would be to find a general formula for the associated first integrals, for any $N$: ideally, it should be possible to obtain this from Lax pair, perhaps by deforming the construction for the periodic case in Section \ref{s:type-A2N}, based on frieze patterns.

Related to this, together with Kouloukas and Vanhaecke, one of us is currently investigating a way to find a Lax pair for the deformation of type $A_4$ via the singularity analysis of the deformed $A_4$ map, and an associated family of Abelian surfaces. We expect that finding such a Lax pair and combining this with our inductive local expansion approach will suggest a way to identify Lax pairs for the general case of type $A_{2N}$, which will prove Liouville integrability. 

Finally, we would like to find a systematic way to obtain integrable deformations of Zamolodchikov periodicity for all Dynkin types. Recently, one of us has found a construction of such deformations for the cluster maps of type $D_{2N}$ \cite{WookyungD2N}.  

\printbibliography

\end{document}